\UseRawInputEncoding
\documentclass[english]{article}
\usepackage{custom_tex}
\title{Joint Coverage Regions:\\ Simultaneous Confidence and Prediction Sets}

\author[1]{Edgar Dobriban}
\author[1]{Zhanran Lin}

\date{}
\affil[1]{Department of Statistics and Data Science, Wharton School,
University of Pennsylvania}

\begin{document}
\maketitle

\begin{abstract}
\sloppy
  We introduce Joint Coverage Regions, which unify confidence intervals and prediction regions within frequentist statistics.
Specifically, joint coverage regions aim to cover a pair formed by an unknown fixed parameter (such as the mean of a distribution), and an unobserved random datapoint (such as the outcomes associated to a new test datapoint). 
The first corresponds to the confidence component, while the second corresponds to the prediction component.
In particular, our notion unifies classical statistical methods such as the Wald confidence interval with distribution-free prediction methods such as conformal prediction \citep{angelopoulos2023conformal}.
We show how to construct finite-sample valid joint coverage regions when a \emph{conditional pivot} is available; 
under the same conditions 
where exact finite-sample confidence and prediction sets have been previously developed.
We further develop efficient joint coverage region algorithms, including split-data versions to reduce the cost of repeated computation.
We illustrate the use of joint coverage regions in statistical problems such as constructing efficient prediction sets when the parameter space is structured.
\end{abstract}
\fussy

\setcounter{tocdepth}{2}
\tableofcontents

\section{Introduction}\label{sec:introduction}

Confidence intervals and prediction sets are two fundamental methods in frequentist statistics, 
covering fixed parameters and random future observables, respectively, with a given probability. 
Finite-sample valid confidence intervals are often  constructed via inverting pivotal quantities \citep[e.g.,][etc]{cox1979theoretical, lehmann1998theory}, functions of parameters and observables whose distribution is known.
Finite-sample valid prediction sets (also known as tolerance regions) have also been widely studied \citep[e.g.,][etc]{wilks1941determination,Wald1943, guttman1970statistical}, 
with renewed recent interest due to their applicability to modern machine learning via conformal prediction \citep{vovk2022algorithmic,lei2013distribution}.
Such prediction sets often rely on conditional pivots; for instance, for exchangeable scalar datapoints whose distribution is unchanged under all permutations, 
any ordering is equally likely given the set of their values.

While it has been noted that confidence intervals and prediction sets are of a similar nature \citep[e.g.,][p. 482]{guille2024conformal,shao2003mathematical}, 
they are nonetheless currently treated as two distinct concepts, both in statistical research and in education.
However, due to the similarities in their definitions and the assumptions---existence of conditional pivots---under which they exist, it is natural to ask if one can unify these notions. 
Our work aims to achieve this unification, 
by developing the new notion of \emph{Joint Coverage Regions} (JCRs).

\begin{figure}
\begin{minipage}[ht]{0.5\linewidth}
\centering
    \includegraphics[height = 6cm, width =6.3cm]{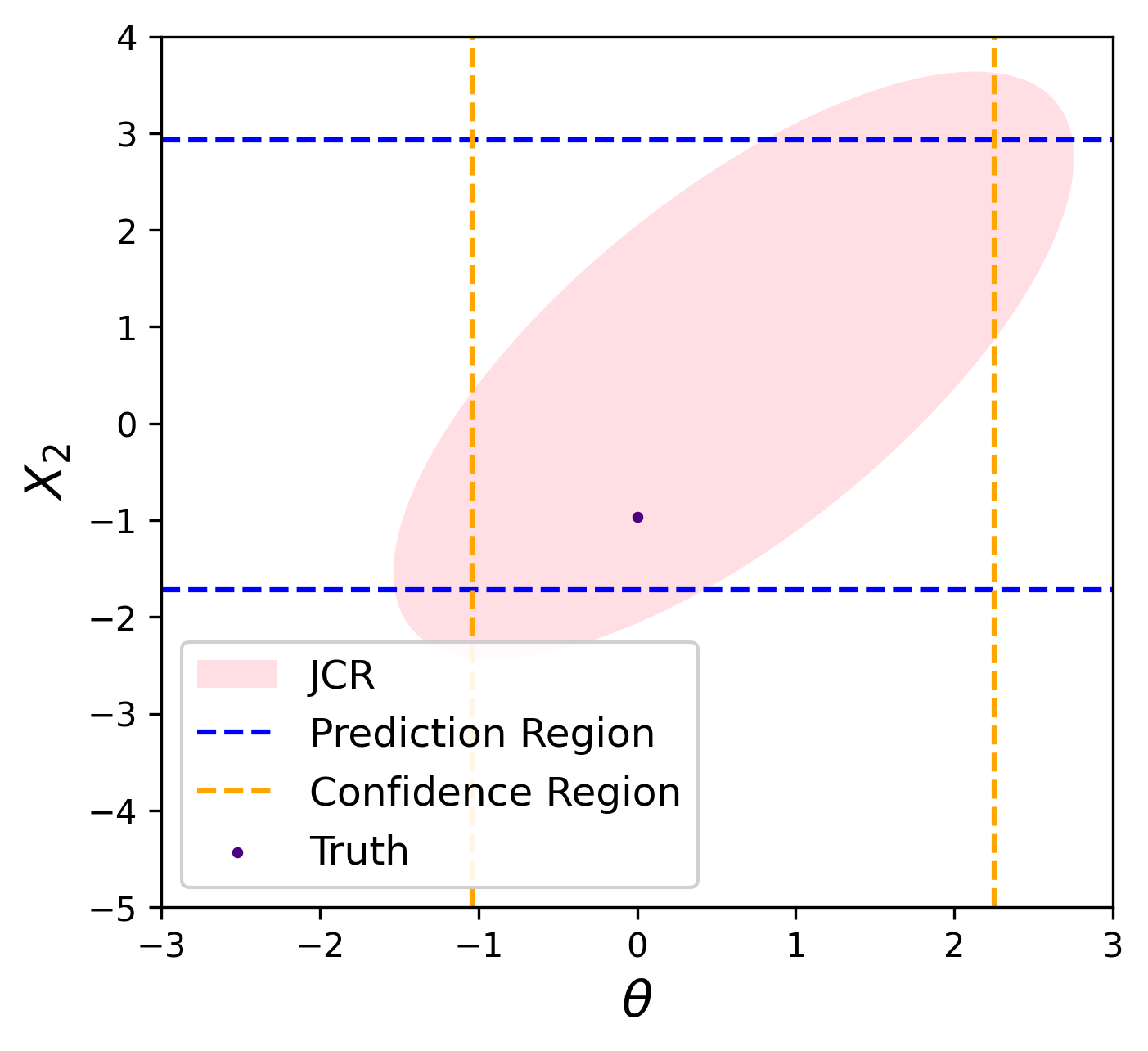}
    \end{minipage}
\begin{minipage}[ht]{0.5\linewidth}
   \centering
    \includegraphics[height = 5.5cm, width =6.8cm]{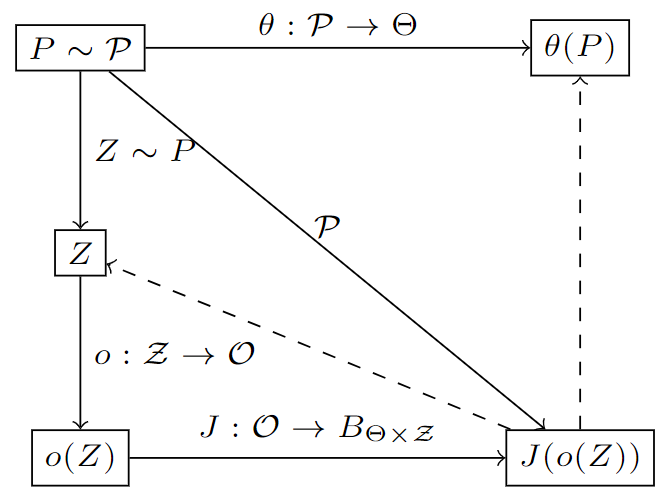}
    \end{minipage}
\caption{Left:  A visualization of the JCR $\{(\theta,X_2):(X_1-\theta)^2+(X_2-\theta)^2 \leq \chi_{1-\alpha}^2(2)\}$ under the model $X_1,X_2 \sim \N(\theta,1)$ with observation $o(X_1,X_2)=x_1$. We show a single trial with $\theta = 0$, $\alpha = 0.1$ and $x_1=0.606$. For contrast, we also plot a confidence interval $x_1 \pm q_{1-\alpha/2}$ for $\theta$ and a prediction region $x_1 \pm \sqrt{2}q_{1-\alpha/2}$ for $X_2$. The purple point labeled ``Truth" shows the true realization $\theta=0$ and $x_2 = -0.962$ in this trial.
Right: A visual representation of our observation model.}
\label{fig:coordinate}
\end{figure}

Joint coverage regions aim to simultaneously cover a pair consisting of an unknown fixed parameter and an unobserved random datapoint.
Formally, consider a class of distribution $\mP$ with a parameter $\theta: \mP \rightarrow \Theta$.
Suppose that the full data $Z\sim P$ is sampled from $P$, but we only observe part of the data, given by $o(Z)$. 
For instance, this can mean that we observe the first $n$ out of $n+1$ datapoints.
We aim to construct a JCR $J$ such that for any distribution $P \in \mP$,
given the observations $o(Z)$ it returns a region covering the pair $\left(\theta(P), Z\right)$ with probability at least  $1-\alpha$:
$$\bP_{Z \sim P}\biggl(\left(\theta(P), Z\right) \in J\left( o(Z)\right)\biggr) \ge 1 - \alpha.$$

Figure \ref{fig:coordinate} (left) shows a JCR
for the model where $X_1,X_2 \sim \N(\theta,1)$ independently, but we only observe $X_1$, and want to cover $(\theta,X_2)$.
The JCR corresponds to any region in $(\theta,x_2)$-space; while a confidence interval for $\theta$ can be viewed as a horizontal strip (and similarly, a prediction interval for $X_2$ can be viewed as a vertical strip).

Generally, when we observe the full data so that $o(Z)=Z$, the first component of the JCR becomes a classical confidence region. 
On the other hand, when are not interested in a parameter (for instance setting $\theta(P)=0$), then the second component of the JCR becomes a classical prediction region for the unobserved full data $Z$ based on the observed data $o(Z)$. 
This can be further simplified in examples, for instance for predicting outcomes $Y_{n+1}$ having observed feature-outcome pairs $(X_i,Y_i)$, $i=1,\ldots, n$ and new features $X_{n+1}$.
In this sense, JCRs unify classical confidence and prediction regions.



In this work, we establish the foundations of JCRs in frequentist settings (Section \ref{JCR-framework}), including their connections to traditional confidence and prediction regions.
We construct JCRs when there is a conditional pivot (Section \ref{sec:conditional-pivot}), i.e., a quantity whose conditional distribution---given some function of the data---is known. 
This is the same condition under which confidence and prediction sets with exact validity have been separately constructed. 
In particular, this holds when there is a function that is invariant in distribution under the action of a group (Section \ref{sec:conditional-invariance}), including permutation-based invariance as for exchangeable data.
As a specific case, we also consider unconditional pivots.

We also introduce efficient algorithms to construct JCRs when there is a separate calibration dataset, and we wish to construct JCRs for several test datapoints (Section \ref{sec:split-JCR}, \ref{sec:split-JCR-invariance}); inspired by split or inductive conformal prediction \citep{papadopoulos2002inductive}.
We introduce the notion of \emph{adequate sets} (Section \ref{sec:adequate-set}, \ref{sec:adequate-set-invariance}), which can significantly improve computational efficiency.

We conduct simulations and empirical studies to illustrate JCRs (Sections \ref{sec:simulation} and \ref{emp}).
We further illustrate how JCRs can be used in two statistical problems
(Section \ref{sec:application}). 
We show how to use JCRs to construct prediction regions when the parameter space is bounded,
by projecting JCRs into their prediction component, which can sometimes be a shorter interval than existing approaches (Section \ref{sec:toy_jcr_pred}).
We also show how JCRs can be used to control the miscoverage when drawing inferences on multiple parameters and future observables  (Section \ref{sec:multiple-task}),
while being more accurate than a more straightforward approach of taking intersections of classical confidence and prediction regions.
Code to reproduce our experiments is available at \url{https://github.com/chris-zhanran-lin/JCR}.

We next outline some notations and conventions which will be used throughout the paper.

{\bf Notations and conventions.}
For a positive integer $m$, we write $[m]:=\{1,2,\ldots,m\}$.
Given numbers $v_1,\ldots, v_n \in \R$ and $\alpha\in[0,1]$,
let $v_{(1)} \leq \ldots \leq v_{(n)}$ 
denote their order statistics.
Let $q_{\alpha}(v_1,\ldots, v_n) = v_{(\lfloor n\alpha \rfloor)}$ 
denote 
the $\alpha$-th quantile of their empirical distribution. For a probability distribution $\mathcal{F}$ on the real line, $q_{\alpha}(\mathcal{F})$ denotes its $\alpha$-th quantile.
For $c\in (0,1)$, 
$q_{c} \in \R$ is the $c$-quantile of the standard normal distribution.
For a probability space $X$ and $a\in X$, let $\delta_a$ denote the point mass at $a$; in other words, the distribution that places all mass at the value $a$.
For two random variables $X,Y$, $X=_dY$ denotes that they have the same distribution.
For two sets $A,B$, a function $f:A\to B$, and a set $S \subset B$, we denote by $f^{-1}(S) = \{a\in A: f(a)\in S\}$ the preimage of $S$ under $f$. 
When $S = \{s\}$ is a singleton, we abbreviate $f^{-1}(\{s\}):=f^{-1}(s)$.
Similarly, for a set $S \subset A$, we denote by $f(S) = \{b\in S: \exists\, a\in A: b = f(a)\}$ the image of $S$ under $f$.
For a finite set $S$, we denote by $|S|$ its cardinality.
For a probability measure $P$ on a measure space $(A,\mathcal{A})$, and a map $f: A\to B$ to another measure space $(B,\mathcal{B})$, we denote by $f(P)$ the probability measure of the random variable $f(Z)$, where $Z\sim P$.
For a positive integer $m$, we denote by $1_m = (1,1,\ldots,1)^\top \in \R^m$ the $m$-dimensional all-ones vector.
For a set $A$, we denote by $I_A$ the identity operator on $A$, defined by $I_A(a)=a$ for all $a\in A$.
For two vectors $a=(a_1,a_2,\ldots,a_n), b = (b_1,b_2,\ldots,b_n)$, denote $a \odot b = (a_1b_1,a_2b_2,\ldots,a_nb_n)$. 
We may abbreviate a sequence as $a_{1:n}=(a_1,a_2,\ldots,a_n)$.
Denote by $I(A)$ the indicator function taking $I(A)=1$ when event $A$ happens and $I(A)=0$ otherwise. 
Denote by $\textnormal{sgn}(x)$ the sign function, where $\textnormal{sgn}(x) = 1$ for $x>0$, $\textnormal{sgn}(x) = -1 $ for $x<0$ and $\textnormal{sgn}(x) =0 $ for $x = 0$.
All functions considered in this paper will be measurable with respect to appropriate sigma-algebras, which will sometimes be implicit from the context.
All sigma-algebras will be assumed to include the singletons over the sets where they are defined.

\subsection{Related Works}\label{sec:related-work}

Confidence intervals, introduced by  \cite{neyman1937outline}, are a core concept in statistics.
There has been an abundance of research focused on constructing them in a variety of settings
\citep[e.g.,][etc]{vsidak1967rectangular,efron1986bootstrap,diciccio1996bootstrap,boldin1997sign,csaji2012non,wasserman2020universal}.
In particular, 
finite-sample confidence intervals are usually constructed via pivotal quantities \citep[e.g.,][etc]{lehmann1998theory,cox1979theoretical}, while functions of data and the parameter whose mean have a known bound can also be used \citep[e.g.,][etc]{wasserman2020universal,xu2022post}.

Prediction sets have a rich statistical history dating back to \protect\cite{wilks1941determination}, \protect\cite{Wald1943}, \protect\cite{scheffe1945non}, and \protect\cite{tukey1947non,tukey1948nonparametric}.
There is an large body of work on constructing prediction sets with coverage guarantees under various assumptions
\protect\citep[see, e.g.,][]{bates2021distribution,Chernozhukov2018,dunn2018distribution,Lei2014,lei2013distribution,lei2015conformal,lei2018distribution,Park2020,park2021pac,Sadinle2019,kaur2022idecode,qiu2022distribution,li2022pac,sesia2022conformal}. 
Among these, one of the best-known methods is conformal prediction (CP) \protect\citep[see, e.g.,][]{saunders1999transduction,vovk1999machine,
papadopoulos2002inductive,vovk2022algorithmic, Chernozhukov2018,dunn2018distribution,Lei2014,lei2013distribution,lei2018distribution}.

Beyond basic confidence intervals and prediction sets,
constructing regions that jointly cover multiple parameters---or alternatively, multiple future variables---has been well studied.
Simultaneous confidence regions, which jointly cover several functions $f_i(\theta)$, $i \in [m]$ of a parameter $\theta$, have been developed using pivots in e.g., \cite{scheffe1953method,scheffe1999analysis}.
On the other hand,
\cite{wolf2015bootstrap} 
construct joint prediction regions for multiple future observables using the bootstrap. 
However, to our knowledge, previous works do not \emph{jointly} consider the confidence and prediction components.
A notable exception is in Bayesian statistics, where parameters and observations are both random variables; and hence both are covered via prediction regions.
Nonetheless, in frequentist statistics there is a fundamental difference between fixed parameters and random observables. 

Pivotal quantities---or, pivots---are  functions of the data and the parameter whose distribution is known; 
this was given a central role in important but mostly unpublished work by G. A. Barnard \citep[][p. 29]{cox2006principles}. 
Finite sample coverage usually relies on the existence of pivots 
\cite[e.g.,][etc]{fraser1966structural,fraser1968structure,fraser1971events,cox1979theoretical,brenner1983models,barnard1995pivotal,fraser1996some}
or "sub-pivots" with bounded moments \citep{wasserman2020universal}.
Conditional pivots have been used, at least implicitly, in areas such as conformal prediction \protect\citep[e.g.,][]{vovk1999machine,vovk2022algorithmic,Lei2014,lei2013distribution,lei2018distribution,romano2019malice,romano2019conformalized,xu2021conformal}.

Our work on group invariance is related to a large literature on using such properties for statistical inference, both for testing and confidence regions \citep[e.g.,][etc]{eden1933validity,fisher1935design,
lehmann1949theory,hoeffding1952large,dwass1957modified,hemerik2018exact,freedman1983nonstochastic,
david2008beginnings,berry2014chronicle,
hemerik2020permutation,dobriban2022consistency}
 For more general discussions of invariance in statistics see \cite{eaton1989group,wijsman1990invariant,giri1996group}. 

Conditional invariance can be useful in
a variety of methodologies for conditional independence testing under the model-X assumption, such as knockoff approaches
\protect\citep[e.g.,][]{candes2018panning,huang2020relaxing},
conditional randomization testing (CRT) \protect\citep[e.g.,][]{candes2018panning,katsevich2020power,liu2022fast},
and conditional permutation tests 
\citep{berrett2020conditional}. 
Going beyond using conditional pivots,  \cite{huang2020relaxing} consider conditional knockoffs,
which require knowing the parametric distribution only up to a parametric model. 

There are also various 
works focusing on improving computational efficiency, such as  
split---or inductive---conformal prediction \citep{papadopoulos2002inductive}; 
and other approaches 
\citep{vovk2022algorithmic,lei2019fast,cherubin2021exact}.
\cite{liu2022fast} develop distilled conditional randomization testing (d-CRT), which computes the main part of the test statistic only once, 
while the remaining part only requires negligible computation. 
Relatedly, we propose adequate sets, which contain information that can be re-used for multiple test datapoints.

\section{Joint Coverage Regions}\label{JCR-framework}

We now introduce our setting.
For some measurable space $\mZ$,
let $Z\in \mZ$ denote data generated from a distribution $P$, where $P$ belongs to a class $\mP$ of probability distributions over $\mZ$. 
Let the observed part of $z$ be $o(z)$, 
taking values in a measurable space $\mathcal{O}$. 
We refer to $o:\mZ\to\mathcal{O}$ as the \emph{observation function}.
We consider the functional $\theta: \mP \rightarrow \Theta$, for some parameter space $\Theta$, determining a parameter $\theta(P)=\theta_P$ of the distribution $P \in \mP$ that we are interested in. 
Without loss of generality, we can assume that the image $\theta(\mP)$ of $\mP$ under $\theta$ is $\Theta$.

Now we discuss some technical conditions and definitions.
We assume that there is a sigma-algebra $B_\mZ$ over $\mZ$, and all $P\in\mP$ are probability distributions defined over $B_\mZ$.
Further, there is a sigma-algebra $B_\mathcal{O}$ over $\mathcal{O}$, and $o$ is measurable with respect to ($B_\mZ,B_\mathcal{O}$).
We also assume that there are sigma-algebras $B_\mP$, $B_\Theta$ over $\mP$, $\Theta$, and $\theta$ is measurable with respect to them.
Further, we consider the product sigma-algebra $B_{\Theta \times \mZ}$ over $\Theta \times \mZ$.
We define the projection operators $\Pi_{ \Theta}: \Theta \times \mZ\to  \Theta$, $\Pi_{ \Theta}(\theta,z)=\theta$, and
$\Pi_{\mZ}: \Theta \times \mZ\to \mZ$, $\Pi_{\mZ}(\theta,z)=z$.
We define their extensions to $B_{\Theta \times \mZ}$ in the obvious way.
For notational convenience, we define the \emph{section operator} $\Phi_{\Theta}: B_{\Theta \times \mZ} \times \mZ \to \Theta$, where $\Phi_{\Theta}(J,z) = \cup_{\theta \in \Theta} \{\theta:\, (\theta,z) \in J \}$ for all $z \in \mZ$ and $J \in B_{\Theta \times \mZ}$. 
This takes the $\Theta$-slice of the set $J\subset \Theta \times \mZ$ given $z\in \mZ$.
We can write $\Phi_{\Theta}(J,z) = \Pi_\Theta(J\cap (\Theta\times\{z\}) )$.
Similarly, define the section operator $\Phi_{\mZ}: B_{\Theta \times \mZ} \times \Theta \to \mZ$ as $\Phi_{\Theta}(J,\theta) = \cup_{z \in \mZ} \{z:\, (\theta,z) \in J \}$ for all $\theta \in \Theta$ and $J \in B_{\Theta \times \mZ}$.

\subsection{Basic Definitions}

Given a desired coverage rate $1 - \alpha \in (0,1)$, 
and having observed $o(z)$,
we aim to construct a \emph{joint coverage region} $J:\mathcal{O} \rightarrow  B_{\Theta \times \mZ}$ for the parameter $\theta_P$ and unobserved data $Z$ that has the following property:

\begin{definition}[Joint Coverage Region]
     We say that $J:\mathcal{O} \rightarrow  B_{\Theta \times \mZ}$ is a $1-\alpha$-joint coverage region (JCR) for $(\theta,Z)$ based on $o(Z)$ if for all $P \in \mP$ we have
\begin{equation}\label{eq:J}
    \bP_{Z \sim P}\biggl(\left(\theta_P, Z\right) \in J\left( o(Z)\right)\biggr) \ge 1 - \alpha.
\end{equation}
\end{definition}


A visualization of our observation model is in Figure \ref{fig:coordinate} (right). 
Thus, 
given observed data $o(z)$, 
a JCR outputs a subset of the space $\Theta \times \mZ$.
This subset is required to cover the parameter $\theta_P$ of interest \emph{and} the unobserved data $Z$ \emph{simultaneously}.
In this sense, a JCR acts \emph{both} as a confidence region, covering the fixed parameter $\theta_P$ with its \emph{confidence component} $\Phi_{\Theta}(J,z)$ for $z \in \mZ$ (where $\Phi_{\Theta}$ is the section operator defined above); 
\emph{and} as a prediction region, covering the random data $Z$ with its \emph{prediction component} $\Phi_{\mZ}(J,\theta)$ for $\theta \in \Theta$.

Of course, having observed $o(z)$, it is only of interest to predict the unobserved part of the data. 
This can be seamlessly included in the above definition.
Given a map $u:\mZ\to\mU$ representing a component of the data that we wish to predict, 
we can transform $(\theta_P,z)\to (\theta_P,u(z))$
and construct a prediction region for $(\theta_P,u(Z))$.
This can be given, for any $o\in \mO$, 
by 
the image 
$\tilde J(o) = (I_\Theta, u) (J(o))$
of $J(o)$ under $(I_\Theta, u)$, 
where $I$ denotes the identity map. 
If $z$ can be decomposed into observed and unobserved parts as $z = (o(z), u(z))$,
then this
reduce the prediction region into one for the unobserved part of $z$. 
Later in Section \ref{sec:c}, we will often say that such JCRs are in a \emph{reduced form}.

To aid our understanding of joint coverage regions, in Section \ref{conn} we will study their connections to classical confidence and prediction regions. 

\section{Constructing JCRs}\label{sec:JCR-construction}
\label{sec:c}

\subsection{Using Pivots}
\label{sec:pivot}
In this section, we outline an approach to construct JCRs based on pivots and conditional pivots. 
For simplicity, we start with pivots, and turn to conditional pivots in Section \ref{sec:conditional-pivot}.
Thus, consider some measurable space $(\mL,B_{\mL})$, and 
let $L:\Theta\times\mZ \to \mL$, be a pivot, in the sense that when $Z\sim P$ for $P\in \mP$,
the distribution $Q$ of $L(\theta(P),Z)$ is known and does not depend on $P$. 
Let $S\subset \mL$ be a measurable set such that $Q(S)\ge 1-\alpha$. 
Then, we can 
 construct a  $1-\alpha$-JCR for $(\theta,Z)$ via
\beq\label{pivjcr}
J(o^*)= \left\{(\theta,z) 
\in\Theta\times\mZ: 
o(z) = o^*,\,
L(\theta,z) \in S \right\}.
\eeq
The validity of this construction is stated below and 
is a direct consequence of Theorem \ref{thm:conditional-JCR-piv} for conditional pivots.
\begin{proposition}\label{prop:pivot}
Suppose the pivot $L$ has distribution $Q$, and $S$ is a measurable set such that $Q(S)\ge 1-\alpha$. 
Then
equation \eqref{pivjcr}
returns a $1-\alpha$-joint coverage region.
\end{proposition}
In Section \ref{conn} we discuss the connection between this construction and classical pivotal confidence regions.
For pivots to lead to informative JCRs, we need $L$ to be more expressive; for instance the constant $L(\theta,Z)=0$ is a pivot, but does not lead to informative regions.
In general, if there are different pivots, the weaker the conditions under which they are pivotal, the more generally the associated JCRs are valid. We will illustrate this later in examples.

Informative pivots are known to exist under a variety of conditions, see
e.g.,
\cite{fraser1966structural,fraser1968structure,fraser1971events,brenner1983models,barnard1995pivotal,fraser1996some},
Sections 7.1.1 and 7.1.4 of \cite{shao2003mathematical},
Section 2.6 of \cite{cox2006principles},
and Section \ref{pive} for a review.
Since standard confidence regions with 
exact finite sample coverage usually
require the existence of pivots,
our methods are typically applicable whenever standard confidence regions can be constructed.

For instance, pivots exist for any parametric statistical model with independent continuously distributed scalar observations (Proposition 7.1 of \cite{shao2003mathematical}).
Another example is injective data generating models, which are often referred to as structural or structured models \citep{fraser1966structural,fraser1968structure,fraser1971events,brenner1983models,fraser1996some}. 
A key example are group invariance models or structural models \citep{fraser1968structure}, 
with classical examples including location-scale families and data with sign-symmetric or spherically distributed noise. 
These are broad enough to include practically important settings such as 
linear mixed effects models, see Section \ref{pive} for details.
See Section \ref{cons} for a discussion of discreteness considerations for constructing pivotal JCRs, including discreteness and using asymptotic pivots.
For clarity, we will usually illustrate JCRs in linear models through this paper.

\begin{example}[Linear regression]\label{example:lr-pivot}
Consider the standard linear regression model $Y_0 = x_0^\top \theta + \ep$ with 
the covariates (features, inputs) $x_0$ belonging to some space $\mX$.
We view $x_0$ as fixed
and study standard normal
noise $\ep \sim \N(0,1)$. 
Suppose $\theta$ belongs to some parameter space $\Theta$.
We denote $z=(x_0,y_0)$ and---for illustration---start with one datapoint, moving to multiple datapoints below.
Our observation consists of the features, i.e.,  $o(z) = x_0$, and we wish to predict the outcome $Y_0$.
Moreover, we wish to make inferences about the parameter $\theta$. 
This calls for constructing a JCR for $(\theta,Y_0)$. 

Since $Y_0$ and $\theta$ are related linearly in this statistical model, we aim for JCRs that capture this relation.
We can form the pivot 
$L:\Theta\times\mX\to \R$ given by
$ L(\theta,z) = y_0 - x_0^\top\theta \sim Q := \N(0,1)$ to derive an 
$1-\alpha$-JCR as
\begin{align*}
    J(x_0)=\{ (\theta,z) \in \Theta\times \mX\times \R: o(z)=x_0,\,  q_{\alpha/2} < y_0 - x_0^\top \theta < q_{1 - \alpha/2}  \}.
\end{align*}
Since $x_0$ is observed, we can simplify this into a prediction region for $y_0$, writing 
\begin{align}\label{eq:toy}
    \tilde J(x_0)=\{ (\theta,y_0)\in \Theta\times  \R: q_{\alpha/2} < y_0 - x_0^\top \theta < q_{1 - \alpha/2}  \}.
\end{align}
A  visualization for the one-dimensional case is shown in Figure \ref{fig:toy}, in which we consider the parameter space $\Theta = \R$, the feature space $\mX = \R$, and suppose that the feature value for which we wish to predict the outcome is $x_0 = 1$.
The diagonal band shape captures the linear relation between $y_0$ and $\theta$, as desired.
\begin{figure}[ht]
    \centering
    \includegraphics[height = 5.5cm, width =5.7cm]{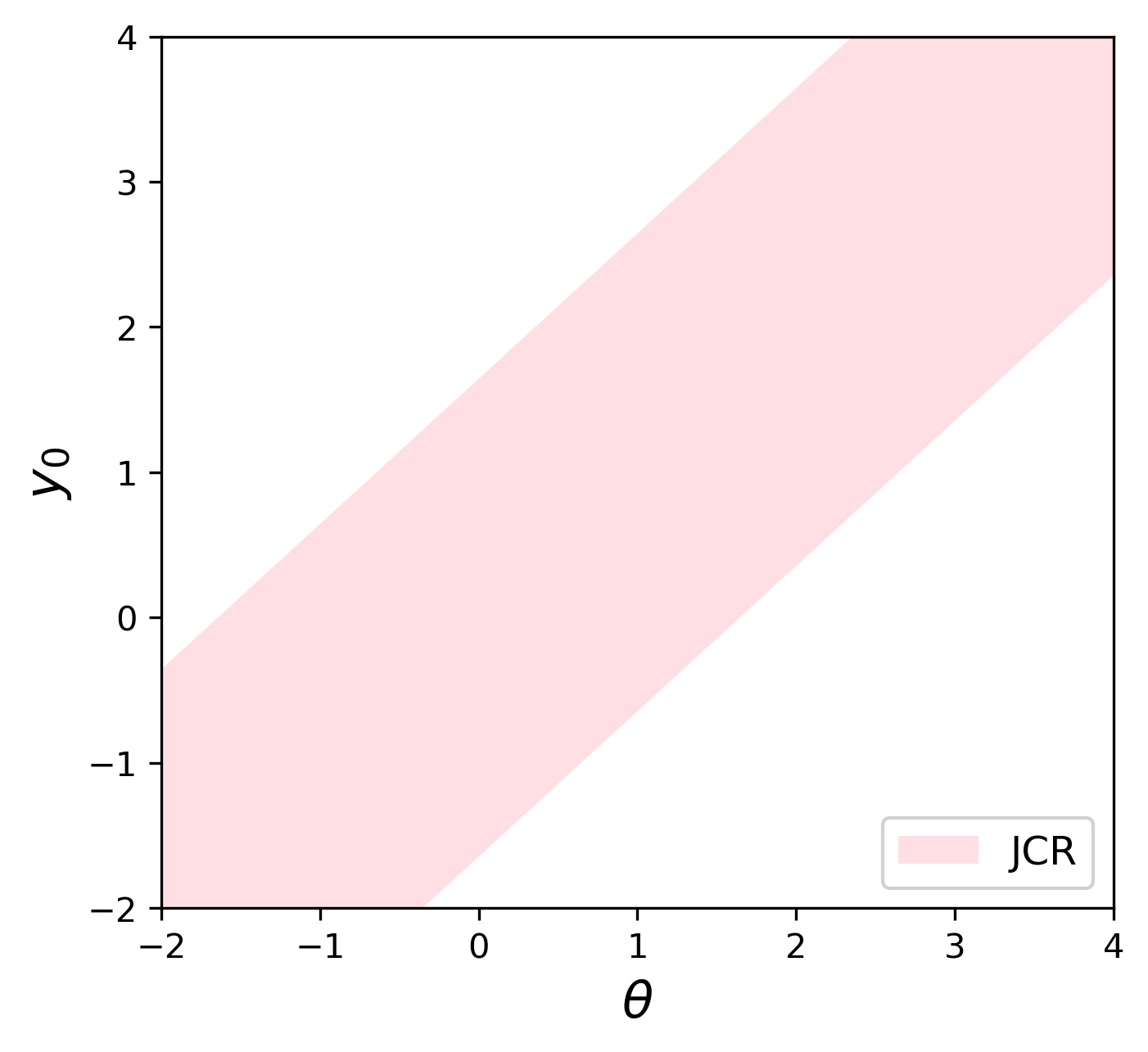}
    \caption{A visualization of the JCR for linear regression defined in \eqref{eq:toy}, where we take $\Theta = \mX = \R$ and $x_0 = 1$.}
    \label{fig:toy}
\end{figure}

Next, 
for a sample size $n\ge 1$,
let $(x_i,y_i)_{i\in [n]}$ be the observed datapoints, where $x_i\in \R^{p}$, $p\ge 1$ and $y_i\in\R$, following the standard linear model 
$Y_i = x_i^\top\theta + \ep_i$. 
Denote 
the $n\times p$ matrix
$X = (x_1^\top,\ldots,x_n^\top)^\top$,
and the $n\times 1$ vectors
$Y = (y_1,\ldots,y_n)^\top$
and $\ep = (\ep_1,\ldots,\ep_n)^\top$.
Consider also a test datapoint $\Yte = \xte^\top \theta + \epte$, where only $\xte$ is observed. 
Let $z = (X^+,Y^+)$ be the full data, where we define
the $(n+1)\times p$ matrix
$X^+ = (X^\top,\xte^\top)^\top$,
and the $(n+1)\times 1$ vector
$Y^+=(Y^\top,\Yte)^\top$.
Thus the complete data includes both $\xte,\Yte$, but the observed data is only
$o(z)=(X^+,Y)$.
We consider $X^+$ fixed and assume that $n\ge p$ and that $X$ has full rank.

For i.i.d.~normal noise  $\ep_i \sim \N(0,\sigma^2)$ and $\ep_{\mathrm{te}} \sim \N(0, \sigma^2)$ with some unknown variance $\sigma^2$, 
we can 
use the pivot ${(\yte - \xte^\top \theta)^2}/{S^2} \sim F_{1,n-p}$, where $S^2 = \sum_{i=1}^n (y_i - x_i^\top \hat{\theta})^2 / (n-p)$ and $\hat{\theta} = (X^\top X)^{-1} X^\top Y$ is the ordinary least squares estimator. 
Hence, we obtain a $1-\alpha$ JCR in reduced form
\begin{align}\label{eq:pivot-F-1}
    \left\{ (\theta,\yte):   |\yte - \xte^\top \theta| < \sqrt{F_{1,n-p}^{1 - \alpha}}S \right\}.
    \end{align}
For each $\theta$, this JCR is a fixed-width interval for $\yte$.
We now consider JCRs for a one-dimensional parameter
$\gamma = c^\top \theta \in \R$, for some $c \in \R^{ p \times 1}$ and $\Yte$.
Suppose that there exists $w \in \R^{n \times 1}$ such that $w^\top X^{+} \theta = c^\top \theta$.
This is guaranteed to hold if $c$ 
belongs to the row span of $X^+$; 
and holds in particular if $(n+1) \ge p$, and
$X^+$ has full rank.
In this case, we can take $w = X^{+,\dagger}c$, where $M^{\dagger}$ denotes the pseudoinverse of the matrix $M$. 
Then, we have the pivot
\begin{align}\label{eq:gau_high_d}
\frac{w^\top Y^{+} - \gamma}{S\|w\|_2} \sim t_{n-p}.
\end{align}

Using this, we can construct a JCR for $(\gamma,\Yte)$.
Since $\yte$ does not appear in $S$, this leads to JCRs with a fixed prediction component width. 
\end{example}

\begin{example}[Non-linear regression]\label{example:nonlinear-regression}
Consider
a non-linear regression model 
where
$Y_i = f(x_i) + \ep_i$, 
with $x_i\in \R^{p}$, $p\ge 1$, and $Y_i\in\R$, for $i\in [n]$ and (with a slight abuse of notation) for $i= \mathrm{te}$.
Suppose that the unknown function $f$ belongs to some function class $\mathcal{F}$. 
We use the same notations as in Example \ref{example:lr-pivot}, and
suppose that $\ep\sim Q$ for a known distribution $Q$.
Let $f(X^+)$ be defined by applying $f$ to each row of $X$.
Then, for all $f\in \mF$,
$Y^+ - f(X^+) \sim Q$.
Hence $Y^+ - f(X^+)$ is a pivot, and we can construct a $1-\alpha$ JCR in reduced form,
$$J(\xte;X,Y) = \{(\yte,f) \in (\R,\mathcal{F}):
Y^+ - f(X^+) \in S
\}$$
for any measurable set $S\in \R^d$ such that $S$ has probability at least $1-\alpha$ under $Q$.

\end{example}

Constructing JCRs with the pivotal approach may require solving a number of potentially challenging computational problems. 
In particular, 
to compute \eqref{pivjcr}, 
we need to search over $\Theta$ and over the level sets of $o$, 
which may require discretization and/or solving potentially challenging non-linear equations.
In some cases, one may be able to find the required sets analytically; in other cases, one may need to compute them numerically.
In this work, we will study examples where computation can be done efficiently.

\subsection{Conditional Pivots}\label{sec:conditional-pivot}

To construct JCRs when informative pivots are not known, 
we next study conditional pivots.
Suppose we have a map $V :\Theta\times\mZ\to\mV$, for some measurable space $\mV$
with a sigma-algebra $B_\mV$.
Then, $L$ is a conditional pivot given $V$, if it has a known distribution $Q_v$
on $(\mL,B_\mL)$,
conditionally on $V(\theta_P,Z)=v$, 
for $P_V$-almost every $v\in\mV$, where $P_V$ is the distribution of $V(\theta_P,Z)$, $Z\sim P$.
The following example underlies the popular conformal prediction methodology.

\begin{example}[Exchangeability of a finite sequence]\label{ex}
Suppose that $Z = (Z_1,\ldots,Z_n)$ has exchangeable entries, in the sense that for any permutation $\pi$ of $[n]$, $Z =_d (Z_{\pi_1}, \ldots, Z_{\pi_n})$. 
Suppose moreover that all entries of $Z$ are distinct almost surely. Then, conditional on the set of entries of $Z$, $Z$ is uniforml over all possible permutations of those entries. Hence, $L(\Theta,Z)= Z$ is a conditional pivot, conditionally on the set $V = \{Z_1, \ldots, Z_n\}$, with a distribution $Q_v$ uniform over all permutations of the entries of $v$. 
\end{example}

Let $S: \mV \to B_{\mL}$ be an assignment of measurable sets such that for $P_V$-a.e. $v$,  $Q_v(S(v))\ge 1-\alpha$.
Then, we can 
 construct a  $1-\alpha$-JCR for $(\theta,Z)$ via
\begin{equation}\label{jcp}
    J(o^*)= \left\{(\theta,z) 
\in\Theta\times\mZ: 
o(z) = o^*,\,
L(\theta,z) \in S(V(\theta,z)) \right\}.
\end{equation}
Its validity is summarized in the following result.
\begin{theorem}\label{thm:conditional-JCR-piv}
Suppose that $L$ is a conditional pivot, having a known distribution $Q_v$ conditionally on $V(\theta_P,Z)=v$; 
for $P_V$-almost every $v\in\mV$. 
Then for any assignment of measurable sets $S: \mV \to B_{\mL}$ with $\{\rho = (\theta,z):\, L(\rho) \in S(V(\rho))\} \in B_{\Theta\times \mZ}$, if for $P_V$-a.e. $v$,  $Q_v(S(v))\ge 1-\alpha$, 
equation \eqref{jcp} returns a $1-\alpha$-joint coverage region. 
\end{theorem}
The proof is given in Section \ref{proof:conditional-JCR-piv} in the Appendix.

We now describe a class of probability distributions where conditional pivots arise, as a generalization of structural or structured models \citep{fraser1966structural,fraser1968structure,fraser1971events}.
\begin{proposition}[Generalized Structural Model, GSM]\label{prop:gsm}
    Suppose that
for some 
measurable map
$\psi:E\times \mV\to \mL$, and
some random variable $\ep$ with a fixed distribution $Q$ over some measurable space $E$,
we have
$L(\theta_P,Z) = \psi(\ep, V(\theta_P,Z))$  for all $P\in \mP$. Then, for $P_V$-a.e. $v$, conditional on $V(\theta_P,Z)=v$, $L(\theta_P,Z)$ has the distribution of $\psi(\ep,v)$; and thus is a conditional pivot.
\end{proposition}
See Section \ref{proof:gsm}
 in the Appendix for the proof.
Next, we will outline several examples of GSMs.
For instance, we can consider a heteroskedastic regression model as an extension of \eqref{eq:toy}, where
$L:\Theta\times\mZ\to \R$ is given by
$ L(\theta,z) = y_0 - x_0^\top\theta \sim Q_{x_0} := \N(0,x_0^2)$, which depends on the input $x_0$.
This satisfies $L(\theta,z)=\psi(\ep,V(\theta,z))$ where $V(\theta,z)=x_0$ and $\psi(\ep,x_0) \sim x_0\cdot \ep$ with $\ep \sim \N(0,1)$.
Thus, given the value of $x_0$, $L(\theta,z)$ has the distribution of $\psi(\ep,x_0)$, and thus is a conditional pivot.
By using the sets 
$S(x_0) = (q_{\alpha/2}|x_0|, q_{1 - \alpha/2}|x_0|)$,
\eqref{jcp} leads to the
$1-\alpha$-JCR
\begin{align*}
    J(x_0)=\left\{ (\theta,z) \in \Theta\times \mX\times \R: o(z)=x_0,\,  q_{\alpha/2}|x_0| < y_0 - x_0^\top \theta < q_{1 - \alpha/2}|x_0|\right  \}.
\end{align*}
This can be also viewed a JCR based on the unconditional pivot $(y_0 - x_0^\top \theta)/x_0$. 

As a second example, 
for independent, possibly non-identically distributed, continuous random variables $Z_1,\ldots,Z_n$ symmetrically distributed around $\theta$, with $Z=(Z_1,\ldots,Z_n)$, $L(\theta,Z)=(Z_1-\theta,\ldots,Z_n-\theta)$ is a conditional pivot, 
conditional on $V(\theta,z) = (|z_1-\theta|,\ldots,|z_n-\theta|)$.
Specifically, 
for $P_V$-a.e. $v$,
conditional on  $V(\theta,z)=v$,
we have $\psi \sim U$, where $U$ denotes the  discrete uniform distribution on the unit cube. 
In general, $L$ is not an unconditional pivot; only the element-wise signs of the entries of $L$ are \citep{boldin1997sign}.
However, using only the signs can lose information;  showing that conditional pivots are useful here.


As already mentioned, the conditional pivotal approach is also a direct generalization of the popular conformal prediction method \citep{gammerman1998learning,vovk1999machine,vovk2022algorithmic},
see Section \ref{sec:conditional-invariance} for discussion. 

{\bf Test statistic-based approach.}
Further, 
pivot-based JCRs can be constructed using a test statistic $m:\mL\to \R$, 
defining $S$ to be the set of 
datapoints where $m$ is sufficiently large, depending on the value of $V$, i.e., 
\begin{align}\label{def:conditional-JCR}
    J(o^*)= \biggl\{(\theta,z) 
\in\Theta\times\mZ: 
o(z) = o^*,\,
m(L(\theta,z)) \ge q_{\alpha}\bigl(m\left(Q_{V\left(\theta,z\right)}\right)\bigr) \biggr\}.
\end{align}
Recall here that $m\left(Q_{V\left(\theta,z\right)}\right)$ is the pushforward of $Q_{V\left(\theta,z\right)}$ under $m$.
This JCR provides coverage at the desired level. 
\begin{theorem}\label{thm:conditional-JCR}
Suppose that a conditional pivot $L(\theta,Z)$ has a known distribution $Q_v$ conditionally on $V(\theta_P,Z)=v$; 
for $P_V$-almost every $v\in\mV$, with $P_V$ the distribution of $V(\theta_P,Z)$, $Z\sim P$. Then for a test statistic $m: \mL\to \R$, 
\eqref{def:conditional-JCR} returns a $1-\alpha$-joint coverage region. 
\end{theorem}

The proof is given in Section \ref{proof:conditional-JCR} in the Appendix. 

\begin{example}[Conformal prediction]
In the setting of Example \ref{ex},
consider a pure prediction region, i.e., $\Theta = \emptyset$, 
and suppose $o(z) = (z_1, \ldots, z_{n-1})$. 
Then \eqref{def:conditional-JCR} becomes, in reduced form,
$$
J(z_1, \ldots, z_{n-1})= \biggl\{z_n \in\mZ:
m(z) \ge q_{\alpha}\bigl(\{m\left(\pi\cdot z\right),\, \pi \in S_n\}\bigr) \biggr\},
$$
where $S_n$ denotes the set of $n$-permutations and for $\pi\in S_n$, $\pi \cdot z = (z_{\pi_1}, \ldots, z_{\pi_n})$.
This recovers the most basic form of conformal prediction with a conformity score $m$
\citep{saunders1999transduction,vovk2005algorithmic}.
\end{example}

In general, note that a non-strict inequality is needed in \eqref{def:conditional-JCR};
as, for instance, if $m(L) = c, \forall L \in \mL$ for some constant $c \in \R$, then 
  using a strict inequality would fail to ensure \eqref{def:conditional-JCR} has $1-\alpha$ coverage.
  While the use of a non-strict inequality may result in slight conservativeness, 
  it is possible to modify the approach to make it exact.

\subsubsection{Randomization}
If finding the quantiles of the distribution $m(Q_v)$
is computationally or analytically hard, 
we can define the following randomized JCR, which reduces the problem to finding the quantiles of a discrete uniform distribution.
For some $K\ge 1$, and for a given value of 
$V(\theta,z)$, we sample 
$M=(M_1,\ldots, M_K)$,
such that each $M_i$, for $i\in [K]$
is i.i.d.~following the distribution $m(Q_{V(\theta,z)})$.
We write
$M\sim m(Q_{V(\theta,z)})^K$,
and
for any $c\in(0,1)$, denote the 
$c$-quantile of the multiset of the entries of $M$
by $q_c(\{M\})$.
Then, we let  
$\alpha' = \lfloor (K + 1)\alpha\rfloor/K$, and define
the \emph{randomized JCR}
\begin{align}\label{eq:conditional-pivot-randomization}
    J_{N}(o^*)= \left\{(\theta,z) 
\in\Theta\times\mZ: 
o(z) = o^*,\,
m(L(\theta,z)) \ge q_{\alpha'}(\{M\}),\,
M\sim m(Q_{V(\theta,z)})^K\right\}.
\end{align}
Thus, for each value of $\theta,z$ such that $o(z)=o^*$, we draw the random vector 
$M\sim m(Q_{V(\theta,z)})^K$, 
and include $(\theta,z)$ in the JCR
if the test statistic $m(L(\theta,z))$ of the pivot
$L$ is larger than $q_{\alpha'}(\{M\})$.
We show that returns a valid JCR.
\begin{theorem}\label{thm:conditional-pivot-randomization}
The set $ J_{N}$ from  \eqref{eq:conditional-pivot-randomization} is a $1-\alpha$-joint coverage region, in the sense that
    \begin{align*}
        \bP_{Z;M \sim m(Q_{V(\theta,Z)})^K}\left((\theta_P,Z) \in J_{N}(o(Z))\right) \ge 1 - \alpha.
    \end{align*}
\end{theorem}
 The proof is in Section \ref{proof:conditional-pivot-randomization} in the Appendix.
In general, constructing \eqref{eq:conditional-pivot-randomization}
requires drawing new random variables $M_i, i\in [K]$ for each $z$, and can thus be computationally expensive.
However,
we will show that  under group invariance, 
randomization with conditional pivots can become computationally efficient (Section \ref{sec:conditional-invariance}). 

For the special case of an unconditional pivot $L$ with distribution $Q$,
randomization amounts to sampling 
$K\ge1 $ i.i.d.~random variables $M_1,\ldots, M_K\sim m(Q)$, and computing, with  $\alpha' = \lfloor (K + 1)\alpha\rfloor/K$
\begin{align*}
    J_M(o^*)= \left\{(\theta,z) 
\in\Theta\times\mZ: 
o(z) = o^*,\,
m(L(\theta,z)) \ge q_{\alpha'}(\{M_1,\ldots, M_K\}) \right\}.
\end{align*}
Intriguingly, randomization 
can be viewed as considering a conditional pivot under an extended probability space including $m(L(\theta,Z))$ and $M_1,\ldots, M_K$.
Since these variables are iid, 
we can consider the  conditional pivot that is uniform over all permutations of datapoints, conditioning on the set of observations (Section \ref{sec:conditional-invariance}). 


\subsection{Split JCRs}\label{sec:split-JCR}
In this subsection, we describe a split, or split-data,  construction of JCRs---inspired by inductive or split conformal prediction \citep{papadopoulos2002inductive}---which can be more computationally efficient. 
We assume that the data 
$z$ can be partitioned into
\emph{calibration data} $\zc$
and \emph{test data} $\zte$, 
as 
$z = (\zc, \zte) \in \mZc \times \mZte =:\mZ$. 
We are concerned with the setting where there are multiple test datapoints $\zte$, 
and we want to construct prediction regions for them based on a given calibration dataset $\zc$.
We assume that the test datapoints are conditionally i.i.d.~given $\zc$; and consider one generic test datapoint $\zte$ for notational clarity.
We may also have  \emph{training data} $\zt$ used to construct, say, a predictor or a test statistic, which we can later use in the JCR. 
We view $\zt$ as fixed, and usually do not mention it further.

We assume the observed data is $o(z) = (\zc,o_0(\zte))$, for some observation function $o_0: \mZte \to \mO$. 
Similarly to Section \ref{sec:conditional-pivot},
assume that  there is a $V(\theta,\zc,\zte)$-conditional pivot $L(\theta,\zc,\zte)$ taking values in $\mL$. 
The  $1-\alpha$-JCR from \eqref{jcp} becomes, in reduced form
\begin{align}\label{redu-ind}
    \tilde J(o^*)= \left\{(\theta,\zte) 
\in\Theta\times\mZte: 
(\zc,o_0(\zte)) = o^*,\,
L(\theta,\zc,\zte) \in S(V(\theta,\zc,\zte)) \right\},
\end{align}
Having observed $\zc$, 
we only need to compute the parts of 
$\tilde J$ that depend on each new test datapoint $\zte \in \mZte$.
As we will see, this can reduce the computational burden.

We can also take a test statistic $m:\mL\to \R$, possibly depending on $L$ and $\zt$, and construct
\begin{align}\label{eq:split-test-stat}
  \tilde J(o^*)= \left\{(\theta,\zte) 
\in\Theta\times\mZte: 
(\zc,o_0(\zte)) = o^*,\,
m(L(\theta,\zc,\zte)) \ge q_{\alpha}(m(Q_{V(\theta,\zc,\zte)})) \right\}.
\end{align}
This has $1-\alpha$ coverage due to Theorem \ref{thm:conditional-JCR}.
The advantage of split JCRs is that we
can fix $\zc,L$ and
$m$ for all future $\Zte \in \mZte$.
As in  split conformal prediction \citep{papadopoulos2002inductive}, 
we can learn 
a useful test statistic based on $\zt$, and then calibrate it over
the calibration data $\zc$.
If $L$ is an unconditional pivot,
this reduces the computational cost to computing a quantile of $m(Q)$, which can be used for all future test datapoints $\zte$.
Further,
as a consequence of Section \ref{sec:conditional-pivot}, randomization also applies here, with the same guarantee.

\subsection{Adequate Sets for Supervised Problems}\label{sec:adequate-set}

Here we propose \emph{adequate sets}, an approach to reduce computational cost in certain supervised problems. 
We assume that the test datapoint has the form $\zte = (\xte,\yte)$, where $\xte$ are the observed features and $\yte$ is the unobserved prediction target, and $o(z) = (\zc,\xte)$. 
We aim to improve the computational efficiency of constructing a JCR to be used 
for a sequence of new test inputs $\xte^1, \ldots, \xte^n \in \mathcal{X}_{\mathrm{te}}$.
As in the previous section,
we assume that the test datapoints are conditionally i.i.d.~given $\zc$; and consider one generic test datapoint $\zte$ for notational clarity.

Suppose for the moment that in the split JCR from \eqref{redu-ind}, 
we do not consider $\xte$ as observed, i.e., 
we take $o_0$ to map to the empty set.
Then, we find that $o^* = \zc$, and so the JCR equals
$$\tilde J(\zc)= \left\{(\theta, \xte,\yte) 
\in\Theta\times\mathcal{X}_{\mathrm{te}} \times \mathcal{Y}_{\mathrm{te}}: 
L(\theta,\zc,\xte,\yte) \in S(V(\theta,\zc,\xte,\yte)) \right\}.$$
We can take the $\Theta \times \mathcal{Y}_{\mathrm{te}}$-section of this set
over $\xte \in\mathcal{X}_{\mathrm{te}}$ to obtain the JCR $J(\zc,\xte)$: 
$$J(\zc,\xte)
= \left\{(\theta,\yte) 
\in\Theta \times \mathcal{Y}_{\mathrm{te}}: 
L(\theta,\zc,\xte,\yte) \in S(V(\theta,\zc,\xte,\yte)) \right\}.$$
Now, we assume that the condition defining 
$J$ 
can be simplified via 
an \emph{adequate map} 
$A:
\Theta \times \mZte \to \mathcal{A},$ 
for some measurable space $\mA$,
and
an \emph{adequate set} $W:\Theta\times \Zc \to B_{\mA}$, 
in the sense that
$L(\theta,\zc,\xte,\yte) \in S(V(\theta,\zc,\xte,\yte))$ is equivalent to
$ A(\theta,\xte,\yte) \in W(\theta,\zc)
$, for all $\theta,\zc,\xte,\yte$ under consideration.
The intuition is that the adequate set and map \emph{decouple the
functional dependence} between $\zc$ and $(\xte,\yte)$ in the condition. 
This is reasonable if the condition is determined entirely 
based on $\zc$, and then the same condition is applied to all future $\xte,\yte$;  we will give examples where this happens.
In this case, the JCR simplifies to
\begin{align}\label{eq:W-J}
    J(\zc,\xte) = \{(\theta,\yte)\in\Theta \times \mathcal{Y}_{\mathrm{te}}:
    A(\theta,\xte,\yte) \in W(\theta,\zc)\}.
\end{align}
This JCR inherits the coverage properties of general JCRs.
\begin{theorem}\label{thm:adequate-set}
   The construction in \eqref{eq:W-J} returns a $1-\alpha$-joint coverage region.
\end{theorem}
The proof is in Section \ref{proof:adequate-set}. 
As an illustration, 
in \cref{example:lr-pivot},
the region \eqref{eq:pivot-F-1} 
can be written via 
an adequate 
map taking values 
$A(\theta,\xte,\yte) = |\yte - \xte^\top \theta|$
and an adequate 
set taking values, 
for $\zc  = (X,Y)$, and $S = S(\zc)$,
\begin{align}
    W(\theta,\zc) = \left\{ (\theta,\xte,\yte):   |\yte - \xte^\top \theta| < \sqrt{F_{1,n-p}^{1 - \alpha}}S \right\}.
    \end{align}
We will give other examples under group invariance in Section \ref{sec:adequate-set-invariance}.

{\bf Test statistic-based approach and randomization.}
Given a test statistic $m$, we can similarly transform \eqref{eq:split-test-stat}
into
\begin{align*}
    J(\zc,\xte)= \left\{(\theta, \yte) 
\in\Theta\times \mathcal{Y}_{\mathrm{te}}: 
m(L(\theta,\zc,\xte,\yte)) \ge q_{\alpha}(m(Q_{V(\theta,\zc,\xte,\yte)}) \right\}.
\end{align*}
This will simplify as above if
$L(\theta,\zc,\xte,\yte)$ 
does not depend on $\zc$ and
its distribution 
given $V(\theta,\zc,\xte,\yte)$ does not depend on $\xte,\yte$. 
In that case, randomization can also be implemented efficiently.
We will show examples under group invariance
in Section \ref{sec:adequate-set-invariance}.

\section{Group Invariance}\label{sec:conditional-invariance}

As an important example of conditional pivots, 
we consider problems with \emph{group invariance}.
Specifically, suppose that there is an \emph{invariant function} $I:\Theta \times \mZ \to \mI$
with a sigma-algebra $B_{\mI}$, 
for some space $\mI$,
and a group $\mG$
acting on $\mI$ via an action $\phi:\mG\times \mI \to \mI$, abbreviated as $\phi(g,I) = gI$.
Suppose that the function $I$ is invariant in distribution under the group $\mG$, namely
\beq\label{if}
gI(\theta_P,Z) =_d I(\theta_P,Z),
\eeq
for all $g\in \mG$ and all $P \in \mP$, when $Z\sim P$.
This assumption covers many examples, as shown below.
If this condition holds for $\mG$, it also holds for all subgroups; so all conclusions below apply to those as well.
Given $z,P$, denote the orbit of $I(\theta_P,z)$ 
under the action of $\mG$
by
$O_I(\theta_P,z) = \{gI(\theta_P,z): g\in \mG \}.$

We assume that $\mG$ is a compact group with a left Haar measure $U$; normalized to be a probability distribution, see e.g., \cite{eaton1989group,wijsman1990invariant}. 
Let $U_{O_I(\theta_P,z)}$ be the uniform measure on $O_I(\theta_P,z)$, induced by the distribution of $GI(\theta_P,z)$ when $G\sim U$.
Then, by taking $V = O_I$, we find that $I$ is a conditional pivot, with the uniform distribution $U_{O_I(\theta_P,z)}$ over $O_I(\theta_P,z)$. See Section \ref{sec:uniform} in the Appendix for details.

We now propose an algorithm for JCR construction, following the general approach for conditional pivots from Section \ref{sec:conditional-pivot}. 
We assume that the orbits $O_I(\theta_P,z)$ 
belong to a space $\mO'$, which is endowed with a sigma-algebra $B_{\mO'}$; alternatively, we may also choose a representative from each orbit in an appropriate measurable way.
We take $L = I$, $Q_{o'} = U_{o'}$,
and let $S: \mO' \to B_{\mI}$ be an assignment of measurable sets such that for $P_O$-a.e. $o' \in \mO'$,  $U_{o'}(S(o'))\ge 1-\alpha$;
where $P_O$ is the distribution of $O_I(\theta_P,Z)$.
Then, we can 
 construct a  $1-\alpha$-JCR for $(\theta,Z)$ via Algorithm \ref{a1}.
 
We then consider a test statistic-based approach.
We consider some $m: \mI\to \R$, mapping $I(\theta,z)$ to $\R$, and possibly depending on $z$. 
Allowing a dependence on $z$ leads to additional flexibility, as we will see from examples.
We then compute
the probability measure $m(U_{O_I(\theta_P,z)})$, 
the distribution of $m(GI(\theta_P,z))$ when $G\sim U$.
As a special case of \eqref{def:conditional-JCR}, we can construct a JCR by 
\begin{equation*}\label{J1}
J(o^*)=\left\{(\theta, z): m(I(\theta, z)) \ge q_{\alpha'}
\bigg(  m(U_{O_I(\theta_P,z)}) \bigg),
\
o(z) = o^*\right\},
\end{equation*}
where $\alpha' = \alpha$ if $\mG$ is infinite, 
and $\alpha' = \lfloor |\mG|\alpha\rfloor/|\mG|$ if $\mG$ is finite. 
There is a slight distinction between the quantiles,
as for a finite group, $I$ has a positive probability  mass function over $U(O_I)$.

\begin{algorithm}
\caption{JCR based on group invariance}
\KwIn{Observation $o^*$; invariant function $I:\Theta \times \mZ \to \mI$; group $\mG$.
}
\KwOut{Joint Coverage Region for $Z$ and $\theta_P$}
Let $J(o^*) = \emptyset$, $\mZ' = \{z \in \mZ, o(Z)=o^*\}$.\\
Choose measurable sets $S: \mO' \to B_{\mI}$ such that for $P_O$-a.e.  $o' \in \mO'$,  $U_{o'}(S(o'))\ge 1-\alpha$.\\
\For{$\theta \in \Theta$ \textnormal{and} $z \in \mZ'$}{
    Compute  the orbit $O_I$ of $I = I(\theta,z)$ under $\mG$\;\lIf{$I(\theta, z) \in S(O_I)$}{Add $(\theta, z)$ to $J(o^*)$}
}
\KwResult{Region $J(o^*)$}
\label{a1}
\end{algorithm}
This JCR construction inherits the coverage guarantee of general JCRs, as shown below.
\begin{theorem}\label{full-group-JCR}
Suppose that for an invariant function $I:\Theta \times \mZ \to \mI$ and for a group $\mG$, $gI(\theta_P,Z) =_d I(\theta_P,Z)$ holds for all $P \in \mP$ and all $g\in \mG$ when $Z\sim P$. 
Then
Algorithm \ref{a1} returns a $1-\alpha$-JCR.
\end{theorem}

The proof is in Section \ref{proof:full-group-JCR} in the Appendix. 
If the group is large, randomization may reduce the computational cost, while ensuring coverage.
\begin{theorem}\label{sample-JCR}
In the setting of \Cref{full-group-JCR}, sample $G_{1:K}$ i.i.d.~from $U$. Define 
$$J_{\gk}(o^*) = \left\{(\theta,z): m(I(\theta,z)) \ge q_{\alpha''}\big( m(g_1I(\theta,z)),\ldots, m(g_KI(\theta,z))\big), o(z) = o^*\right\},$$ where $\alpha'' = \lfloor \alpha(K+1)\rfloor/K$.
Then $J_{\Gk}$ is a $1-\alpha$-joint coverage region:
    \begin{align*}
        \bP_{Z,\Gk}\big( (\theta_P,Z) \in J_{\Gk}(o(Z))\big) \ge 1 - \alpha.
    \end{align*}
\end{theorem}
The proof is given in Section \ref{proof:JCR-infinite-group} in the Appendix. 
Randomization 
can be viewed as considering 
the conditional pivot $L(\theta,Z,G_{1:K}) = \big(m(I),m(G_1I),\ldots,m(G_KI)\big)$,
whose entries are exchangeable, 
and thus its distribution is conditionally uniform under the permutation group, given the multiset of its entries. 
Taking the section over $\zc,\xte,g_{1:K}$,
we obtain the JCR from \Cref{sample-JCR}.

\subsection{Split Version}\label{sec:split-JCR-invariance}
The split JCR construction from Section \ref{sec:split-JCR}
can lead to computational savings under group invariance.
Suppose
that $z=(\zc,\zte) \in \mZ$, 
and 
consider a test statistic $m: \mI \to \R$;
this can potentially depend on  training data, a dependence we do not display since $\zt$ is suppressed.
As a special case of the methods from Section \ref{sec:split-JCR}, we propose the JCR in reduced form
\begin{align}\label{eq:split-group-inv}
    \tilde J(o^*) = \left\{(\theta,\zte) \in \Theta \times \mZte: (\zc,o_0(\zte)) = o^*,\, m(I(\theta,\zc,\zte)) \ge q_{\alpha'}
\big(m(U_{O_I(\theta,\zc,\zte)})\big) \right\},
\end{align}
where $\alpha' = \alpha$ if $\mG$ is infinite, 
and $\alpha' = \lfloor |\mG|\alpha\rfloor/|\mG|$ if $\mG$ is finite (see Algorithm \ref{a2}).
For each new input $o_0(\zte^i)$, 
we only need to search over $\mZte^* = \{\zte \in \mZte: o_0(\zte) = o_0(\zte^i)\}$ to construct the JCR, as $\zc$ is fixed. 
This can improve efficiency for a series of test datapoints $\zte$.
We show that this algorithm returns a valid JCR.

\begin{algorithm}
\caption{Split JCR under group invariance}
\KwIn{Observations $\zc, o_0(\zte^1), \ldots, o_0(\zte^n)$, invariant function $I:\Theta \times \mZ \to \mI$, group of transforms $\mG$. 
}
\KwOut{JCRs for $(\theta,\zte^i)$, for $i\in [n]$.}
{Choose a test statistic $m: \mI \to \R$.}\\
\For{\textnormal{each input} $o_0(\zte^i)$}{
Let $o^*=(\zc,o_0(\zte^i))$. Set $J(o^*) = \emptyset,\, \mZte^* = \{\zte \in \mZte: o_0(\zte) = o_0(\zte^i)\}$\;
\For{$\theta \in \Theta$ \textnormal{and} $\zte \in \mZte^*$}{
     Compute the probability measure $m(U_{O_I(\theta,\zc,\zte)})$, i.e., the distribution of $m(GI(\theta,\zc,\zte))$ when $G\sim U$\;\lIf{$m(I(\theta,\zc,\zte)) \ge q_{\alpha'}
\big(m(U_{O_I(\theta,\zc,\zte)})\big)$}{Add $(\theta,\zte)$ to $\tilde J(o^*)$}
}
Return region $\tilde J(o^*)$.
}
\label{a2}
\end{algorithm}

\begin{proposition}[Split JCR]\label{general JCR-split}
Suppose that for an invariant function $I:\Theta \times \mZc \times \mZte \to \mI$ and for a group $\mG$, $gI(\theta_P,\Zc,\Zte) =_d I(\theta_P,\Zc,\Zte)$ holds for all $P \in \mP$ and all $g\in \mG$ when $(\Zc,\Zte) \sim P$. Then 
Algorithm \ref{a2} is a $1-\alpha$-joint coverage region. 

\end{proposition}
The proof follows from 
the results for conditional pivots in Section \ref{sec:split-JCR}.

{\bf Randomization.} As before, we can replace computing the quantile over the entire orbit by that over only an i.i.d. sample $G_{1},\ldots,G_{K}$ from $U$. With $g_{1:K}=(g_{1:K})$, we obtain a JCR
\begin{align}\label{eq:sample-split}
    \tilde J_{\gk}(o^*) = \left\{(\theta,\zte):(\zc,o_0(\zte))=o^*,\, m(I(\theta,\zc,\zte)) \ge q_{\alpha''}\big( m(g_iI(\theta,\zc,\zte)),i\in[K]\big)\right\}
\end{align}
similar to the one from Theorem \ref{sample-JCR}.
We have argued in Section \ref{sec:split-JCR} 
that the main computational cost in split JCRs
is computing the appropriate quantiles. 
Here we illustrate that 
this becomes simpler
under group invariance.
Given $\zc$, 
and sampling elements $G_{1:K} \sim  U$,
we can compute the required quantile
for any new $\zte$ based on
$m(G_1V(\theta,z))$, $\ldots,m(G_KV(\theta,z)))$.
Thus, we do not need to sample new elements from the orbit induced by $\zte$, and can instead re-use $G_i$, $i\in[K]$.

\subsection{Examples}\label{conditional-invariance-JCR-example}
In this section, we show how group invariance can be used to construct JCRs.

\subsubsection{Regression}

We return to the regression setting from Example \ref{example:lr-pivot} and outline an approach to construct JCRs based on weaker assumptions.

\begin{example}[Linear regression]\label{lr}
We consider the regression setting from Example \ref{example:lr-pivot}, but
now assume only that $\ep_1,\ldots,\ep_n,\epte$ are exchangeable. 
Specifically, we denote
$I(\theta, z) = Y^+-X^+\theta$, 
and consider the  permutation group 
$\mG = S_{n+1}$ on $n+1$ elements.
This group acts by permuting the entries of $g$, 
represented via $(n+1)\times(n+1)$ permutation 
matrices $g$. 

Since $I = (\ep_1,\ldots,\ep_n,\epte)^\top$ is an invariant function---in the sense of \eqref{if}---under the permutation group, 
we can consider arbitrary test statistics $m$ of $I$.
For instance, we may take the absolute covariance $m(I) = \left|\sum_{i \in N} (x_i - \overline{x})(I_i - \overline{I})\right|/n$, where $N=\{1,2,\ldots,\mathrm{te}\}$ and $\overline{x}, \overline{I}$ denotes the mean of $x_i$ and $I_i$ (respectively) over $i\in N$.
Since $\mG$  has $(n+1)!$ elements, which can be large, we can randomize and sample $K$ group elements $G_{1:K}$ from $\mG$. 
The corresponding randomized JCR from Theorem \ref{sample-JCR} is thus
    $$
    \left\{(\theta,\yte): m(Y^+-X^+\theta) \leq q_{1-\alpha''}\bigg(m\big(g_i(Y^+-X^+\theta)\big),i\in[K]\bigg)\right\}.
    $$
This permutation-based JCR is illustrated in Section \ref{sec:comparison-lr}.

  Alternatively, we can make the even weaker assumption
  of invariance under subgroups of $\mG$. 
  For instance, if we only assume that the noise is invariant under all cyclic shifts
  $(\ep_1, \ldots, \ep_n, \epte) \to (\ep_k, \ldots, \ep_n, \epte, \ep_1, \ldots, \ep_{k-1})$ for $k\ge 1$, 
  we can take the corresponding cyclic shift group $\mG_1$ acting on $(\ep_1,\ldots,\ep_n,\epte)$. 
  Considering a test statistic $m(I_1,\ldots,I_n,I_{\mathrm{te}}) = f(I_{\mathrm{te}})$, for some function $f$,
  a two-sided JCR turns out to depend on the empirical quantiles of $f$ over the coordinates:
    \begin{align}\label{jcyc}
   \tilde{J} = \left\{(\theta,\yte) : q_{\alpha_1}\big(f(y_i - x_i^\top\theta),\, i\in [n] \big) \leq f(\yte - \xte\theta) \leq q_{\alpha_2}\big(f(y_i - x_i^\top\theta),\, i\in [n] \big)\right\}.
\end{align}
This coincides 
with the JCR  under full permutation invariance;
but it is valid more generally under cyclic-shift invariance.
\end{example}

\subsubsection{Signal-plus-noise model}
We next consider certain signal-plus-noise models, aiming to jointly provide confidence regions for the signal, 
and prediction regions for future observables from the model.
\begin{example}\label{example-spn}
  Consider
  $n$ independent observations $X_i \in \R^{p \times 1}$, $i\in [n]$ such that $X_i = \theta + \ep_i$ for a signal parameter $\theta \in \Theta\subset \R^{p \times 1}$, and for noise vectors $\ep_i$.
  Using these observations, we are interested to construct a JCR for $\theta$ and a future independent observation $\Xte = \theta + \epte$.
  Thus, 
  we have the full data $z=(x_1,\ldots,x_n,\xte)$ and the observed data $o(z)= (x_1,\ldots,x_n)$. 
  While one could consider several types of invariance, here 
  we assume spherically distributed noise \citep{kai1990generalized,gupta2012elliptically,fang2018symmetric}, 
  i.e., that 
   for all $i\in \{N\}$ and
  for any orthogonal matrix $O$ belonging to the orthogonal group $\mG_0 = O(n+1)$,
  $\ep_i =_d O \ep_i$.
  Then the noise is invariant under the direct product $\mG = \mG_0^{n+1}$.
  However, the noise distribution can vary across observations.

  We re-arrange the model as $X^+ = 1_{n+1}\theta^\top + E^+$, where $X^+=(x_1^\top;\ldots;x_n^\top;\xte^\top)^\top$, $E^+=(\ep_1^\top;\ldots;$ $\ep_n^\top;\epte^\top)^\top$. We also denote $X=(x_1^\top;\ldots;x_n^\top)^\top$, $E=(\ep_1^\top;\ldots;\ep_n^\top)^\top$.
    We consider 
    the invariant function
    $I(\theta,z) = X^+ - 1_{n+1}\theta^\top$
    and the
    test statistic $m(I) = \| I^\top 1_{n+1}\|_\infty/(n+1)$.
    A randomized JCR is obtained by sampling $G_1,\ldots, G_K$ i.i.d.~from the Haar measure over $\mG$:
       \begin{align*}
        J_{G_{1:K}}(x_1,\ldots,x_n) = \{(\theta,\xte): m(I) \leq q_{1-\alpha''}\left(m(G_iI), i\in[K] \right)  \}.
    \end{align*}

\end{example}

Above, we relied on  orthogonal invariance. 
If we only assume the weaker condition
that the noise vectors have independent sign-symmetric entries, 
then we can use the sign-flip matrix group $\mG_0 = \{\diag(a_1, \ldots, a_{p+1}), a_i\in \{\pm 1\}, i\in [p+1]\}$. 
However, a limitation is that
the prediction component of the JCR is less informative. 
Nevertheless, since we only know that the noise is symmetrical around zero and $\epte$ is independent of $\ep_i$, $i\in [n]$, it is reasonable to be conservative in the prediction component without additional information.

\subsection{Adequate Sets under Group Invariance}\label{sec:adequate-set-invariance}

Although split JCRs can be faster to compute under group invariance, 
in certain cases it might still be intractable to compute $q_{\alpha}\big(m(U_{O_I(\theta,z)})\big)$ or $m(g_iI(\theta,z))$ for $g_i \in \mG, i\in[K]$. 
Moreover, in the split setting, 
we need to compute this for each $\zte$.
We now show how to use
adequate sets
from Section \ref{sec:adequate-set} under group invariance to reduce the computational cost.

We assume that the action of $\mG$ decomposes under the adequate map
$A: \Theta \times \mZte \to \mathcal{A}$:
 for all $g\in\mG$, there is a function $g' = g'(g): \Theta \times \mZc \times \mathcal{A} \to \mI$
depending on $g$, 
such that 
for all $\theta\in\Theta$, $z = (\zc,\zte) \in \mZ$ and $g\in\mG$, the group
action has the structure
$$gI(\theta,z) = g'\left[\theta, \zc, A(\theta,\zte)\right].$$ 
Thus, the adequate map $A$ 
captures the dependence of the action of $g$ 
on $\zte$.
Then evaluating JCRs reduces to computing regions to which $A$ belongs. Intuitively, we aim to obtain a region for $A$ that obeys \eqref{eq:split-group-inv}. We construct this as an adequate set $W:\Theta\times \Zc \to B_{\mA}$, for $A$ such that,
for all $a\in \mA$,
\begin{align}\label{ade-inv}
a \in W(\theta,\zc)
\textnormal{ iff } 
m(g'[\theta,\zc,a]) \ge q_{\alpha'}\left(m(U_{O_I(\theta,z)})\right).
\end{align}
Then, using the adequate set $W(\theta,\zc)$, we can construct a JCR for $\yte$, 
computed for each new input $\xte$ by querying $A(\cdot)$:
\beq\label{ade}
J(\zc,\xte) = \left\{ (\theta,\xte,\yte): A(\theta,(\xte,\yte)) \in W(\theta,\zc) \right\}.
\eeq

Generally, we want $A$ to have a simple form
(e.g., a linear function taking values in $\mathcal{A} = \mathbb{R}$). 
We give an example below.

\begin{example}
    We consider one-dimensional regression (Example \ref{lr}),
    assuming the noise is invariant under the cyclic shift group $\mG$, 
    with $|\mG|=n+1$, acting on $I = (y_1-x_1^\top\theta,\ldots,y_n-x_n^\top\theta,a)^\top$, where  $a =
    A(\theta,(\xte,\yte)) = \yte - \xte^\top\theta \in \R$.
        Let $m(I) = \left|I_{n+1} - \sum_{i=1}^{n+1} I_i/(n+1)\right|$. 
    For a group element $g \in \mG$, 
    such that the last coordinate of $gI$ is $I_j$,
    $m(I) \leq m(gI)$
    amounts to
    \begin{align*}
     \left|a- \frac{(\sum_{i=1}^{n} I_i + a)}{n+1}\right| \leq \left|I_j- \frac{(\sum_{i=1}^{n} I_i + a)}{n+1}\right|.
\end{align*}
Let $W_g \subset \R$ 
denote the set of $a\in \R$ satisfying the above inequality. 
For $n>1$, one can verify directly that this is an interval, 
since the coefficient $n/(n+1)$ of $a$ on the left-hand side is greater than the corresponding coefficient $1/(n+1)$ on the right. 
Then an adequate set for $a \in \R$ is
$$
W(\theta,\zc) = \left\{ (\xte,\yte) : \yte - \xte^\top\theta \in W_g \textnormal{ for at least } \lfloor \alpha|\mG|\rfloor \textnormal{ group elements } g \in \mG \right\}.
$$
One can verify that \eqref{ade-inv} holds.
The associated JCR is
    \begin{align*}
        \tilde J(\zc,\xte) = \big\{(\theta,\yte): \yte - \xte^\top\theta \in W(\theta, \zc)\big\}.
    \end{align*}
\end{example}   
After computing $W$, 
one can compute this for a new test feature $\xte$, by checking when the condition $\yte - \xte^\top\theta$ holds.
If we can find $W$ in a closed form, this may be done analytically.

{\bf Randomization.} 
For randomization with an adequate set,
we assume that a finite number of transforms $\{g_{1:K}\}$ are obtained via sampling, 
and let $g_0$ be the identity element of $\mG$. 
We aim to compute \eqref{eq:sample-split}
for all given $\xte \in \mathcal{X}_{\mathrm{te}}$.
To begin, for each transform $g_i, i \in [K]$, we compute the 
set of 
$a \in W_i(\theta, \zc) \subset \mathcal{A}$ for which
$m(g_i'\left[\theta, \zc, a \right]) \leq m(g_0'\left[\theta, \zc, a \right]).$
Then, we can construct the adequate set $W(\theta, \zc)$, which includes those $a \in \mathcal{A}$ that appear in more than $\lfloor{\alpha (K+1)}\rfloor$ sets $\{W_i(\theta, \zc)\}_{i\in[K]}$.
With the adequate set $W$, we construct the JCR for any $\xte$ via \eqref{ade}.
     The full procedure
    is shown in Algorithm \ref{a3}.
We show below that this returns a valid $1-\alpha$ prediction region.

\begin{algorithm}
\caption{JCR based on adequate sets under group invariance}
\label{a3}
\KwIn{Observations $\zc$, $\xte \in \mathcal{X}_{\mathrm{te}}$, invariant function $I:\Theta \times \mZ \to \mI$, transforms $\gk$ in $\mG$.}
\KwOut{Joint Coverage Region for $Z$ and $\theta_P$ for each input $\xte \in \mathcal{X}_{te}$.}
Set $J(o^*) = \emptyset, \mZ^* = \{z \in \mZ: o(z) = o^*\}$\\
Choose a test statistic $m: \mI \times \mZt \to \R$, which may depend on  $\zt$.\\
\For{$\theta \in \Theta$}{
    \For{$i \in [K]$}{
        Compute the set $W_i(\theta, \zc) \subset \mathcal{A}$ of 
        $a$ for which
        $m(g_i'\left[\theta, \zc, a \right]) \leq m(g_0'\left[\theta, \zc, a \right])$
    }
    Construct $W(\theta, \zc)$, containing $a \in \mathcal{A}$ that appear in more than $\lfloor{\alpha (K+1)}\rfloor$ of the sets $\{W_i(\theta, \zc)\}_{i\in[K]}$.\\
}    
\For{$\xte \in \mathcal{X}_{\mathrm{te}}$}{
    Compute $J(\zc,\xte) = \left\{ (\theta,\xte,\yte): A(\theta,(\xte,\yte)) \in W(\theta,\zc) \right\}.$
}
\end{algorithm}

\begin{theorem}
 Algorithm \ref{a3} returns a 
$1 - \alpha$-joint coverage region. 
\end{theorem}

This result can be viewed as a special case of Theorem \ref{thm:adequate-set}, and its proof is given in Section \ref{proof:adequate-set}.
Example \ref{example:adequate-set-spn} illustrates it for multivariate regression.

\begin{example}
\label{example:adequate-set-spn}
    We consider a multivariate multiple-output regression model as an extension of Example \ref{example-spn}. Assume that we have $n$  observations $x_i \in \R^{k \times 1}, y_i \in \R^{p \times 1}$, $i\in [n]$ from the model $Y_i = \theta^\top x_i + \ep_i$ for 
    $i \in N =\{1,2,\ldots,n,\mathrm{te}\}$. 
    Here $\theta \in \R^{k \times p}$ is the unknown regression parameter.
    With $\zc = \{(x_i,y_i) \textnormal{ for } i\in [n]\}$, we are interested in obtaining a JCR jointly for $\theta$ and new observations $\Yte$, for each of a sequence of $\xte$-s. 
    
For illustration, as in Example \ref{example-spn}, we assume that 
the noise $\ep_i \in \R^{p \times 1}$, $i \in N$,
are independent 
and orthogonally invariant.
We then consider a test statistic $m(I) = \| I^\top 1_{n+1}\|_\infty/(n+1)$, where $I(\theta,z) = Y^+ - \theta X^+$ is invariant under $\mG$. 
As discussed in Example \ref{example-spn}, a randomized JCR is
       \begin{align*}
        J(x_{1:n},y_{1:n},\xte) = \{(\theta,\yte): m(I) \leq q_{1-\alpha''}\left(m(g_iI), i\in[K] \right)  \},
    \end{align*}
    where $G_{1:K}$ are sampled i.i.d.~from the Haar measure over $\mG$.
    
    Here we describe a corresponding adequate set for reducing the computational cost. We take the adequate map $A(\theta,\zte) = \yte - \theta^\top \xte$, and denote $I_i = y_i - \theta^\top x_i$ for simplicity. For each $g_j$, $j\in[K]$, we consider
    \begin{align*}
    W_i(\theta,\zc) = 
        &\left\{\zeta: \max\{(|I_j^\top 1_{n+1}|)_{j\in[n]},|\zeta^\top 1_{n+1}|\}    \leq
        \max\{(|(g_jI_1)^\top1_{n+1}|)_{j\in[n]},| (g_j\zeta)^\top1_{n+1}|\} \right\}.
    \end{align*}
    Only $|\zeta^\top 1_{n+1}|$ and $|(g_j\zeta)^\top 1_{n+1}|$ depend on $\zeta$ when $\theta$ is fixed, thus we can find the region $W_i(\theta,\zc)$ for $\zeta$ by simple algebra. Then, by considering $\zeta\in \A$ belonging to more than $\lfloor \alpha (K+1) \rfloor$ sets $\{W_i(\theta,\zc)\}_{i \in [K]}$, we find the adequate set $W(\theta,\zc)$.
    Finally, for each input $\xte$, we can find the corresponding JCR via \eqref{ade},
    which may save computation
    when we have a large number of test points $\xte$.
    
\end{example}

The next example is an extension of Example \ref{lr}.

\begin{example}[{One-dimensional Regression:  Spherical Noise}]\label{linear-regression-adequate-set}
In Example \ref{lr}, we assumed that the noise $(\ep_1^\top;\ldots;\ep_n^\top;\epte^\top)^\top$ is jointly invariant under a  group $\mG$
with a linear matrix representation.
To illustrate adequate sets,
we partition the $(n+1)\times(n+1)$ matrix $g_i$ as $g_i = (g_{i1};g_{i2})$, where $g_{i2}$ is the last column of $g_i$. We note
\begin{align*}
    g_i(Y^+ - X^+\theta) = g_{i1}(Y-X\theta) + g_{i2}(\yte - \xte^\top \theta) =: g'[\theta,(X,Y),\yte-\xte^\top \theta].
\end{align*}
For given $\theta$, 
we take $\mathcal{A} = \mathbb{R}$ and consider the adequate map $A(\theta,(\xte,\yte)) = \yte - \xte^\top  \theta \in \mathbb{R}$. 
Then Algorithm \ref{a3} proceeds as follows.
First, for $i \in [K]$ and some test statistic $m$ (e.g., $m(\cdot)=\|\cdot\|_\infty$), we compute the region $W_i(\theta, \zc)$ of $a \in \mathbb{R}$ that satisfy 
    $$m(g_{i1}(Y-X\theta) + g_{i2}a) \leq m(g_{01}(Y-X\theta) + g_{02}a) = m(Y-X\theta + \eta a),$$ where $\eta = (0,\,0,\ldots,\,0,1)^\top \in \R^{n+1}$. 
Next, we compute the adequate set $W(\theta,\zc)$ to include those $a \in \mathbb{R}$ that belong to at least $\lfloor{\alpha (K+1)}\rfloor$ sets $\{ W_i(\theta, \zc)\}_{i\in[K]}$.
Then, given $\xte$, we find a prediction region for $\yte$ via 
$J(\theta, o(Z)) = \{a + \xte^\top \theta |\,\, a \in W(\theta, \zc)\}$.

This can be extended to any regression model $Y = f(\theta,x) + \ep$ when $f$ is separable in the sense that for $X = (X_1,\ldots,X_n)$, we have---overloading notation---$f(\theta,X) = (f(\theta,X_1),\ldots,f(\theta,X_n))$. 
For instance, consider a simple neural network $f(\theta,x) = B_1\sigma(B_2\sigma(\ldots \sigma(B_lx)))$, where $\sigma$ denotes the ReLU activation with $\sigma(x) = \max\{0,x\}$, 
for all $x\in \R$; extended elementwise to matrices.
Here, $\theta = (B_1,\ldots,B_l)$ 
where $B_t \in \mathbb{R}^{m_t \times m_{t+1}}$ are the weight matrices for $t=1,\ldots,l$. 
Specifically, $m_{l+1} = p$ is the row-dimension of each input $x \in \mathbb{R}^{p \times q}$, for some $q$. 
We then take $I(\theta,Z) = (
Y - f(\theta,X),\yte - f(\theta,\xte))$. 
For $g_{i1}$ and $g_{i2}$ as in the linear case,
$$g_iI(\theta,Z) = g_{i1} (Y - f(\theta,X)) + g_{i2} (\yte - f(\theta,\xte)),$$ 
and thus we can find $W(\theta,\zc)$ as before. 
Thus, we find a prediction region for $\yte$ via
$T(\theta, o(Z)) = \{a + f(\theta,\xte) |\, a \in W(\theta, \zc)\}$.
\end{example}
    

\section{A Case Study in Linear Models}\label{sec:simulation}

In this section, we present a case study of JCRs
in linear models.
We compare several aspects, such as the coverage rate, showing that permutation-based JCRs are empirically valid under weaker assumptions than JCRs based on stronger invariance.
We also compare the shapes (size, boundedness, height and width) of JCRs. 
 

\subsection{The Shapes of JCRs Based on Spherical Invariance}\label{sec:shape-example}

\subsubsection{Normal Mean Problem}\label{sec:toy-model}

We start with a one-dimensional normal mean problem
to illustrate the shape of various JCRs.
Suppose for simplicity that 
we have independent observations $y_i \sim \N(\theta,1)$ 
for $i\in [n]$. 
We aim to find a JCR for $\theta$ and 
for an independent future observation 
$\yte \sim \N(\theta,1)$.
For any $\wte$,
 we have a pivot
$\sum_{i \in N} y_i + \wte \yte - (n+\wte)\theta \sim \N\left(0,n+ \wte^2\right).
$
Thus, we obtain the pivotal JCR \eqref{pivjcr} 
\begin{align*}
    \bigg\{(\theta,\yte): 
    \left|\wte\yte - (n+\wte)\theta +  \sum_{i=1}^n y_i\right| \leq \sqrt{n + \wte^2}q_{1-\alpha/2}
    \bigg\}.
\end{align*}
Further, 
denoting $\omega = 1/\wte \neq 0$, this equals
\begin{align}\label{eq:toymodel-JCR1}
    \bigg\{(\theta,\yte): 
    \yte -  \left(1+\cfrac{n}{\omega}\right)\theta  + \frac{1}{\omega}\sum_{i=1}^n y_i \in 
    \sqrt{1 + \cfrac{n}{\omega^2}} \cdot [q_{\alpha/2}, q_{1-\alpha/2}]
    \bigg\}.
\end{align}
This parametrization directly controls the shape of the JCR.
When $\omega$ increases, 
the JCR has a shorter prediction component length $2\sqrt{1 + n/\omega^2}q_{1-\alpha/2}$. 
As $\omega \rightarrow \infty$, the JCR
becomes approximately 
$|\yte - \theta| \in [q_{\alpha/2},q_{1-\alpha/2}]$,
whose bounds are the normal quantiles. 
The observations $y_i$, $i \in [n]$ do not play a role in this limit.

Further, when 
$\omega = -n$, the region is parameter-free, and 
is equivalent to a pure prediction region generated by the ancillary statistic
$\yte - \frac{1}{n}\sum_{i=1}^n y_i \sim \N\left(0,1+\frac{1}{n}\right).$
When $\omega \neq -n$, 
the confidence component of the JCR is:
\begin{align}\label{eq:toymodel-JCR2}
    \left\{(\theta,\yte): 
    \left|\theta - \frac{\omega}{\omega + n}\yte + \frac{1}{\omega + n}\sum_{i=1}^n y_i\right| \leq \sqrt{\frac{\omega^2 + n}{(\omega + n)^2}} q_{1-\alpha/2}
    \right\}.
\end{align}
When $\omega \rightarrow 0$, this becomes approximately $|\theta - \sum_{i=1}^n y_i / n| \leq q_{1-\alpha/2}/\sqrt{n}$, which is the standard two-sided normal confidence interval for $\theta$.

For a fixed $\omega < \infty$, as the sample size $n$ increases, both 
the slope of $\theta$ in \eqref{eq:toymodel-JCR1} 
and the width $2\sqrt{1 + n/\omega^2}q_{1-\alpha/2}$  of the vertical section increase. 
Intuitively, if we use more datapoints in \eqref{eq:toymodel-JCR2} while keeping the weight $\omega$ of $\yte$ unchanged, the relative influence of $\yte$ decreases,
as reflected in the slope $\omega/(\omega + n)$, yielding a less informative prediction component.
On the other hand, more data causes the region's confidence component to shrink, as expected.

Figure \ref{fig:shape-example1} shows an example with $n=100$ 
and $\omega = -n,0,1,0.1n,\infty$, 
which lead to very different JCRs. Specifically, $\omega \rightarrow 0$ yields a vertical strip, i.e., a confidence region, 
while $\omega = -n$ leads to a  horizontal strip, i.e., a prediction region. 

\begin{figure}[ht]
    \centering
    \includegraphics[height = 7.3cm, width =8.2cm]{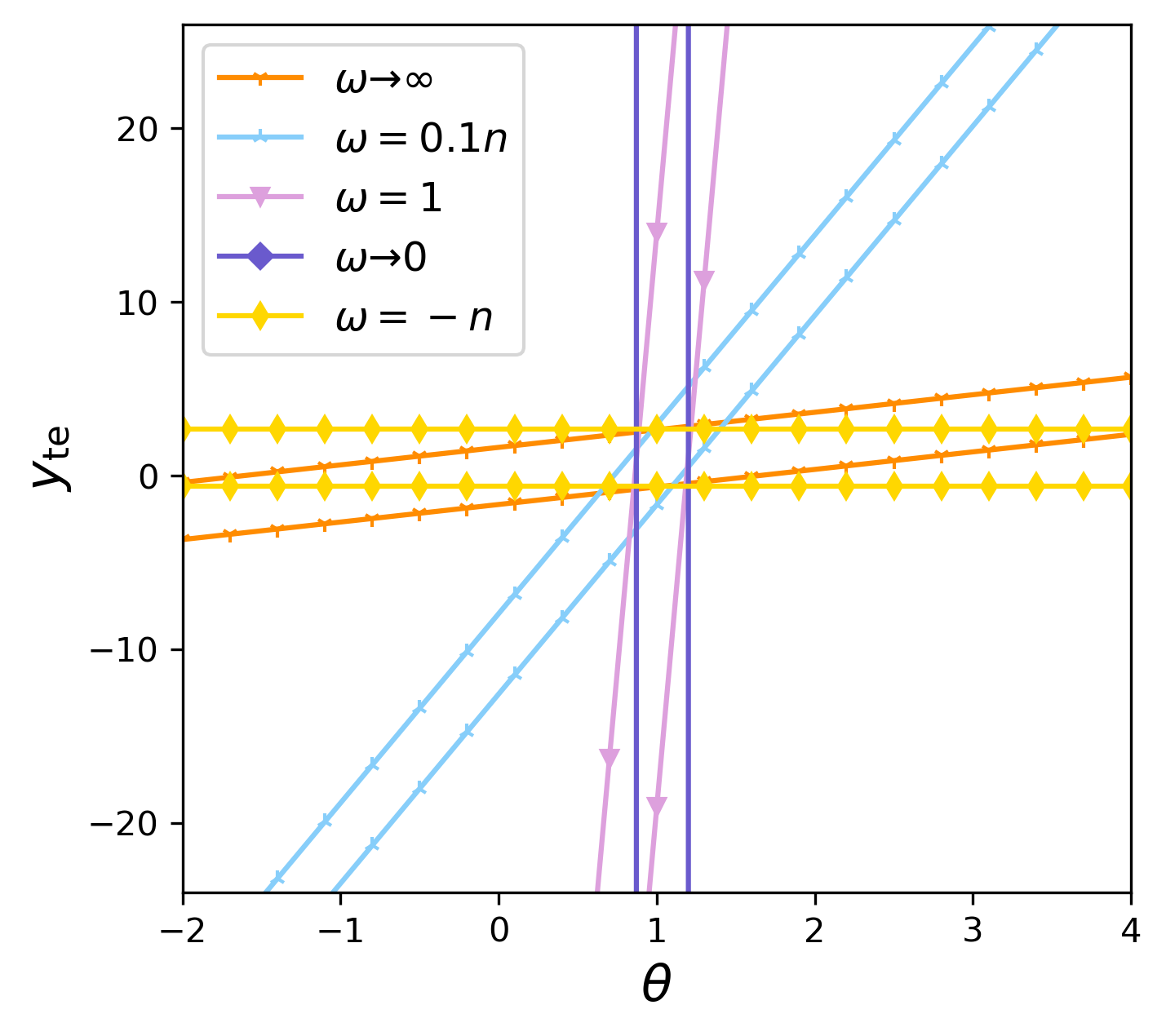}
\caption{A comparison of JCRs generated by various $\omega$-s as described in Section \ref{sec:shape-example}.}
\label{fig:shape-example1}

\end{figure}

\subsubsection{Linear Regression}

We now consider linear regression with normal noise. 
The same analysis applies to spherically distributed noise; this is omitted for clarity.
We first study the one-dimensional case, followed by the multi-dimensional case below.
Thus, let $y_i = x_i\theta + \ep_i, i \in N = \{1,2,\ldots,\mathrm{te}\}$ where $\ep_i \sim \N(0,\sigma^2)$ are iid, 
$x_i$ are considered fixed, and $\sigma^2$ is unknown.
As above, for any $w_i,i\in N$,
we have a pivot
\begin{align*}\label{eq:lr-pivot}
    \sum_{i \in N} w_i y_i - \bigg(\sum_{i \in N} w_i x_i\bigg)\theta \sim \N\left(0,\bigg(\sum_{i \in N} w_i^2\bigg)\sigma^2\right).
\end{align*}

Letting $\htheta_{\mathrm{OLS}}$ be the usual OLS estimator,
and
$S^2 = \sum_{i=1}^n (y_i - x_i\htheta_{\mathrm{OLS}})^2/(n-1) \sim \sigma^2 \chi^2_{n-1}/(n-1)$,
if $\sum_{i \in N} w_i^2>0$,
we find a JCR based on the Student $t$-distribution:
\begin{equation}\label{eq:lr-JCR}
    \left\{(\theta,\yte): t_{n-1, \alpha/2} \leq \frac{\sum_{i=1}^n w_i y_i + \wte \yte - (\sum_{i \in N} w_i x_i)\theta}{S\sqrt{\sum_{i \in N} w_i^2}} \leq t_{n-1, 1-\alpha/2}\right\}.
\end{equation}
Here, for $c\in(0,1)$, $t_{n-1,c}$ is the $c$-quantile of the $t$ distribution with $n-1$ degrees of freedom.
For instance, 
for $(w_1,\ldots,w_n) = X^\top/\|X\|^2$, $\wte = 0$, we obtain a confidence region for $\theta$ based on the pivot $(\htheta_{\mathrm{OLS}} - \theta)/(S/\|X\|) \sim t_{n-1}$.
On the other hand, taking $(w_1,\ldots,w_n) = X^\top/\|X\|^2$ and $\wte = -(\sum_{i=1}^n w_i x_i)/\xte$ yields the usual prediction region
$$
    \yte \in \xte^\top\thetaols + 
    \sqrt{\xte^2/\|X\|^2+1}\cdot S\cdot [t_{n-1, \alpha/2},
     t_{n-1, 1 - \alpha/2}].
$$
We visualize and compare these methods in Section \ref{sec:comparison-lr}.

\subsection{JCRs Based on Permutation Invariance}\label{sec:permutation-lr}

In this section, we 
study JCRs in the linear model based on the weaker assumption of permutation invariance. 
We assume that the noise vectors $\ep_i, i \in N$ are exchangeable and
represent the action of permutations on $\R^{n+1}$ by the group $\mG$ of $(n+1)\times(n+1)$ permutation matrices.

Recalling $\gamma = c^\top \theta$, suppose that there is a $\delta \in \R^{p \times 1}$ such that 
for all $\theta\in\Theta$,
and for some $\Psi \in \R^{(n+1) \times 1}$,
$X^{+} \theta + 1_{n+1} \delta^\top \theta = \Psi \gamma$. 
This 
holds for $p=1$,
with $\delta = 0_{n+1}$ and $\Phi = X^\top/c$. 
In higher dimensional settings, 
it does not always hold, but 
it does in the important case of a two-sample 
problem
where 
$X_i \sim \mu_1 + \ep_i$ 
for $i\in [m]$
and 
$Y_i \sim \mu_2 + \ep'_i$
for $i\in [n]$;
with i.i.d.~noise variables.
This is a regression model
with $\theta = c^\top \gamma$, where $\gamma = (\mu_1,\mu_2)^\top$, $c = (1,-1)^\top$
and $x_i = (1,0)^\top$ or $(0,1)^\top$ determined by the group.

In the general model, 
since 
the coordinates of $Y^{+} - X^{+} \theta -1_{n+1} \delta^\top=E-1_{n+1} \delta^\top$ are exchangeable, 
we have the invariant function 
\begin{align}\label{eq:permutation-JCR}
    I(\gamma, (X^{+},Y^{+}))
    = Y^{+} - h(\gamma) =  Y^{+} - X^{+} \theta - 1_{n+1} \delta^\top \theta.
\end{align}
For any test statistic $m$ 
and $G_{1:K}$ sampled i.i.d.~from the uniform measure over $\mG$,
a permutation-based JCR is
\beq\label{pj}
\bigg\{(\theta,\yte): m(I) \leq q_{1-\alpha/2}\big( m(g_iI), i\in[K] \big) \bigg\}.
\eeq

Different $m$-s lead to JCRs 
with varying foci.
For instance, 
in a
one-dimensional setting
where 
$I = Y^+ - X^+ \theta$, we can consider the weighted statistic $m(I) = \sum_{i \in N} |a_i(I_i- \overline{I})|$, where $\overline{I} = \sum_{i \in N} I_i / (n+1)$
and $a_i \in \R$ for all $i\in N$. 
For a statistic where $a_i$ are  relatively balanced across observations, 
adding a prediction component for an unknown element $I_\mathrm{te}$ does not greatly influence the region when $n$ is large. 
Since $m$ treats the elements of $I$ similarly, the associated JCR would still focus on the confidence side. 

Further, if we consider the subgroup of $\mG$ that keeps the last coordinate fixed, the JCR turns out to be a confidence region, since $\yte$ does not contribute to the results of the comparisons of test statistics for various $g_i$.
Also, if we take $a_{\mathrm{te}} = 1$, while $a_i = 0$ for $i\in [n]$ and use the cyclic shift group, 
the region can be viewed as estimating a quantile of the distribution of the residual $|\epte - \overline{\ep}|$ using a quantile of the empirical distribution of $|\ep_i -\overline{\ep}|$, for $i\in[n]$.
We will compare the above approaches in Sections \ref{sec:comparison-lr}.

\subsection{A Comparison of JCRs}\label{sec:comparison-lr}

In this section, we compare several JCRs in a one-dimensional linear regression model.
We consider independent
training datapoints $(x_i,y_i),\, i\in [n]$ generated from a linear model $y_i = x_i \theta + \ep_i$, and we aim to find JCR for 
$\theta$ and
an independent observation $\yte = \xte \theta + \epte $ given a new input $\xte$.
We take $n=500$, generate (and fix) i.i.d.~$x_i \sim U[0,1]$,
set $\xte = 5$, 
and consider noise entries sampled i.i.d.~from
$\ep_i \sim \N(0,1)$. 
We consider the following $1-\alpha$ JCRs. 

\begin{compactitem}
  \item {\bf Intersection-based JCR}: We show an intersection of classical confidence and prediction intervals, each with coverage level $1-\alpha/2$, namely
$
C = \htheta_{\mathrm{OLS}} + S/\|X\|^2\cdot [t_{n-1,\alpha/4},t_{n-1,1-\alpha/4}]$
and
$$
T = \xte^\top\thetaols + S \sqrt{\xte^2/\|X\|^2+1} \cdot 
\left[ - \sqrt{F_{1,n-1}^{1 - \alpha/2}},  \sqrt{F_{1,n-1}^{1 - \alpha/2}}\right],
$$
where $\htheta_{\mathrm{OLS}} = X^\top Y/\|X\|^2$, $S^2 = \sum_{i=1}^n (y_i - x_i\thetaols)^2 /(n-1)$.
    \item 
    {\bf Gaussian pivotal JCR}:
    As described in Section \ref{sec:shape-example}, we consider a JCR 
    based on a Gaussian noise distribution:
    $$
    \left\{ (\theta,\yte): t_{n-1, \alpha/2} \leq \frac{\yte - \xte \theta}{S} \leq t_{n-1, 1-\alpha/2} \right\},
    $$
    which is a special case of \eqref{eq:lr-JCR} 
     with $w_i = 0$, $i\in [n]$ and $\wte = 1$.

    \item {\bf Permutation-based JCR}: As described in Section \ref{sec:permutation-lr},  we consider the permutation group acting on $(\ep_1,\ldots,\ep_n,\epte)$, with the test statistic 
    \begin{align}\label{eq:pb-m}
     m(\ep_1,\ldots,\ep_n,\epte) = \left|\sum_{i\in \{N\}} (x_i - \overline{x})(\ep_i - \overline{\ep})\right|,   
    \end{align}
    where $\overline{x}, \overline{\ep}$ denotes the mean of $x_i$ and $\ep_i$, respectively, for $i\in \{N\}$. 
    We  use the JCR from \eqref{pj} with $K=500$.
    This is a one-dimensional special case of \eqref{eq:permutation-JCR}.

    \item
    {\bf Cyclic-shift-based JCR}:
    We consider the cyclic shift group, as discussed in Section \ref{sec:permutation-lr}, 
    and construct the JCR from \eqref{jcyc} with $f$ being the identity map and $\alpha_1 = 1-\alpha_2 = \alpha/2$.

\end{compactitem}

Figure \ref{fig:lr_normal_compare} visualizes and compares these JCRs. 
For reference, we also show the following.
 \begin{compactitem}    
    \item {\bf Oracle Prediction Component}: If we knew $\theta$, the shortest $1-\alpha$-prediction region for $\yte$, given $\xte$, would be $[\xte\theta + q_{\alpha/2}, \xte\theta + q_{1-\alpha/2}]$.
    We refer to this as the oracle prediction component, associated with an oracle JCR. In general, this depends on knowing $\theta$ and is not implementable. 

    \item {\bf $1-\alpha$-Bounded JCR}: To obtain a bounded JCR with $1-\alpha$-coverage, we can take the intersection of two $1-\alpha/2$-level regions in Figure \ref{fig:lr_normal_compare}, such as the intersection of a JCR and the classical confidence region.
    \item {\bf Truth}: the true parameter and outcome.
\end{compactitem}

\begin{figure}[ht]
  \begin{minipage}{0.45\linewidth}
          \includegraphics[height = 6.7cm, width =7cm]{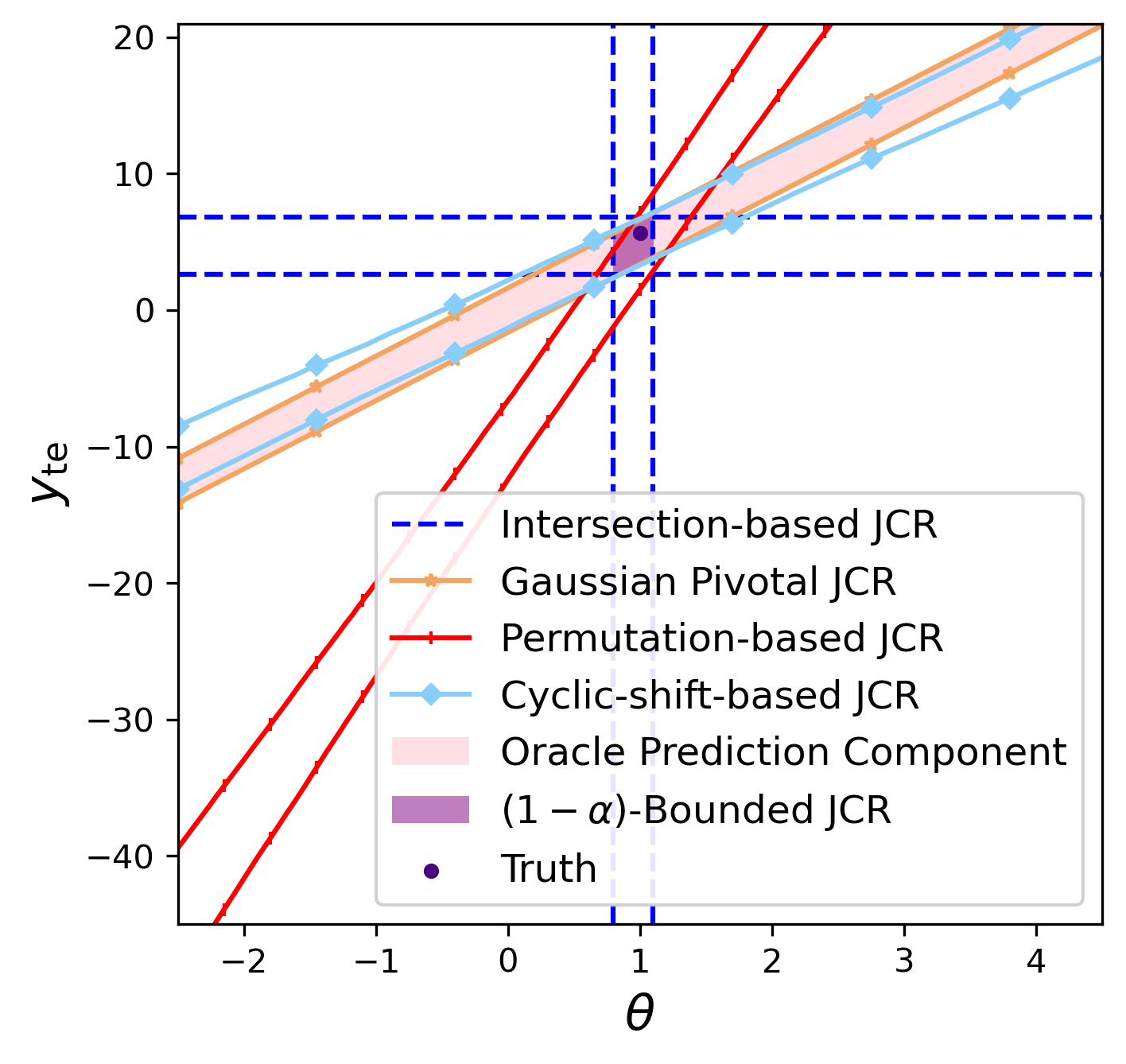}
\caption{A comparison of JCRs with $1-\alpha$ coverage, as presented in Section \ref{sec:comparison-lr}.}
\label{fig:lr_normal_compare}
  \end{minipage}
  \qquad
  \begin{minipage}{0.45\linewidth}
        \includegraphics[height = 6.7cm, width =7cm]{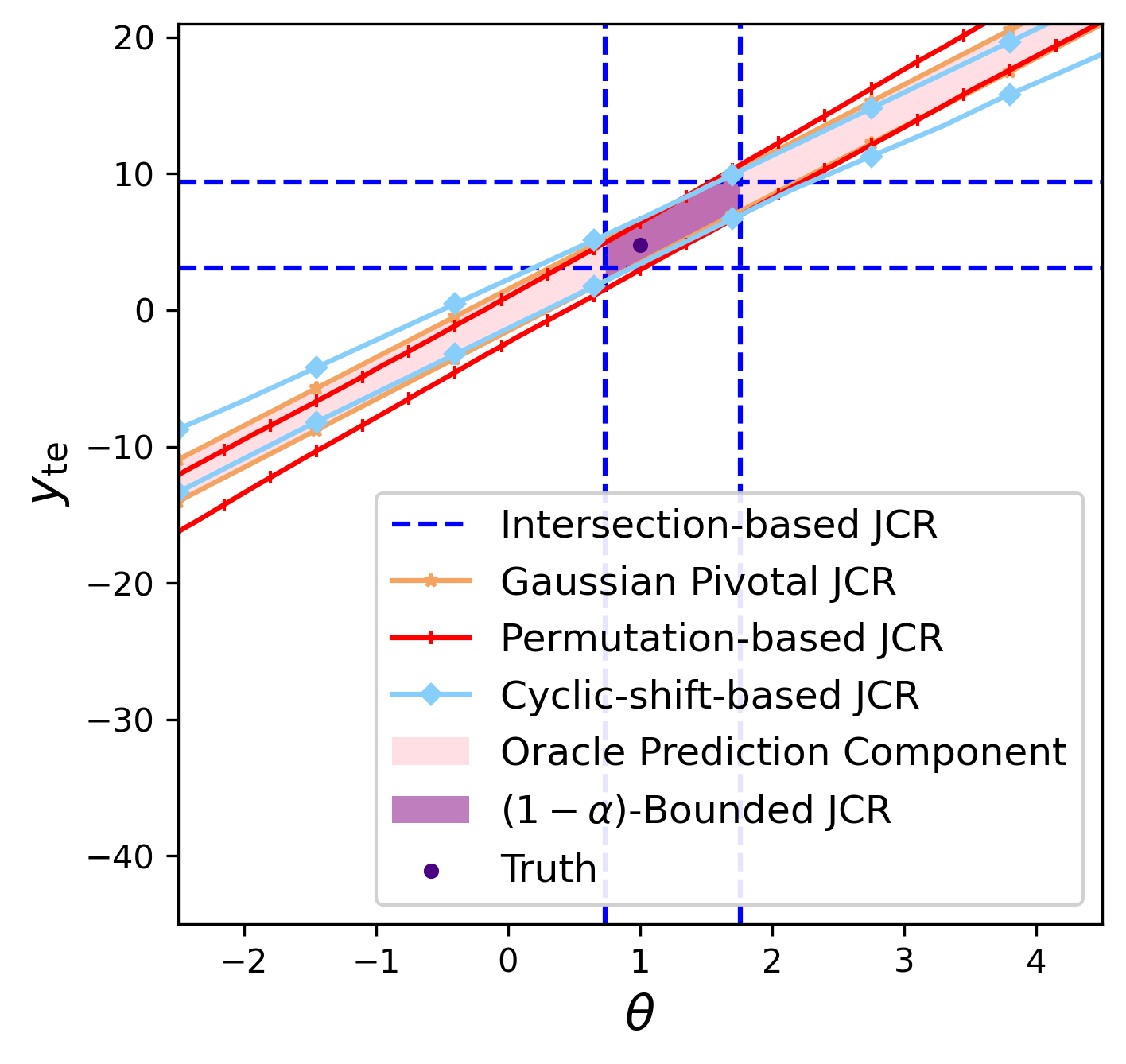}
\caption{A comparison of JCRs presented in Section \ref{sec:comparison-lr} with a small sample size $n=50$.}
\label{fig:lr_normal_compare_small_samplesize_50}
  \end{minipage}
\end{figure}

   As in Section \ref{sec:shape-example},
   the shapes of JCRs have various implications. 
   In any JCR, the horizontal sections over $\yte$ can heuristically be viewed as the confidence regions for $\theta$ given specific $\yte$ (while of course they are in general not conditionally valid regions). 
   On the other hand, the vertical section over $\theta = 1$ is a parameter-aware prediction region for $\yte$ given $(x_i,y_i), i\in [n]$. 
   In a sense, under the true parameter $\theta= 1$, 
   the cyclic shift-based and the Gaussian pivotal JCR estimate the quantiles of the distribution of the noise.
    This yields prediction components close to the oracle one.  
    
    Figure \ref{fig:lr_normal_compare}
    also shows that the permutation-based JCR has a shorter confidence component. 
    Indeed, \eqref{eq:pb-m} treats the elements of $I$ equally, and the associated JCR focuses more on its confidence component. 
    On the other hand, $I_\mathrm{te}$ and other $I_i$-s are treated asymmetrically in the Gaussian and cyclic pivot based JCRs. 
    Compared to the intersection of classical confidence and prediction intervals, an intersection of 
    a JCR and the classical confidence interval
    yields a smaller region, better capturing the linear structure of the problem.

Figure \ref{fig:lr_normal_compare_small_samplesize_50} further visualizes JCRs in an example with a smaller sample size $n=50$, while keeping all else unchanged. 
Again, we can take the intersection of two $1-\alpha/2$ regions to yield a bounded $1-\alpha$-JCR, 
such as the intersection between the permutation-based JCR and a classical confidence region. 
Generally, when the sample size is smaller, the intersection-based JCR is wider
than invariance-based JCRs and
their intersections with classical confidence regions. 
We further show the advantages of bounded JCRs in Section \ref{sec:multiple-task}.
Meanwhile, as the sample size increases, the permutation-based JCR approaches a vertical strip,
i.e., a pure confidence region, 
which conforms with the analysis from Section \ref{sec:toy-model}.

\subsubsection{Weaker Assumptions}

In this subsection, we show the robustness of the permutation-based JCR 
beyond normal noise. 
We use the setup from Section \ref{sec:comparison-lr}, with a sample of size $n=100$.
We compare the intersection-based JCR, Gaussian pivotal JCR, cyclic shift-based JCR, and permutation-based JCR (see Section \ref{sec:comparison-lr})
for various noise distributions. 

In Table \ref{table:compare-noise}, we report two-sided 95\% Clopper-Pearson confidence intervals  (CPCIs)
for the binomial parameters of coverage, 
based on the empirical coverage rate over $2,000$ repeated experiments. 
If an interval contains $0.9$, the corresponding approach is consistent with valid coverage. 
All approaches are empirically valid
under normal noise. 
The intersection-based JCR based on $95\%$ confidence and prediction regions is slightly conservative.
Permutation-based approaches empirically have a correct coverage under i.i.d.~noise, 
while approaches based on Gaussian or orthogonal assumptions are not valid anymore.

\begin{figure}
    \centering
    \includegraphics[width=0.5\linewidth]{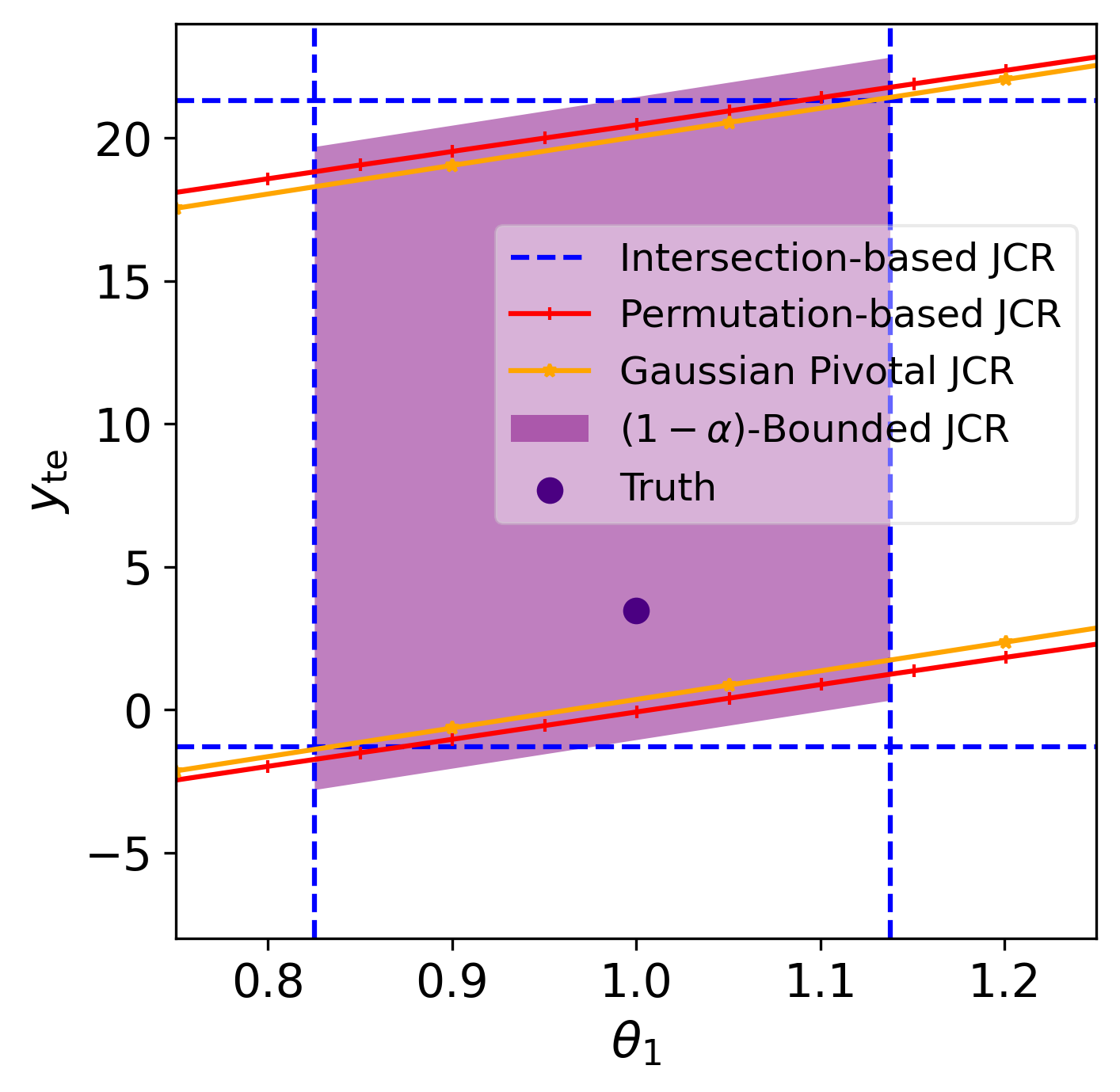}
    \caption{A comparison of joint coverage regions with $1 - \alpha$ coverage, as presented above with $p = 5$.}
    \label{fig:comparison_highd}
\end{figure}

We may also consider a multi-dimensional regression model $y_i = x_i^\top \theta + \ep_i$, $\ep_i \sim \N(0,\sigma^2), i \in [n]$ for  $\theta \in \R^p$.  
We write $X \in \R^{n \times p}$ for the observation matrix and $y = (y_1,\ldots,y_n)^\top $ for the response.
We consider  bounded $\theta_i \in [a_i,\ b_i], i \in [n]$,
and 
for an independent datapoint $\yte = \xte^\top \theta + \epte$,
aim to construct a 
$(1 - \alpha)$-joint coverage regions for $\theta_1$ and $\yte$. 
Similar to the one-dimensional case, we consider various joint coverage regions.
The details are defered in the Appendix.

In Figure \ref{fig:comparison_highd}, we compare several regions
using $n = 50,000$, $p = 5$, 
and the coordinates of all features of all datapoints sampled uniformly over $[0,1]$. 
We assume that the parameter $\theta = (\theta_1,\ldots,\theta_p)^\top $ satisfies $\theta_i \in [0.75, 1.25]$, 
and set the true parameter to be $\theta = (1,1,1,1,1)^\top$.
We  take $\sigma^2=5$.
For a new datapoint with $\xte = (10, 0.05, 0.05, 0.05, 0.05)^\top$,
we construct 0.95-joint coverage regions for $\theta_1$ and $\yte$. 
The results mirror those from Section \ref{sec:comparison-lr}. 
Compared to intersection-based approaches, the slopes in 
Gaussian-pivot-based and permutation-based joint coverage regions better reflect the signal pattern in our model. 
We show a further example for $\xte = (0.5,0.5,0.5,0.5,0.5)$ in  the Appendix.

\begin{table}
\centering
\begin{tabular}{ccccccc}
\toprule 
\multicolumn{2}{c}{Noise Distribution}& {Intersection}& {Gaussian pivot}  & {Cyclic-shift} & {Permutation}\\  
\hline 
\multicolumn{2}{c}{Normal $\N(0,1)$}& [90.03\%,\,92.55\%]& [88.39\%,\,91.09\%]& [87.81\%,\,90.57\%] & [89.55\%,\,92.12\%]\\  
\multicolumn{2}{c}{Heterosk. normal}& [92.81\%,\,94.95\%]& [99.82\%,\,100.00\%]& [88.50\%,\,91.19\%] & [87.28\%,\,90.10\%]\\  
\multicolumn{2}{c}{Cauchy noise}
& 
[92.76\%,\,94.91\%]& [95.88\%,\,97.48\%]& [88.18\%,\,90.90\%] & [89.50\%,\,92.08\%]\\   
\multicolumn{2}{c}{Uniform $U[-5,5]$}& [92.11\%,\,94.36\%]& [94.17\%,\,96.09\%]& [87.81\%,\,90.57\%] & [88.34\%,\,91.05\%]\\  
\bottomrule
\end{tabular}
\caption{Coverage of JCRs for various noise distributions, see Section \ref{sec:comparison-lr}. 
Heteroskedastic noise corresponds to the mixture distribution $0.2\cdot \N(0,1)+0.4\cdot \N(10,1)+0.4\cdot \N(-10,1)$.
We show confidence intervals for the coverage based on $2,000$ independent trials. 
The permutation-based approaches have valid coverage for arbitrary i.i.d.~noise, unlike approaches based on Gaussian or orthogonal assumptions.}
\label{table:compare-noise}
\end{table}

\section{Applications of JCRs}\label{sec:application}

In this section, we outline a few applications of JCRs.
We study simplified models, because our goal is to illustrate that JCRs can be used; and future work is needed to develop these applications.

\subsection{Prediction by JCR Projection}\label{sec:toy_jcr_pred}

We show that in some cases we can obtain better prediction regions by projecting a JCR.
Consider a setting where 
we have a measurement $X_j$ that is recorded periodically over periods $j=1,2,\ldots$. 
For instance, this could be the amount of funds in an account that records a lot of transactions.
Having observed the first few measurements $X_1, \ldots, X_T$, we wish to predict the next measurement $X_{T+1}$.
Suppose that it is reasonable to assume that the measurements follow a a parametric distribution 
$(X_1, \ldots, X_{T+1})\sim p_{\theta}$, with certain restrictions on the parameter $\theta\in \Theta$ captured by the parameter space $\Theta$.

One approach to this prediction problem is to construct a JCR $J$ in reduced form for $(\theta,X_{T+1})$, based on the observed data $X_1, \ldots, X_T$.
Then, we project JCR $J$
into the space where $X_{T+1}$ belongs, taking
$\tilde{T}(J) = \cup_{\theta \in \Theta} \Phi(J,\theta)$.
If $J$ is a $1-\alpha$ JCR, 
this is a $1-\alpha$-prediction region.
Depending on the structure of the parameter space and the noise level, we will show in an example below that this can lead to more efficient prediction regions than reasonable alternative methods.

For simplicity of illustration, 
we consider a setting with two measurements $X_1 \sim \N(\theta,1)$, $X_2 \sim X_1 + \N(\theta,1)$
that are jointly normal,
for an unknown mean parameter $\theta$.
In the above example, if there are a lot of small transactions (both deposits and withdrawals) recorded in the account, then a normal approximation for the amount in the account may be reasonable. 
Suppose that the parameter space $\Theta = [\theta_1, \theta_2] \subset \R$ is a  bounded interval. In our example, this may be motivated by the observation that, based on historical records, the average daily total transaction amount is bounded between some known values.
We observe $x_1$, i.e., $o(X_1,X_2) = x_1$, and we aim to find a $1-\alpha$-prediction region for $X_2$.

One can construct prediction regions in several ways.
One prediction region is $T_\alpha= \{X_2 \in 2x_1 \pm \sqrt{2}q_{1-\alpha/2}\}$,  generated by combing the pivots $X_1 - \theta \sim \N(0,1)$, $X_2 - X_1 - \theta \sim \N(0,1)$. 
This eliminates the unknown $\theta$ from the pivot $X_2 - 2X_1 \sim \N(0,1)$. 
Another method is to use an estimate $\htheta$ instead of $\theta$. 
In this case, for $\htheta = x_1$, we 
heuristically obtain the prediction region $T'= \{X_2 \in x_1 + \htheta \pm q_{1-\alpha/2}\} = \{2x_1 \pm q_{1-\alpha/2}\}$
using 
the approximation $\htheta\approx \theta$, which leads to 
the approximation
$X_2 - X_1 - \htheta \sim \N(0,1)$.
However, since $\htheta$ is noisy,
this approximation is inaccurate, and the JCR does not have the desired level of coverage.

Alternatively, we consider projecting the JCR $$
J_\alpha(x_1) = \{(\theta,X_2): X_2-x_1 - \theta \in [q_{\alpha/2},q_{1-\alpha/2}]\}
$$
into the space where $X_2$ belongs, i.e., we take
$\tilde{T}(J_{\alpha}) = \cup_{\theta \in \Theta} \Phi(J_\alpha,\theta).$
This is clearly a valid $1-\alpha$-prediction region,
and, we argue that it can sometimes be shorter than $T_\alpha$.
 In fact, the three intervals $T_\alpha, T', \tilde{T}(J_{\alpha})$ have widths $2\sqrt{2}q_{1-\alpha/2}, 2q_{1-\alpha/2}, 2q_{1-\alpha/2}+(\theta_2-\theta_1)$ respectively, while only $T_\alpha, \tilde{T}(J_{\alpha})$ have $1-\alpha$ coverage. 
 When $\theta_2 - \theta_1 < 2(\sqrt{2}-1)q_{1-\alpha/2}$, the projection $\tilde{T}(J_{\alpha})$ 
 is shorter than $T_\alpha$.
 
 We conduct simulations 
 with $\Theta = [-0.2,0.2]$, $\theta = 0$,
 and  $\alpha = 0.9$. 
 In a single trial with $x_1 = 0.3$ in Figure \ref{fig:toy_pred} (left), we show the regions $T_\alpha, T', \tilde{T}(J_{\alpha})$ defined above. 
 Here $\tilde{T}(J_{\alpha})$ is shorter than $T_\alpha$. 
 We also consider the coverage of the three methods over $10,000$ independent trials. 
 Specifically, we consider the models $X_1 \sim \N(\theta,\sigma^2)$, $X_2 \sim X_1 + \N(\theta,\sigma^2)$, with $\sigma$ varied from one to ten. 
 Figure \ref{fig:toy_pred} (right) 
 supports that only $T_\alpha, \tilde{T}(J_{\alpha})$ 
 have coverage above $1-\alpha$, 
 while $T'$ is anti-conservative due to the noise in estimating $\theta$. 
 The projection JCR becomes less conservative as the noise increases, 
 while its length $2 \sigma q_{1-\alpha/2} + (\theta_2 - \theta_1)$
 is shorter than  $2 \sqrt{2} \sigma q_{1-\alpha/2}$ for $T_{\alpha}$, once $\sigma > (\theta_2-\theta_1)/[2(\sqrt{2}-1)]q_{1-\alpha/2}$.
 This example illustrates that projecting a JCR can be an effective way to construct prediction regions in certain regimes; here specifically in the case of relatively large noise and bounded parameter space. 
 
\begin{figure}
\begin{minipage}[ht]{0.5\linewidth}
\centering
    \includegraphics[height = 7cm, width =7.4cm]{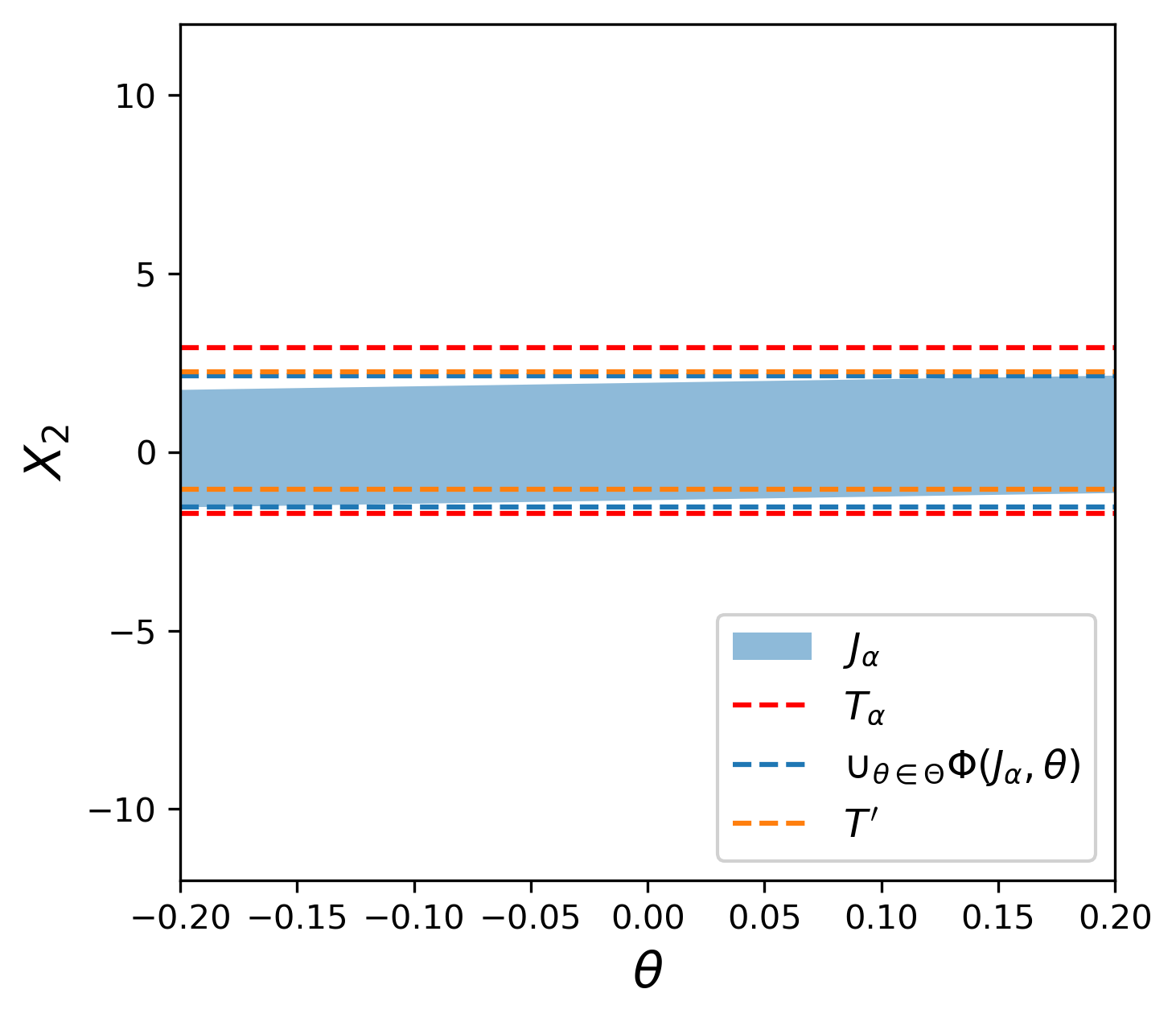}
    \end{minipage}
    \qquad
\begin{minipage}[ht]{0.5\linewidth}
\centering
    \includegraphics[height = 7cm, width =7.4cm]{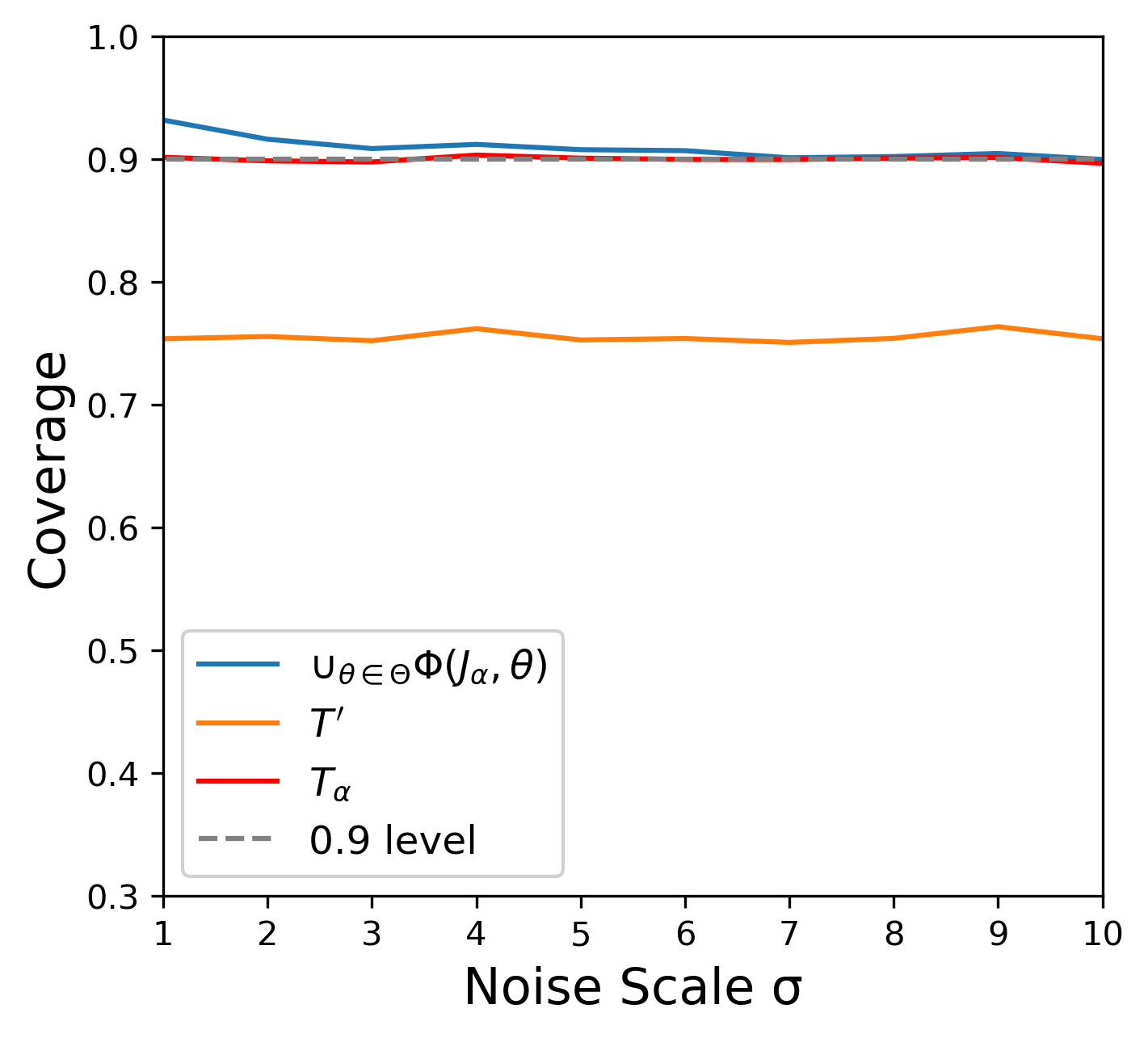}
    \end{minipage}
\caption{Left: A visualization of  $T_\alpha, T', \tilde{T}(J_{\alpha})$ as defined in Section \ref{sec:toy_jcr_pred}.
We consider a single trial with $x_1 = 0.3$. 
Right: The empirical coverage as a function of the noise level $\sigma$, ranging from one to ten.}
\label{fig:toy_pred}
\end{figure}

\subsection{Miscoverage Control in Multiple Inference Problems}\label{sec:multiple-task}

    
JCRs can be used for drawing inferences on multiple parameters and future observables.
As in Section \ref{sec:application}, consider two random variables $X_1,X_2$ that satisfy $X_1 \sim \N(\theta,1)$, $X_2 \sim X_1 + \N(\theta,1)$.
Suppose that we only observe $x_1$, i.e., $o(X_1,X_2) = x_1$  and we aim to (1)
construct a valid confidence region for $\theta$;
and 
(2) construct a valid prediction region for $X_2$.

Our goal is to control the probability of miscoverage.
If we deal with the two tasks separately,
the criterion turns out to be the family-wise error rate (FWER).
In this case,
denote by $I_i$ the indicator of the miscoverage for the $i^{\mathrm{th}}$ task, so $I_i = 0$ for successful coverage and $I_i = 1$ for failure.
We thus aim to control the error rate $P(I_1 + I_2 > 0)$ at a given level $\alpha$.
Typical solutions 
would be a confidence region $C_\alpha=\{\theta \in x_1 \pm q_{\alpha/2}\}$ and a prediction region $T_\alpha=\{X_2 \in 2x_1 \pm \sqrt{2}q_{1-\alpha/2}\}$.
However, to control $P(I_1 + I_2 > 0)$ at a certain level $\alpha$, the following problems arise:
\begin{compactitem}
    \item 
    %
    Due to multiplicity, we need an additional correction to control the family-wise error rate, e.g., the Bonferroni correction.
    \item We may want to avoid the most stringent multiplicity corrections. 
    However, using the same data $x_1$ for both tasks may make this challenging. For instance, we have
    \begin{align*}
        &P(\theta \in C_\alpha, X_2 \in T_\alpha)\\
        =& \int p(X_1)[\Phi(X_1 + \sqrt{2}q_{\alpha/2}) - \Phi(X_1 + \sqrt{2}q_{1-\alpha/2})][\Phi(2X_1 + q_{\alpha/2}) - \Phi(2X_1 + q_{1-\alpha/2})] dX_1.
    \end{align*}
    This correlation due to $x_1$
    might make the joint probability even harder to compute in cases with more tasks.
\end{compactitem}

Instead of considering $I_1 = I(\theta \notin C_\alpha)$ and $I_2 = I(X_2 \notin T_\alpha)$ separately and aiming to control $P(I_1\neq0, \textnormal{ or } I_2\neq 0)$, 
we can consider the \emph{joint miscoverage indicator}
$I_{12}=I((\theta,X_2) \notin J_\alpha)$
for a joint coverage region $J_\alpha$, 
and aim to control 
the miscoverage rate $P(I_{12}\neq0)$.
This conforms with the structure of JCRs, which involve both the unknown parameter $\theta$ and the future observable $X_2$.

Specifically, we consider the JCR
\begin{align}\label{eq:jcr-toy}
    J_\alpha = \{(\theta,X_2): X_2-2\theta \in [\sqrt{2}q_{\alpha/2},\sqrt{2}q_{1-\alpha/2}]\},
\end{align}
which covers $\theta$ and $X_2$ simultaneously with $P(I_{12}\neq0) = P((\theta,X_2) \notin J_\alpha) \leq \alpha$.
Formally, 
since $X_1 - \theta \sim \N(0,1)$, $X_2 - X_1 -\theta \sim \N(0,1)$, we have the pivot $X_2 - 2\theta \sim \N(0,2)$, so that
$$
P((\theta,X_2) \notin J_\alpha) = P(X_2-2\theta \notin [\sqrt{2}q_{\alpha/2},\sqrt{2}q_{1-\alpha/2}]) = \alpha.
$$

Of course, we may also use the pivot $X_2 - X_1 - \theta \sim \N(0,1)$, defining
\begin{align}\label{eq:jcr-toy1}
J_\alpha'(x_1) = \{(\theta,X_2): X_2-x_1 - \theta \in [q_{\alpha/2},q_{1-\alpha/2}]\}.
\end{align}
This involves $X_1$  and is thus random, but has a shorter prediction component. 
We can further intersect the JCRs in \eqref{eq:jcr-toy}, \eqref{eq:jcr-toy1} with confidence regions for $\theta$ 
to obtain bounded JCRs. 
For instance, we can intersect $J_{\alpha/2}$ or $J_{\alpha/2}'$ with $C_{\alpha/2}= \{\theta \in x_1 \pm q_{\alpha/4}\}$ to yield a slightly conservative region with coverage rate over $1-\alpha$. 

In Figure \ref{fig:toy_jcr}, we show the regions $J_{\alpha}, J_{\alpha}', C_{\alpha/2}, T_{\alpha/2}$ as defined above,
as well as the intersections $J_{\alpha/2} \cap C_{\alpha/2}$, $J_{\alpha/2}' \cap C_{\alpha/2}$ as we described. 
Our JCR approach better captures problem structure.
To validate coverage, 
we take $\theta = 0$ and run $10,000$ independent trials. 
In each trial, we record the following events: 
$(\theta, X_2) \in C_{\alpha/2} \times T_{\alpha/2}$, $(\theta,X_2) \in J_{\alpha}$, $(\theta,X_2) \in J_{\alpha/2} \cap C_{\alpha/2}$.
We compute the coverage rates and their corresponding Clopper-Pearson CIs (CPCIs) for $\alpha = 0.1$. 
The coverage rates turn out to be $91.93\%$, $89.91\%$ and $90.23\%$ with their $95\%$-CPCIs $[91.38\%,\ 92.46\%]$, $[89.30\%,\ 90.49\%]$ and $[89.63\%,\ 90.81\%]$, respectively. 
As expected, the region $J_{\alpha/2} \cap C_{\alpha/2}$ is slightly conservative, 
while 
the intersection JCR $ C_{\alpha/2} \times T_{\alpha/2}$ is more so. 

\begin{figure}
    \centering
    \includegraphics[height = 7cm, width =7.8cm]{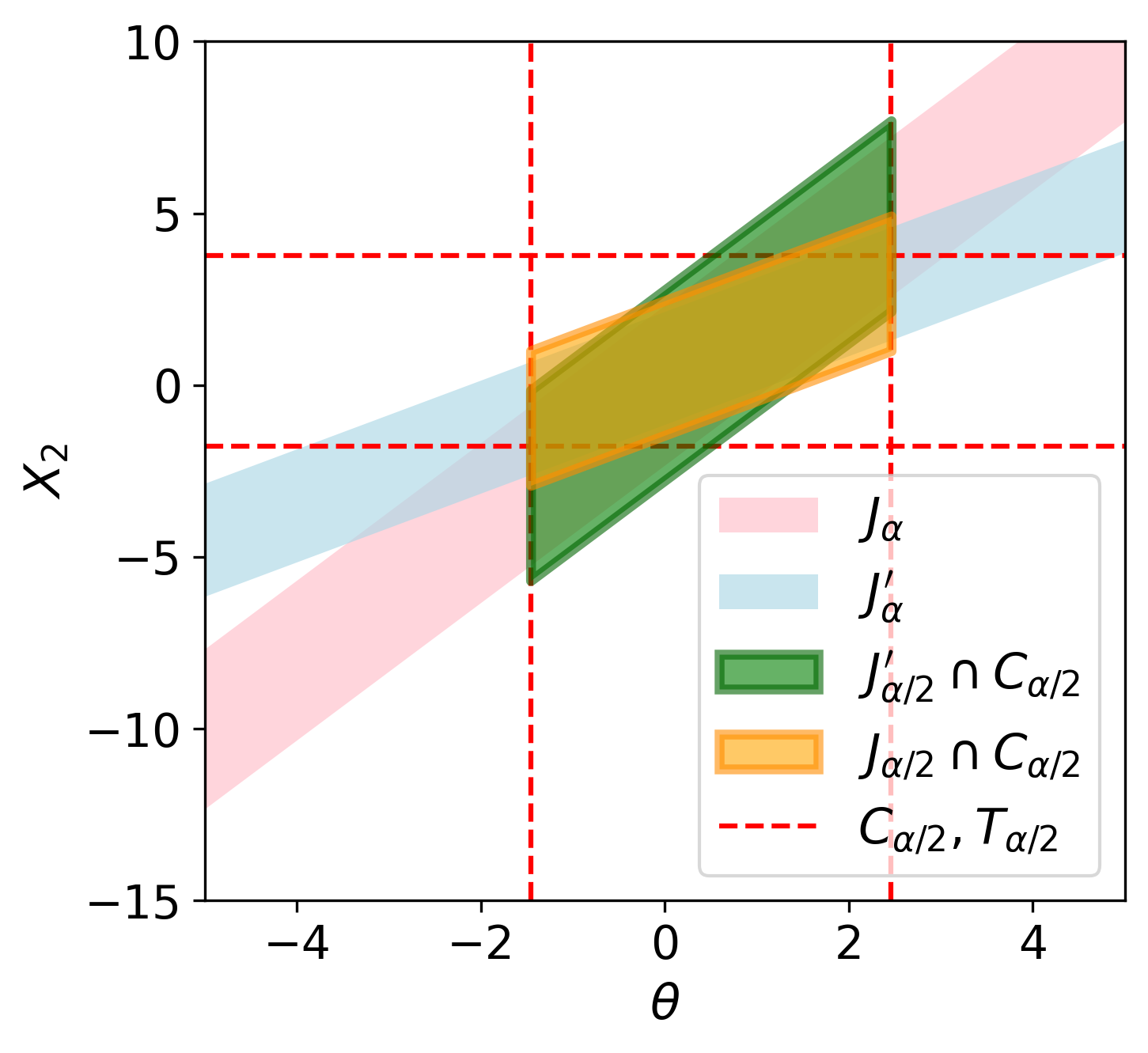}
\caption{A visualization for $J_{\alpha}, C_{\alpha/2}, T_{\alpha/2}$ as defined in Section \ref{sec:multiple-task}, for a single trial with $x_1 = 0.5$. }
\label{fig:toy_jcr}
\end{figure}


\section{Empirical Illustration}
\label{emp}

\subsection{Diabetes Data}
\label{diab}

We evaluate 
the JCRs from 
Section \ref{sec:comparison-lr} on the diabetes dataset used in \cite{efron2004least}, 
which reports ten variables of $442$ diabetes patients at baseline, 
as well as the response of interest, a quantitative measure of disease progression one year after baseline.
We consider the linear effect of body mass index (BMI)  on disease progression,  
centering both measurements.
We fit a linear model $y = x\theta + \ep$, 
where $\ep$ is independent noise, of disease progression $y$ on BMI $x$. 
See Figure \ref{fig:diabete-JCR-supp} in the Appendix for a plot.

As discussed, each JCR is valid under specific assumptions on the noise.
However, it is unclear which assumptions hold for this dataset. 
Moreover, the true value of the parameter $\theta$ for the linear effect is not known; making confidence statements hard to evaluate.
Therefore, we consider semi-empirical data to evaluate our methods. 

We randomly select a preliminary sample of $242$ 
measurements---denoted $X', Y'$---to derive a preliminary estimate $\htheta_{\mathrm{OLS}}^0$, via ordinary least squares.
We randomly select one datapoint 
from the remaining $200$, 
and use the others to construct JCRs.
We repeat the following experiment $1,000$ times: 
we randomly select one datapoint and use the features of the remaining datapoints and the preliminary estimated parameter $\htheta_{\mathrm{OLS}}^0$ to generate outcomes from a linear regression model with normal noise and approximate variance $S^2 = (Y'-X'\htheta_{\mathrm{OLS}}^0)^2/(n-1)$. Then, we construct the Gaussian pivotal, cyclic-shift-based, and permutation-based JCRs from \eqref{eq:lr-JCR}, \eqref{jcyc}, \eqref{pj} respectively using those data, with $\alpha = 0.05$.
The Gaussian and cyclic-shift based JCRs are computed in closed form.
For the permutation-based JCR, 
we randomize using $K=1,000$ transforms. 
Then, we evaluate the coverage of the JCRs on the test datapoint with outcomes generated using the same linear model.

The empirical coverages are $94.2\%$, $94.1\%$, and $94.3\%$ for the Gaussian pivotal, cyclic shift-based, permutation-based JCRs.
Their corresponding $95\%$-CPCIs are $[92.57\%$, $\,95.57\%]$, $[92.46\%,\,95.48\%]$  and $[92.68\%$, $\,95.65\%]$.
The results are consistent with valid coverage. 
A trial is shown in Figure \ref{fig:diabete-JCR} (right), where the three JCRs have different shapes, as in the simulation.
\begin{figure}[ht]
  \centering
    \includegraphics[height = 5cm, width =5.5cm]{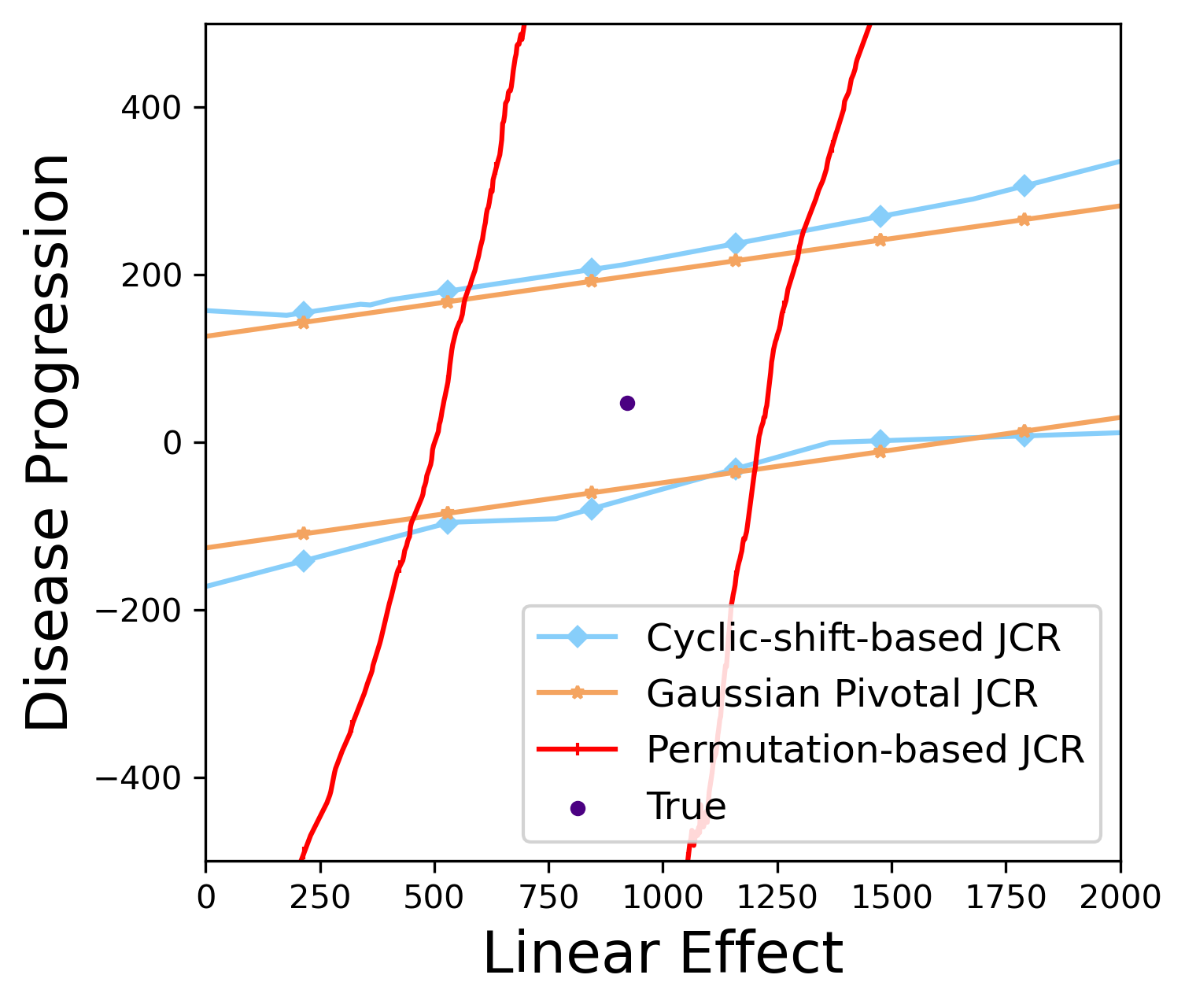}
  \caption{
  Diabetes dataset.
  We show the
  JCRs for the linear effect of BMI on disease progression, and the disease progression for a new patient given their BMI. 
  The purple point labeled ``True"
  shows $(\htheta_{\mathrm{OLS}}^0, \yte)$ for one test datapoint,
  with the progression level $\yte = 47.16$ of the new patient and the approximated linear effect $\htheta_{\mathrm{OLS}}^0 = 922.39$. 
  See also Section \ref{sec:diab-supp} in the Appendix for a scatterplot of the outcome and BMI for all 442 datapoints, with a least squares line.}
  \label{fig:diabete-JCR}
\end{figure}
Next, we illustrate methods for linear regression 
with ten features. 
We consider JCRs for the effect of BMI (a fixed parameter), 
as well as the disease progression outcome (a random variable). 
We use the same protocol as before.
The coverage of the JCR \eqref{eq:gau_high_d} 
is $95.1\%$ with its corresponding $95\%$-CPCI $[93.57\%, \,96.35\%]$, 
which is consistent with 95\% coverage.

\subsection{NYC Flight Delay Data}\label{nyc}
\begin{figure}[ht]
  \begin{minipage}{0.45\linewidth}
    \centering
    \includegraphics[height = 6cm, width =6.6cm]{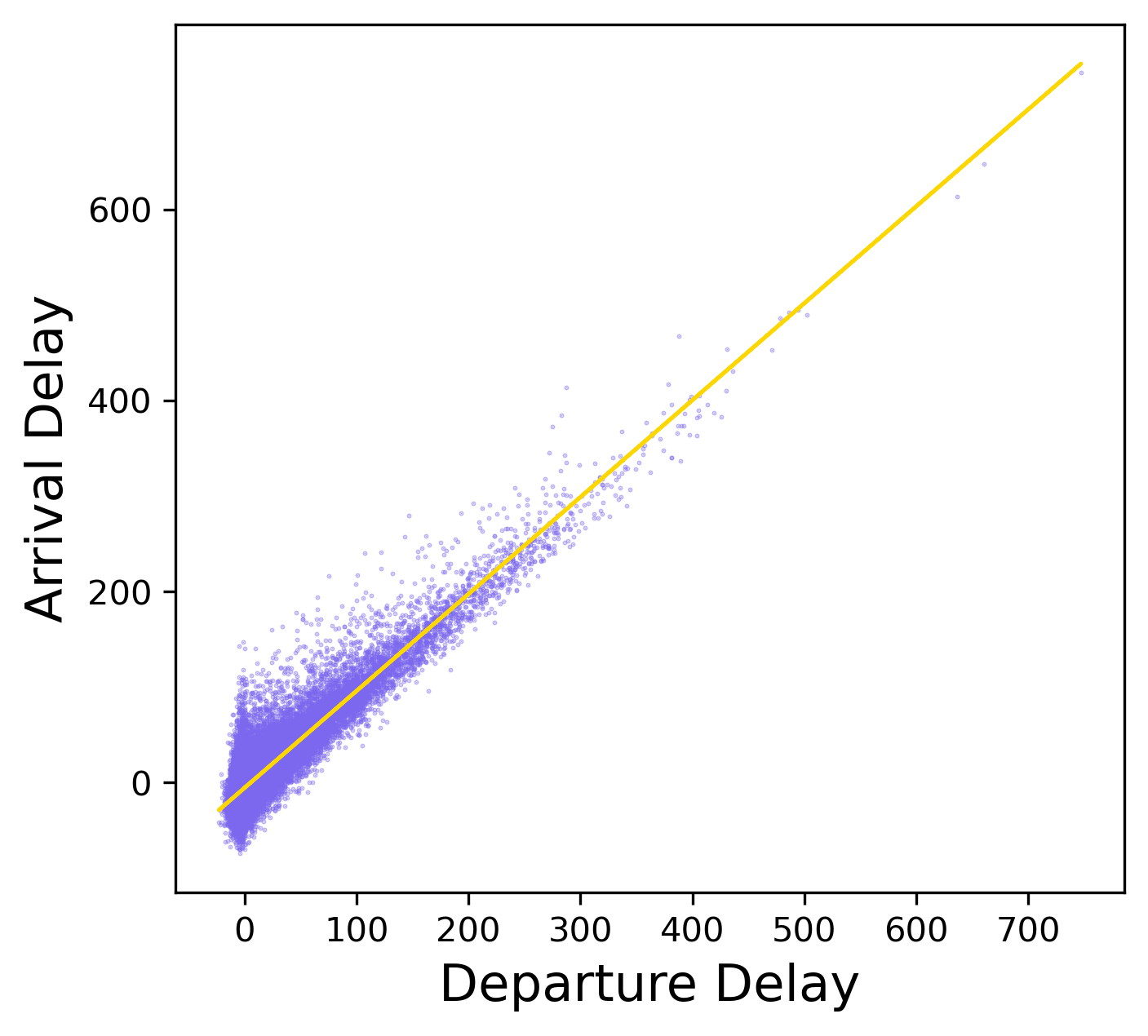}
  \end{minipage}
  \qquad
  \begin{minipage}{0.45\linewidth}
  \centering
    \includegraphics[height = 6cm, width =6.6cm]{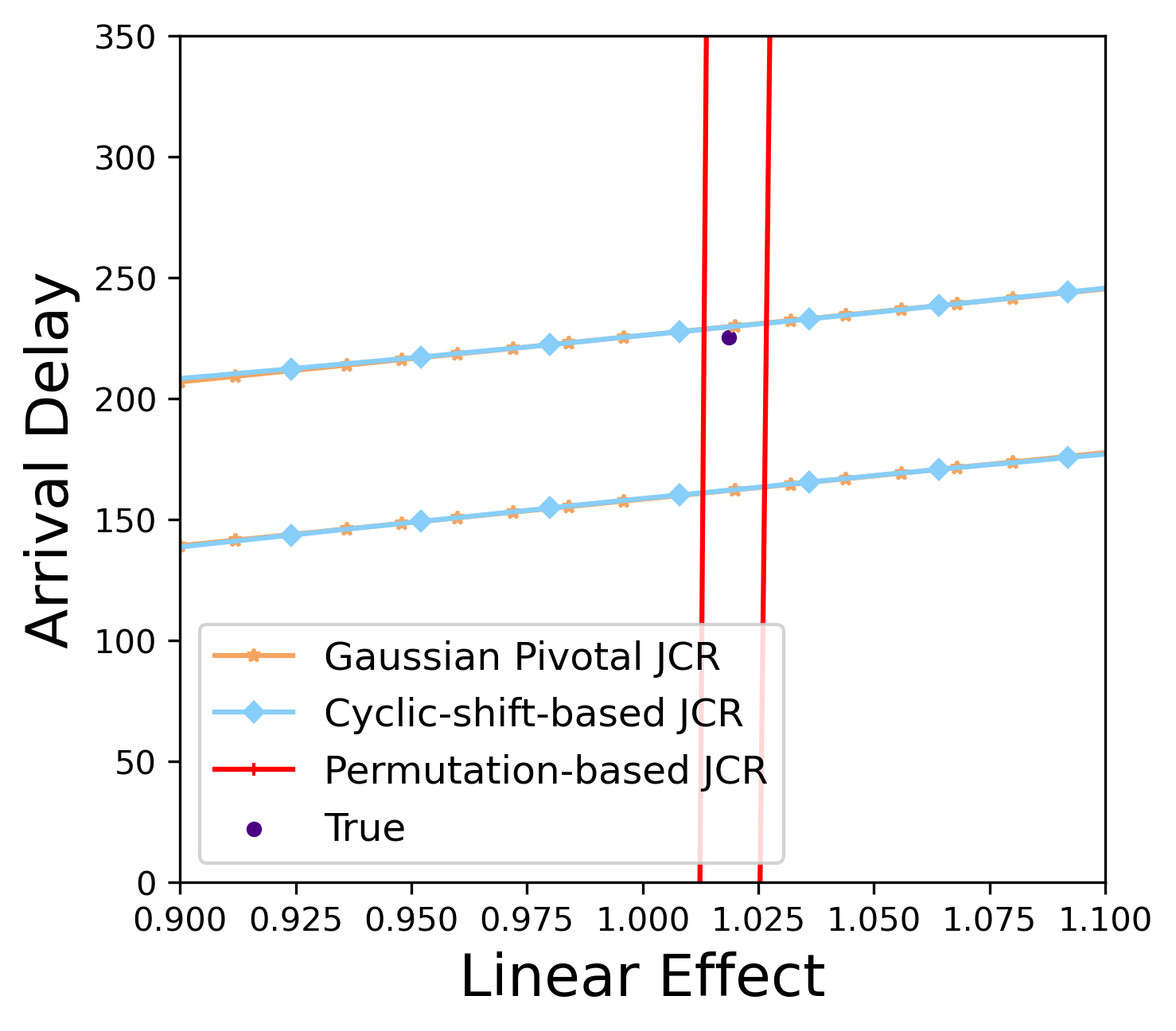}
\label{fig:nyc-JCR} 
  \end{minipage}
  \caption{NYC flights dataset.
  Left:
  Scatterplot of the arrival delay outcome and the departure delay feature with best-fit line.
  Right:
  JCRs for the linear effect of departure delay on arrival delay and the arrival delay for a new flight given its departure delay $\xte = 192.2$ (after centralized). The purple point labeled ``True"
  shows $(\htheta_{\mathrm{OLS}}^0, \yte)$ for one generated test datapoint,
  with the arrival delay $\yte = 225.3$ of the flight and the approximated linear effect $\htheta_{\mathrm{OLS}}^0 = 1.019$.}
  \label{fig:nyc}
\end{figure}
We evaluate the JCRs from Section \ref{sec:comparison-lr} on the NYC flight dataset \citep{wickham2018nycflights13}, 
which reports various features for $60,448$ flights, including a response of interest, the arrival delay of each flight. 
We follow the protocol from Section \ref{diab}, 
fitting a linear regression model
of arrival delay $y$ to departure delay $x$ on a randomly chosen half of the data; the centered data in Figure \ref{fig:nyc} shows a good linear relation. 


The empirical coverage is $94.8\%$, $94.6\%$, and $95.8\%$ for the Gaussian pivotal, cyclic shift-based, permutation-based JCRs with their corresponding $95\%$-CPCIs $[93.24\%,\,96.09\%]$, $[93.01\%$, $\,95.92\%]$ and $[94.36\%$, $\,96.96\%]$.
The results are consistent with valid coverage. A single trial is shown in Figure \ref{fig:nyc} (right).
Similarly, we fit a regression using all $14$ features, and construct a JCR for the effect of departure delay on arrival delay. 
The JCR \eqref{eq:gau_high_d}
has $94.9\%$ coverage with corresponding $95\%$-CPCI $[93.35\%,\,96.18\%]$, which is consistent with 95\% coverage.

\section*{Acknowledgements}


We are very grateful to Jacob Bien, Eugene Katsevich, Hua Su and Larry Wasserman for helpful discussions.
During this work, ED was supported in part by NSF award DMS 2046874 (CAREER).

\section{Appendix}\label{appendix}

\subsection{Connections between JCRs, Confidence Regions, and Prediction Regions}
\label{conn}
\begin{figure}
    \centering
    \includegraphics[height = 5.5cm, width =7cm]{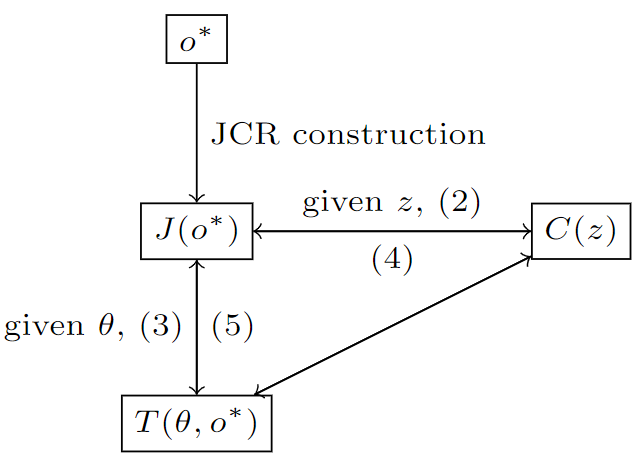}
    \caption{The relationship and transformations between the trio of regions.}
    \label{fig:trio}
\end{figure}
To aid our understanding of joint coverage regions, we now explain some of their connections to classical confidence and prediction regions. 
We first recall the classical definitions of confidence and prediction regions, as they arise in our framework.
Recall the setting from Section \ref{JCR-framework}:
the full data is $Z\sim P$, but we observe only $o(z)$.
The parameter of interest is $\theta$.

A $1-\alpha$-confidence region  for $\theta$ based on the observed data $o(z)$
is a map $\tilde C: \mO \to B_\Theta$ 
such that for all $P\in \mP$, $\bP_{Z \sim P}(\theta_P \in \tilde C(o(Z)) ) \ge 1-\alpha$.
However,
as we explain below, 
to understand JCR, it is helpful to consider a different, hypothetical, form of a confidence region, which is based on the generally unobserved full data $z$.

\begin{definition}[Full-data Confidence Region] \label{def-conf}
    We say that $C:\mZ\to B_\Theta$ 
    is a $1-\alpha$-full data confidence region for $\theta$ if for all $P\in \mP$,
    $\bP_{Z \sim P}(\theta_P \in C(Z) ) \ge 1-\alpha.$
\end{definition}
Of course, a full-data confidence region is in general not implementable, as 
we do not in general observe $z$. However, this theoretical notion will still be useful for understanding JCR, as it turns out they are in a one-to-one correspondence.

Further, if we observe the full dataset, so that the observable is $o(z)=z$, then a full-data confidence region $C$ for $\theta_P$ can be found from a JCR $J$ for $(\theta_P,Z)$ by dropping the second component.
In this case
we can also use $C$ to construct hypothesis tests for $\theta$, via the usual duality between testing and confidence regions: 
we reject the null hypothesis $H_0 : \theta = \theta^*$ when $\theta^* \notin C(z)$. 

To introduce the connection to prediction regions, we will temporarily need to consider a slightly different notion of full data; and we indicate this by a  ``$+$" superscript notation for all notions related to the full data.
In particular, consider full data $Z^+\sim P$, 
over a measurable set $\mZ^+$ with an associated sigma-algebra $B_{\mZ^+}$,
and consider an observation map $o:\mZ^+ \to \mO$.
Then, a map $\tilde T: \mO \to B_{\mZ^+}$
is a $1-\alpha$-prediction region for $Z^+$ based on $o(Z^+)$ if for all $P\in \mP$, $\bP_{Z^+ \sim P}(Z^+ \in \tilde T(o(Z^+)) ) \ge 1-\alpha$.

In principle, we can define $Z^+$ to be an arbitrary quantity that is associated with $P$, and thus we could also consider it to be the pair of the parameter and the observation $Z$  we have considered before, i.e., $Z^+ = (\theta(P),Z)$.
This is allowed by the formal definition of prediction regions; 
but is a bit unusual. 
Thus, formally, our notion of JCR can be viewed as an instance of standard prediction regions.
However, 
considering JCRs as we do here---and separating their coverage target into a deterministic parameter and a stochastic observable---leads a number of new insights, illustrated throughout our paper.
This supports that our JCR notion is a valuable addition to statistical methodology.

Returning to prediction regions, to understand JCRs, it is thus helpful to consider a different form of a prediction region, which can also depend on the generally unobserved parameter $\theta(P)$.

\begin{definition}[Parameter-Aware Prediction Region] \label{def-pred}
     We say that $T:\Theta\times\mathcal{O}\to B_\mZ$ is a $1-\alpha$-parameter-aware prediction region for $Z$ if for all $P \in \mP$,
    $\bP_{Z \sim P}(Z \in T(\theta_P,o(Z)) ) \ge 1-\alpha.$
\end{definition}

In general, a parameter-aware prediction region $T$ depends on the unknown parameter $\theta_P$, and is thus not practically implementable.
However, as before, it turns out that this notion is also useful for understanding JCR, as again they are in a one-to-one correspondence.
Further, if we only have a pure prediction problem, i.e., $\theta_P$ is a constant independent of $P$, 
then 
a parameter-aware prediction region becomes a usual prediction region.
Such a region can be constructed directly from a JCR by dropping the component in the $\Theta$ space.
There are important pure prediction examples, in particular in the area of conformal prediction.

Given a standard  $1-\alpha_1$-confidence region $\tilde C$ and $1-\alpha_2$-prediction region $\tilde T$, 
direct ways to define JCRs include 
$\tilde C(o) \times \mZ$ (a $1-\alpha_1$-JCR) and
$\Theta \times \tilde T(o)$ (a $1-\alpha_2$-JCR),
which however are informative in only one coordinate.
An alternative
is via the intersection
$
J(o) = \left(\tilde C(o) \times \mZ\right) \cap \left(\Theta \times \tilde T(o)\right),
$
which is a $1-(\alpha_1+\alpha_2)$ JCR. Indeed, 
$$
P(Z\in J(o(Z))) = P(\theta_P\in \tilde C(o(Z)) \textnormal{ or } Z \in \tilde T(o(Z)))
\ge 1-(\alpha_1+\alpha_2).
$$
However, this JCR \emph{does not take into account the relation} between the parameter and the data, and thus generally does not reflect the structure of the statistical problem. 
For instance, if the data $Z$ to be predicted  has the form $Z = \theta_P + \ep$ for some noise $\ep$, then we expect that a reasonable JCR could be a "band" in $\Theta\times \mZ$. This would capture the  relation between the parameter and the data.

With Definitions \ref{def-conf} and \ref{def-pred},
we can construct a
full-data    confidence region $C$ from a JCR $J$ by defining $C(z)$, for all $z\in\mZ$, as 
    \begin{equation}\label{eq:J-C}
     C(z) =\{\theta\in \Theta : 
      (\theta,z) \in J(o(z)) \}.
    \end{equation}
    
Equivalently, $C(z) =\cup_{\theta \in \Theta} \{\theta:\, (\theta,z) \in J(o(z)) \}$, or more abstractly
$C(z) =\Phi_{\Theta}[J(o(z)),z]$.
We can also write 
$C(z) =\Pi_{\Theta}[ J(o(z)) \cap (\Theta\times \{z\})]$.
See Figure \ref{fig:J-C-T} for an illustration of this and the following constructions.

We can also   
construct a parameter-aware prediction region $T$ by defining $T(\theta,o^*)$, for all $\theta\in\Theta$, $o^*\in\mathcal{O}$, as
\begin{equation}\label{eq:J-T}
T(\theta,o^*) = \{z\in\mZ: 
o(z) = o^*,\,
(\theta, z) \in J(o^*)\}.
\end{equation}
More abstractly,
$T(\theta,o^*) = \Phi_{\mZ}[J(o^*),\theta] \cap o^{-1}(o^*)$; see Figure \ref{fig:J-C-T}.
Further, we can construct a JCR $J$ based on a full-data confidence region $C$ via  
\begin{equation}\label{eq:C-J}
J(o^*) = \{(\theta, z)\in\Theta\times\mZ:
o(z) = o^*,\,
\theta \in C(z)\}. 
\end{equation}
More abstractly, $J(o^*) = \cup_{z\in o^{-1}(o^*)} \left(C(z) \times \{z\}\right)$. See Section \ref{proof:measurability} for conditions under which this construction leads to a measurable function $J$.
Finally, we can construct a JCR $J$ based on a  parameter-aware prediction region $T$ by 
\begin{equation}\label{eq:T-J}
J(o^*) = \left\{(\theta,z) 
\in\Theta\times\mZ: 
o(z) = o^*,\,
z \in T(\theta,o^*) \right\}.
\end{equation}
More abstractly, $J(o^*) = 
\cup_{\theta\in\Theta} \{\theta\} \times [T(\theta,o^*) \cap o^{-1}(o^*)]$. 
See Section \ref{proof:measurability} for conditions under which this construction leads to a measurable function $J$.
\begin{figure}[ht]
    \centering
    \includegraphics[height = 6cm, width =6.81cm]{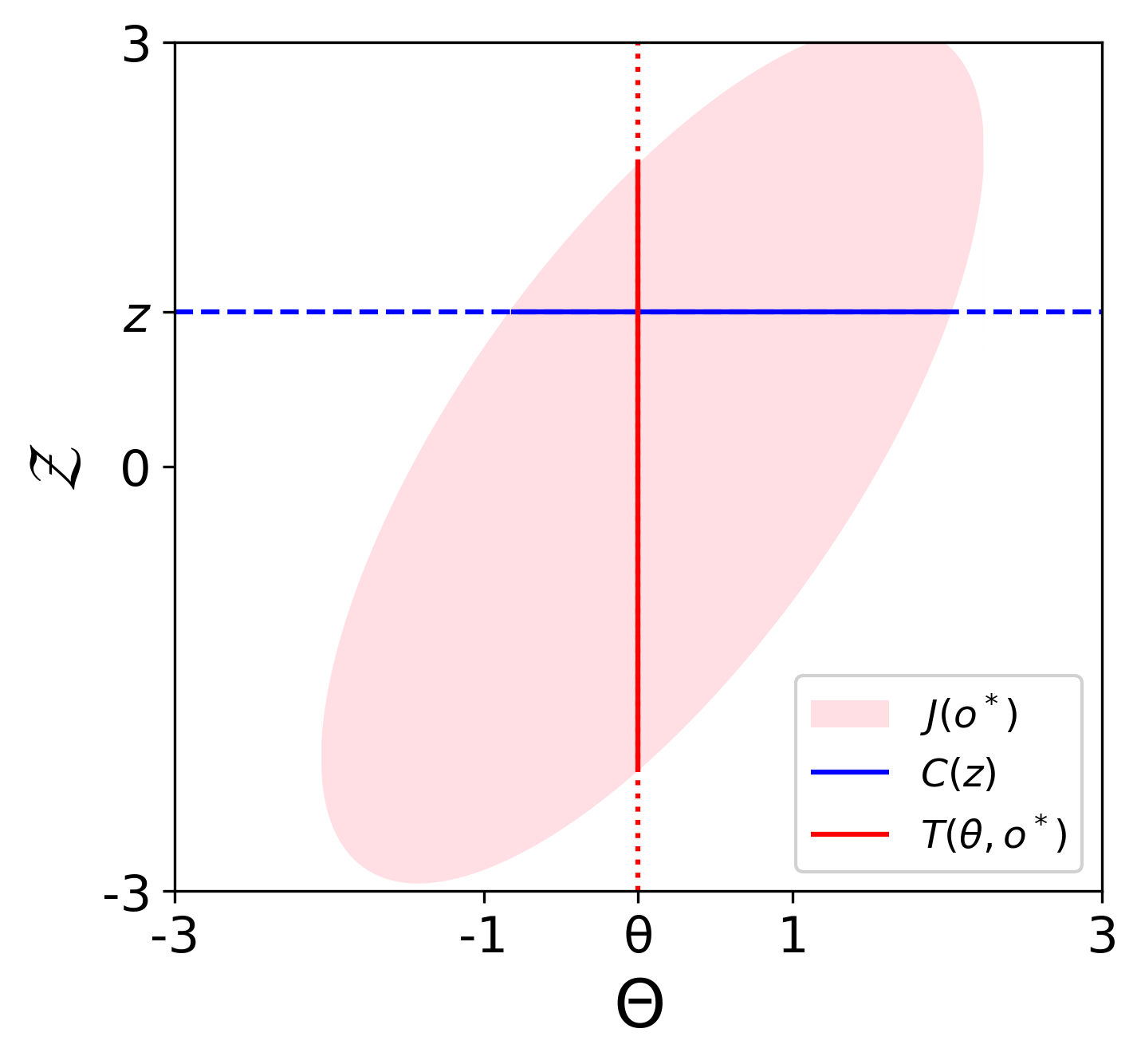}
    \caption{Visualizing the correspondences within the trio of regions. 
    The construction \eqref{eq:J-C}, \eqref{eq:J-T}, \eqref{eq:C-J}, \eqref{eq:T-J} can then be understood in the natural way. For instance, Equation \eqref{eq:J-C} can be viewed as taking a section of $J$ over $z \in \mZ$. Conversely, by considering \eqref{eq:C-J}, we merge the region $C(z)$ for all valid $z \in \mZ$ that satisfy $o(z) = o^*$. Equation \eqref{eq:J-T}, \eqref{eq:T-J} can also be understood in a similar way.}
    \label{fig:J-C-T}
\end{figure}

The following result shows that these operations are inverses:
\begin{lemma}\label{equivalance}
We have the following:
\begin{enumerate}
    \item Given any region $J$, construct the region $C$ using \eqref{eq:J-C} and then construct the region $\tilde J$ using \eqref{eq:C-J}. Then $\tilde J = J$.
    \item Given any region $J$, construct the region $T$ using \eqref{eq:J-T} and then construct the region $\tilde J$ using \eqref{eq:T-J}. Then $\tilde J = J$.
\end{enumerate}
\end{lemma}

\begin{proof}
For the first claim, fix $o^*\in \mathcal{O}$. Suppose that $(\theta,z)\in J(o^*)$ and $o^*=o(z)$. Then, since $\theta \in \Phi_\Theta(J(o(z))$, by \eqref{eq:J-C} we have that $\theta \in C(z)$; or equivalently $(\theta,z)\in C(z)\times \{z\}$. 
Hence, by \eqref{eq:C-J} it follows that $(\theta,z)\in \tilde J(o(z)) = \tilde J(o^*)$. This shows that $J(o^*) \subset \tilde J(o^*)$. 
Similarly, suppose that $(\theta,z)\in \tilde J(o^*)$. 
Then, since $o^*=o(z)$, by \eqref{eq:C-J} we have $(\theta,z)\in C(z)\times \{z\}$, or equivalently $\theta \in C(z)$. Thus, by \eqref{eq:J-C} it follows that $\theta \in \Phi_\Theta(J(o(z))$, and thus $(\theta,z)\in J(o^*)$.
Since these claims hold for all $o^*\in \mathcal{O}$, it follows that $J = \tilde J$.

The proof of the second claim is similar. Fix $o^*\in \mathcal{O}$. 
Suppose that $(\theta,z)\in J(o^*)$ and $o^*=o(z)$. 
Then, since $z\in \Phi_{\mZ}[J(o^*)] \cap o^{-1}(o^*)$, by \eqref{eq:J-T} we have that $z \in T(\theta,o^*)$. Hence, by \eqref{eq:T-J} it follows that $(\theta,z)\in \tilde J(o(z)) = \tilde J(o^*)$. This shows that $J(o^*) \subset \tilde J(o^*)$. 
Similarly, suppose that $(\theta,z)\in \tilde J(o^*)$. Then, since $o^*=o(z)$, by \eqref{eq:T-J} we have $z \in T(\theta,o^*)$. Thus, by \eqref{eq:J-T} it follows that $(\theta,z)\in J(o^*)$.
Since these claims hold for all $o^*\in \mathcal{O}$, it follows that $J = \tilde J$.
\end{proof}

This shows that the three regions are in a one-to-one correspondence. We call such a triple $(J,C,T)$ a \emph{trio of regions}.
\begin{definition}[Trio of Regions]\label{trio}
We say that $(J,C,T)$ are a trio of regions if they satisfy \eqref{eq:J-C}, \eqref{eq:J-T}, \eqref{eq:C-J} and \eqref{eq:T-J}.
\end{definition}

The relationship between the elements of a trio is shown in Figure \ref{fig:trio}. 
We also have the following result:
\begin{lemma}\label{J-RC}
 Given a $1-\alpha$ JCR $J$, the region  $C$ from \eqref{eq:J-C} is a $1-\alpha$ confidence region, and the region $T$ from \eqref{eq:J-T} is a $1-\alpha$ prediction region.
\end{lemma}

\begin{proof}
    For a given $1 - \alpha$ JCR $J$, from \eqref{eq:J-C} we have $\theta \in C(z)$ is equivalent to $(\theta,z) \in J(o(z))$. Combining this with \eqref{eq:J} we have:
    \begin{align*}
        \bP_{Z \sim P}\biggl(\theta_P \in C(Z)\biggr) = \bP_{Z \sim P}\biggl(\left(\theta_P, Z\right) \in J\left( o(Z)\right)\biggr) \ge 1 - \alpha,
    \end{align*}
    which shows that $C$ is a $1 - \alpha$ confidence region.
    
    Similarly, for the second claim, we know from \eqref{eq:J-T} that $T(\theta,o^*)$ includes all $z$ that satisfies $o(z) = o^*$ and $(\theta,z) \in J(o^*)$. Specifically, the first condition will always be satisfied when $o^* = o(z)$. Combining this with \eqref{eq:J} we have:
    \begin{align*}
        \bP_{Z \sim P}\biggl(Z\in T(\theta_P,o(Z))\biggr) &= 
        \bP_{Z \sim P}\biggl(o(Z) = o(Z),\left(\theta_P, Z\right) \in J\left( o(Z)\right)\biggl) \\
        &=\bP_{Z \sim P}\biggl(\left(\theta_P, Z\right) \in J\left( o(Z)\right)\biggr)
        \ge 1 - \alpha,
    \end{align*}
    which shows that $T$ is a $1 - \alpha$ prediction region.
\end{proof}
Combined with the previous result, this shows the following corollary:

\begin{corollary}\label{RC-J}
We have the following:
\benum
    \item Given a $1-\alpha$ confidence region $C$, the region  $J$ from \eqref{eq:C-J} is an $1-\alpha$ JCR.
    \item Given a $1-\alpha$ prediction region $T$, the region  $J$ from \eqref{eq:T-J} is an $1-\alpha$ JCR.
\eenum
\end{corollary}
Finally, we conclude that JCRs, full-data confidence regions, and parameter-aware prediction regions are in a one-to-one correspondence.

We also explain the connection between pivotal JCRs and classical pivotal constructions of confidence and prediction regions.
Consider the pivotal JCR from \eqref{pivjcr}.
In the trio of regions from Definition \eqref{trio}, the associated confidence region for $\theta$ is the classical confidence region based on the pivot $L$:
 $C(z) = \{\theta \in \Theta : L(\theta,z) \in S\}.$
This shows that pivotal JCRs and pivotal confidence regions are in a one-to-one correspondence.

\subsection{When do Pivots Exist?}
\label{pive}
Here we review conditions for statistical models under which pivots exist, to illustrate the range of problems to which JCRs apply. 
See e.g.,
\cite{fraser1966structural,fraser1968structure,fraser1971events,brenner1983models,barnard1995pivotal,fraser1996some},
Section 7.1.1 of \cite{shao2003mathematical}
for references on pivotal variables.
Standard confidence regions with finite sample coverage usually 
require the existence of pivots,
and thus our methods are typically applicable whenever standard confidence regions can be constructed.

As mentioned in the main text, pivots exist for any parametric statistical model with independent continuously distributed scalar observations (Proposition 7.1 of \cite{shao2003mathematical}).
Specifically, suppose that
for some $a\ge 1$,
$Z = (Z_1, \ldots, Z_a)$, where $Z_a \in \R$ are independent scalar random variables with continuous distributions $Z_a\sim F_{\theta_i(P)}$. 
Then,
for $\theta(P) = (\theta_I(P),\ldots,\theta_a(P))$,
and for any measurable function $\tau:[0,1]^a\to \R$,
$L(\theta(P),Z) = \tau(F_{\theta_I(P)}(Z_1),\ldots,F_{\theta_a(P)}(Z_a))$
is a pivot.

Another example is provided by injective data generating models, which are often referred to as structural or structured models \citep{fraser1966structural,fraser1968structure,fraser1971events,brenner1983models,fraser1996some}. 
Suppose $Z = f_{\theta_P}(\ep)$, where $\ep$ is noise with a fixed distribution $Q$ over some measurable space $E$, and for all $P\in\mP$, $f_{\theta_P}:E\to \mZ$ is injective. 
Then, having observed $Z=z$, we can write equivalently that $f_{\theta_P}^{-1}(z)=\ep$,
where $f_{\theta_P}^{-1}(z) \in E$ is the unique value such that $f_\theta(f_{\theta_P}^{-1}(z))=z$.
Thus, $L(\theta,Z) = f_\theta^{-1}(Z)\sim Q$ 
is a pivotal random variable.
A key example is group invariance models or structural models \citep{fraser1968structure}, where for some group $\mH$, and injective group action $h $, $Z$ follows the model $Z = h \ep$. Then, $h^{-1}Z = \ep$ is a pivotal random variable.
Classical examples include location-scale families and data with sign-symmetric or spherically distributed noise.


To illustrate the breadth of these models, we discuss
the example of 
Gaussian linear mixed effects models
$Y = X\beta + W\gamma +\ep,$
where $Y$ is the $n\times 1$ vector of outcomes, $X$ is the $n\times p$  matrix of  features with deterministic effects, 
$W$ is the $n \times p'$  matrix of features with random effects, 
$\beta$ is the $p\times 1$ vector of fixed effects,
$\gamma \sim \N(0,\Gamma)$ is the $p'\times 1$ vector of random effects 
and $\ep\sim \N(0,\sigma^2 \Sigma)$ is the random noise.
Here $\Sigma$ is assumed known.
There are a wide range of special cases, such as various ANOVA models.
We may consider $\Gamma $ and $\sigma^2$ known or unknown.
Then, the model is equivalent to 
$$Y = X\beta + (W\Gamma W^\top+\sigma^2 \Sigma)^{1/2}\ep',$$
for some noise $\ep'\sim \N(0,I_n)$.
If $X,W$ are observed, 
this can be viewed as an injective generative model with $\theta = (X\beta,  (W\Gamma W^\top+\sigma^2 \Sigma)^{1/2})$, and $f_\theta(\ep')$ as displayed above.
Then, consistent with injective generative models,
$$L = (W\Gamma W^\top+\sigma^2 \Sigma)^{-1/2}(Y -X\beta)\sim \N(0,I_n)$$
is a pivot.
Moreover, $L$ is still a pivot even if---some parts of---$X,W$ are not observed. 
This is related to the setting of inverse regression \citep{williams1959regression,krutchkoff1967classical}, where part of $X$ is unobserved.

\subsection{Considerations}
\label{cons}
Here we discuss several crucial considerations for constructing pivotal JCRs.


{\bf Discreteness.}
If the distribution of the pivot $L$, taking values in $\R^m$ for some $m\ge 0$, is not absolutely continuous with respect to the Lebesgue measure,
there may not exist a set $S$ such that $Q(S)=1-\alpha$. 
This can be resolved by considering randomized decision rules $\phi: \mathcal{L} \to [0,1]$, such that we include $l\in \mathcal{L}$ in the region with probability $\phi(l)$. Then, we can find $\phi$ such that $\E_{L\sim Q} \phi(L(\theta,Z)) = 1-\alpha$.
A randomized JCR includes $(\theta,Z)$ in the region with probability $\phi(L(\theta,Z))$. Clearly, this region has exact $1-\alpha$ coverage. 
In this work, we mainly consdider deterministic JCRs.


{\bf Asymptotic pivots.} 
We can obtain asymptotic coverage given 
a sequence of 
asymptotically pivotal random variables.
We consider an 
asymptotic setting 
where
all quantities 
are indexed by 
an index
$n \in \mathbb{N}_+$.
Thus, 
there is a sequence of statistical models 
$(\mP_n)_{n\ge 1}$,
a sequence of probability distributions
$(P_n)_{n\ge 1}$, 
observations $(z_n)_{n\ge 1}$, etc.
Suppose that we have a random variable
$(L_n)_{n\ge 1}$,
$L_n:\Theta_n\times\mZ_n \to \mL$, 
for some fixed measurable space $\mL$ that does not depend on $n$.
Suppose
that 
when $Z_n \sim P_n$, 
$L_n(\theta_n(P_n),Z_n)$
has distribution 
$(Q_n)_{n\ge 1}$, 
which may depend on $P_n$.
Suppose
that $L_n$ is an asymptotic pivot in the sense that
the limiting distribution $\lim_{n 
\to \infty}Q_n = Q$ 
exists and
does not depend on 
the sequence
$(P_n)_{n\ge 1}$.

Let $S\subset \mL$ be a measurable set such that $Q(S)\ge 1-\alpha$. 
Then, we can 
 construct an asymptotic
 $1-\alpha$-JCR for $(\theta_n,Z_n)$ 
 via
\beq\label{pivjcra}
J_n(o_n^*)= \left\{(\theta_n,z_n) 
\in\Theta_n\times\mZ_n: 
o_n(z_n) = o_n^*,\,
L_n(\theta_n,z_n) \in S \right\}.
\eeq

\begin{corollary}
Suppose that 
$\liminf_{n\to\infty} Q_n(S) \ge Q(S)$.
Then
equation \eqref{pivjcra}
returns an asymptotically $1-\alpha$-joint coverage region in the sense that
\begin{equation*}
   \liminf_{n \rightarrow \infty} \bP_{Z_n \sim P_n}\biggl(\left(\theta_n(P_n), Z_n\right) \in J_n\left( o_n(Z_n)\right)\biggr) \ge 1 - \alpha.
\end{equation*}
\end{corollary}

\subsection{Proofs}

\subsubsection{Measurability}\label{proof:measurability}
We provide conditions under which the constructions from Section \ref{conn} are measurable.
For $z\in \mZ$ and  
$J \subset  \Theta\times \mZ$,
recall that
$\Phi_\Theta(J,z) =\cup_{\theta \in \Theta} \{\theta:\, (\theta,z) \in J \}$. 
Given $z\in \mZ$,
if $J'=J(o(z)) \in B_{\Theta \times \mZ}$ is measurable, we aim to prove that $\Phi_\Theta(J',z)$  is $B_{\Theta}$-measurable.
To see this, we will show that 
$B' := \{J: J \subseteq B_{\Theta \times \mZ}, \Phi_\Theta(J,z) \subseteq B_{\Theta} \} = B_{\Theta \times \mZ}$.
    
    First, we show that $B'$ is a sigma-algebra. 
    Since $\Phi_\Theta(\Theta \times \mZ,z) = \Theta \in B_\Theta$,
    we have that $\Theta \times \mZ \in B'$. 
    For $J_1,J_2,\ldots,J_n,\ldots \in B'$, we have that $\Phi_\Theta(\cup_{i=1}^\infty J_i,z) = \cup_{i=1}^\infty \Phi_\Theta(J_i,z) \in B_{\Theta}$, thus $\cup_{i=1}^\infty J_i \in B'$. 
    In addition, for $J \in B'$, we have $\Phi_\Theta(J^c,z) = (\Phi_\Theta(J,z))^c \in B_\Theta$, thus we find $J^c \in B'$. 
    Thus $B'$ is a sigma-algebra.
    
    Now, for any set $J = 
    D_\Theta \times D_\mZ \in 
    B_\Theta \times B_\mZ$,
    we have that $\Phi_\Theta(J,z) = D_\Theta$.
    Thus, $B_\Theta \times B_\mZ \subseteq B'$. 
    Since $B'$ is a sigma-algebra, we have that $B_{\Theta \times \mZ} = \sigma(B_\Theta \times B_\mZ) \subseteq B'$; i.e., the 
    sigma-algebra generated by $B_\Theta \times B_\mZ$ is a sub-sigma algebra of $B'$. 
    Combined with $B' \subseteq B_{\Theta \times \mZ}$, which holds by definition, we find that $B' = B_{\Theta \times \mZ}$, which shows that $\Phi_\Theta(J',z)$ is measurable for $J' \in B_{\Theta \times \mZ}$.

For 
$T(\theta,o^*) = \cup_{Z \in \mZ}\{Z: (\theta, Z) \in J(o^*)\} \cap o^{-1}(o^*) $, the first term in the intersection  is
measurable due to an argument similar 
to the one above.
If $o$ is a measurable map
and the singleton  $\{o^*\}$ belongs to the sigma-algebra $B_\mO$, 
the second term is also measurable.

\subsubsection{Proof of Theorem \ref{thm:conditional-JCR-piv}}\label{proof:conditional-JCR-piv}

   We have $L \sim Q_v$ conditionally on $V(\theta,z)=v$
   for $P_V$-almost every $v\in\mV$.
   For such $v$, we have 
   $
   \bP[L(\theta,z) \in S(v)|V(\theta,z)=v] \ge 1-\alpha.
   $
   Thus, 
   $$
      \bE[\bE[I((\theta,z) \in J(o(z)))|\, V(\theta,z)=v]] \ge \bE_{P_V}(1-\alpha) = 1-\alpha.
   $$
   Hence, \eqref{jcp} returns a $1-\alpha$ JCR and this finishes the proof.

\subsubsection{Proof of Proposition \ref{prop:gsm}}\label{proof:gsm}

Since the mapping $\psi: E \times \mathcal{V} \to \mathcal{L}$ is fixed and $\ep$ has a fixed distribution, 
when we condition on $V(\theta_P,Z)=v$ for arbitrary $v \in \mathcal{V}$, we
find $\psi(\ep,v)\sim Q_v$
for some $Q_v$ determined 
by $v$, the distribution of $\ep$,
and $\psi$. 
Thus, for any $P \in \mP$ and for $P_V$-a.e. $v$, conditionally on $V(\theta_P,Z)=v$, $L(\theta_P,Z)\sim Q_v$, 
which shows it is a conditional pivot.

\subsubsection{Proof of Theorem \ref{thm:conditional-JCR}}\label{proof:conditional-JCR}
  
Since $L\sim Q_v$ conditionally on $V(\theta_P,Z)=v$, for $P_V$-almost every $v \in \mathcal{V}$,
$m(L(\theta,Z)) \sim  m(Q_{v})$,
conditionally on $V(\theta_P,Z)=v$, for $P_V$-almost every $v \in \mathcal{V}$. 
For these $v \in \mV$, 
$P((\theta,Z) \in J(o(Z))) \ge 1-\alpha$ conditionally on $V(\theta_P,Z)=v$, 
from 
\eqref{def:conditional-JCR}
and
the definition of quantiles.
Consider the sigma-algebra $B_\mV'$ generated by $\{(\theta,z) : V(\theta,z) = v\}$, for $v\in \mV$.
Since $V$ is $B_{\Theta \times \mZ} \to B_\mV$ measurable,
$B_\mV' \subset B_{\Theta \times \mZ}$.
Since the conditional guarantee holds for $P_V$-almost every $v \in \mathcal{V}$,
    we find
        \begin{align*}
        \bE\,\bE　\left[ 1\big\{m(L(\theta,z)) \ge  q_{\alpha}(m(Q_v))\big\}|\, B_\mV' \right] \ge 1 - \alpha,
    \end{align*} 
which finishes the proof.
    
\subsubsection{Proof of Theorem \ref{thm:conditional-pivot-randomization}}\label{proof:conditional-pivot-randomization}

Since $m(L(\theta,z))\sim m(Q_v)$ conditionally on $V(\theta_P,Z)=v$, 
for $P_V$-almost every $v \in \mathcal{V}$, 
conditioning any $v \in \mV$ in this set, 
$m(L),M_{1:K}$ are i.i.d.~random variables. 
Then we have $P(m(L) \ge q_{\alpha'}(M_{1:K})) \ge 1- \alpha$,
see e.g., Chapter 11 in \cite{vovk2022algorithmic}. Hence, similarly to the proof of Theorem \ref{thm:conditional-JCR}, 
we find
$\bP_{Z;M_{1:K} \sim m(Q_{V(\theta,Z)})^K}\left((\theta_P,Z) \in J_{M_{1:K}}(o(Z))\right) \ge 1 - \alpha$,
which finishes the proof.

\subsubsection{Proof of Theorem \ref{thm:adequate-set}}\label{proof:adequate-set}

Due to \eqref{eq:W-J}, 
we have $$P((\theta,\Yte) \in J(\Zc,\Xte)) = P(A(\theta,\Xte,\Yte) \in W(\theta,\Zc)) = P((\theta,\Xte,\Yte) \in \tilde{J}(\Zc)).$$
 In addition, from arguments similar to those in Section \ref{proof:conditional-JCR}, we conclude that $P((\theta,\Xte,\Yte) \in \tilde{J}(\Zc)) \ge 1-\alpha$. Combining this with the equation above, we find 
 $P((\theta,\Yte) \in J(\Zc,\Xte))\ge 1-\alpha$, which finishes the proof.

\subsubsection{The Group Invariance Property}\label{sec:uniform}


For the uniform measure $U$ on $\mG$, 
and for some fixed $i\in \mI$
we let $G \sim U$ and $Gi$ be a random variable over $\mI$.
For a Borel set $B \in \mI$, 
we have,
for a distribution $\mu_i$ on $\mI$,
$
\mu_i(B):= \textnormal{Prob}(GI \in B) = U\{g:gi \in B\}.
$
We claim that $\mu_i$ is $\mG$-invariant.
Indeed, since $Gi\sim\mu_i$, we have for any $g \in \mG$
    that
    $g(Gi)\sim g\mu_i$.
    Since
    $(gG)i =_d Gi$,
    it follows that $\mu_i = g \mu_i$.

Taking an average over $I$
with respect to its distribution $P_I$,
we then find
that 
the distribution 
$P_I'$ of $GI$, defined by 
$
    P_I' = \int \mu_I P_I(dI)
$
is also $\mG$-invariant,
with 
$P_I'(B) = P_I'(gB)$ 
for any $B \in B_\mI$ and $g \in \mG$.
Thus,
letting $U_{O_I}$ be the $\mG$-invariant measure on $O_I$ defined by $P_I'$, 
we see that $I$ is uniform conditional on its orbit $O_I$, 
with distribution $U_{O_I}$ induced by the distribution $P_I'$ of $GI$ when $G \sim U$. 

\subsubsection{Proof of Theorem \ref{full-group-JCR}}\label{proof:full-group-JCR}
For a finite group $\mG = \{g_{1:K}\}$ with $|\mG|=K$, we denote the rank of $m(g_iI(\theta_P,Z))$ in $\{m(I)\}_{I \in O_I}$ by $R_i$:
$$R_i
    = \sum_{u=1}^K I
[ m(g_iI(\theta_P,Z)) \ge m(g_uI(\theta_P,Z)) ] + 1.
$$
Since the left coset of $\mG$ under $g_1$ is $\mG$, 
for any $j \in [K]$, there exists $l \in [K]$ such that $g_1g_l = g_j$. 
Since $I(\theta_P,Z) =_d g_jI(\theta_P,Z)$ for any $k \in [K]$, we have
\begin{align*}
    P_Z(R_1 = k)
    &= P_Z\left(\sum_{u=1}^K I
[ m(g_1I(\theta_P,Z)) \ge m(g_uI(\theta_P,Z)) ] = k-1\right)\\
    &= P_Z\left(\sum_{u=1}^K I
[ m(g_1g_lI(\theta_P,Z)) \ge m(g_ug_lI(\theta_P,Z)) ] = k-1\right) \\
    &= P_Z\left(\sum_{u=1}^K I
[ m(g_jI(\theta_P,Z)) \ge m(g_ug_lI(\theta_P,Z)) ] = k-1\right).
\end{align*}
Since $\{g_u g_l\}_{u\in[K]} = \{g_v\}_{v\in[K]}$,
\begin{align*}
    P_Z(R_1 = k)
    =& P_Z\left(\sum_{v=1}^K I
[ m(g_jI(\theta_P,Z)) \ge m(g_vI(\theta_P,Z)) ] = k-1\right) 
    = P_Z(R_j = k),
\end{align*}

Next, we first suppose that ties happen with zero probability.
Thus, $\{R_1,\ldots,R_K\} = [K]$
and $\sum_{i=1}^K P_Z(R_i = k) = 1$;
so that 
$P_Z(R_i = k) = 1/K$
for all $k \in [K]$. 
Hence, we obtain
\begin{align}\label{fe}
    P_Z\left(m(g_1I(\theta_P, Z)) \ge q_{\alpha'}
\left(\frac{1}{K}\sum_{i=1}^K \delta_{m(g_iI(\theta_P, Z))}\right)\right) \ge 1 - \alpha,
\end{align}
where $\alpha' = {\lfloor K\alpha\rfloor}/{K}$. Thus, for any $P \in \mP$, 
we have
$
P_{Z \sim P}((\theta_P, Z) \in J(o(Z))) \ge 1 - \alpha$,
which finishes the proof.

When ties can happen, we claim that $P_Z(R_i \ge k)$ does not decrease compared to the case without ties. Formally, we consider randomized test statistics $\tilde{m}_i$, $i\in[K]$, defined by $\tilde{m}_i(g_iI(\theta_P, Z)) = m(g_iI(\theta_P, Z)) + \xi_i$,
where $\xi_i$-s are independent random variables with $\xi_i \sim U[0,\ep]$; where $0 < \ep < \min\{|m(g_iI(\theta_P, Z))-m(g_jI(\theta_P, Z))|, (i,j) \in \mathcal{S}\}$ for $\mathcal{S} = \{(i,j) \in [K] \times [K]:m(g_iI(\theta_P, Z)) \neq m(g_jI(\theta_P, Z))\} \neq \emptyset$ and $\ep=1$ for $\mathcal{S} = \emptyset$. 
Then $\tilde{m}_i$, $i\in[K]$ can be viewed as test statistics for which ties happen with zero probability. 
Denoting the new ranks
 \begin{align*}\tilde{R_i}
    = \sum_{u=1}^K I
[ \tilde{m}_i(g_iI(\theta_P,Z)) \ge \tilde{m}_u(g_uI(\theta_P,Z)) ] + 1,
 \end{align*}
we then have $P(\tilde{R_i} = k) = 1/K$ for all $i \in [K]$, by the argument above. 

Now, if $\mathcal{S} = \emptyset$,
we clearly have
\begin{align}\label{rt}
  \sum_{v=1}^K I
[ m(g_jI(\theta_P,Z)) \ge m(g_vI(\theta_P,Z)) ] \ge  \sum_{v=1}^K I
[ \tilde{m}_j(g_jI(\theta_P,Z)) \ge \tilde{m}_v(g_vI(\theta_P,Z)) ]. 
\end{align}
When $\mathcal{S} \neq \emptyset$,
note that $0 < \ep < \min\{|m(g_iI(\theta_P, Z))-m(g_jI(\theta_P, Z))|, (i,j) \in \mathcal{S}\}$,
so that
if $\tilde{m}_u(g_uI(\theta_P,Z)) =\tilde{m}_v(g_vI(\theta_P,Z)))$, we must have $m(g_uI(\theta_P,Z)) = m(g_vI(\theta_P,Z))$ since $|\xi_u-\xi_v|<\ep < \min\{|m(g_iI(\theta_P, Z))-m(g_jI(\theta_P, Z))|,$ $(i,j) \in \mathcal{S}\}$.

Moreover, if $\tilde{m}_j(g_jI(\theta_P,Z)) > \tilde{m}_v(g_vI(\theta_P,Z))$, we have $m(g_jI(\theta_P,Z)) > m(g_vI(\theta_P,Z))$; and if  $\tilde{m}_j(g_jI(\theta_P,Z)) < \tilde{m}_v(g_vI(\theta_P,Z))$, we have $m(g_jI(\theta_P,Z)) < m(g_vI(\theta_P,Z))$. 
Hence, it is shown that
if $\tilde{m}_j(g_jI(\theta_P,Z)) \ge \tilde{m}_v(g_vI(\theta_P,Z))$, then $ m(g_jI(\theta_P,Z)) \ge m(g_vI(\theta_P,Z))$ for all $v \in [K]$; and \eqref{rt} follows.
Hence, we can derive
\begin{align*}
    &  P_Z(R_j \ge k)
    = P_Z\left(\sum_{v=1}^K I
[ m(g_jI(\theta_P,Z)) \ge m(g_vI(\theta_P,Z)) ] \ge k-1\right) \\
    &\ge P_Z\left(\sum_{v=1}^K I
[ \tilde{m}_j(g_jI(\theta_P,Z)) \ge \tilde{m}_v(g_vI(\theta_P,Z)) ] \ge k-1\right) 
= P_Z(\tilde{R}_j \ge k) = (K-k+1)/K.
\end{align*}
Hence, we obtain $P(R_i \ge k) \ge (K-k+1)/K$ for all $k \in [K]$. Considering $k=\lfloor K\alpha\rfloor$, we again conclude \eqref{fe}.

For an infinite group, we have $\alpha' = \alpha$. Thus, conditioning on each sub-sigma-algebra $B_{O(I)}$, we obtain
$
P\bigg(m(I(\theta, z)) \ge q_{\alpha'}
\big(  m(U_{O_I(\theta_P,z)}) \big)\bigg) \ge 1 - \alpha
$
directly from the definition of quantiles.
As this holds almost surely with respect to $I$, we have
\begin{align*}
        \bE\bE\left[ 1\left\{m(I) \ge q_{\alpha}\big(  m(U_{O_I(\theta_P,z)}) \big)\right\} | B_{O(I)}\right] \ge 1 - \alpha,
    \end{align*} 
which finishes the proof.

\subsubsection{Proof of Theorem \ref{sample-JCR}}\label{proof:JCR-infinite-group}

We write $I = I(\theta_P,Z)$ 
for simplicity. 
For the uniform measure $U$ on $\mG$, 
random variables $G_{1:K} \sim U^K$, 
and $I(\theta_O,Z)$ with $Z \sim P$, we have:
    \begin{lemma}
        The vector $A = (I,G_1I,\ldots,G_KI)$ has exchangeable entries.
    \end{lemma}
    \begin{proof}
    Consider $B = (GI,G_1I,\ldots,G_KI)$, where $G\sim U$ is
    independent of $G_{1:K}$ and  $Z$. 
    Denoting $I' = GI =_d I$, $G_i' = G_iG^{-1}$ for $i \in [K]$, 
    for any subsets $\mG_1,\ldots,\mG_K$ of $\mG$, we have
        \begin{align*}
            &U(G_1' \in \mG_1,\ldots, G_K' \in \mG_K) =U(G_1G^{-1} \in \mG_1, \ldots, G_KG^{-1} \in \mG_K) \\
            &= U(G_1 \in G\mG_1,\ldots, G_K \in G\mG_K).
        \end{align*}
    Due to the independence of the entries of $G_{1:K}$, we have
    for any $B \in B_\mG$  such that $P(G\in B)>0$,
        \begin{align*}
            &U(G_1 \in G\mG_1,\ldots, G_K \in G\mG_K|G\in B) =
             \prod_{j=1}^K U(G_j \in G\mG_j|G\in B)\\
            &= \prod_{j=1}^K U(G_j \in \mG_j|G\in B)
            = \prod_{j=1}^K U(G_j \in \mG_j),
        \end{align*}
        where the second step follows from the left-invariance of the Haar measure $U$ on $\mG$.
        Thus,  $(G_{1:K}')$ and $(G_{1:K})$ have identical distributions, 
        and therefore so do $A$ and $B$. 
        Since $G \sim U$ is
    independent of $G_{1:K}$ and  $Z$,
    the entries of $B$ are exchangeable; and the same follows for $A$, finishing the proof.
    \end{proof}
    
    Thus, since the entries of $A$ are exchangeable, $m(I),m(G_1I),\ldots,m(G_KI)$ are exchangeable random variables.
    Then, 
    the result follows 
    from 
    standard results on order statistics,
    see e.g., Chapter 11 in \cite{vovk2022algorithmic}, finishing our proof.
    

\subsection{Supplemental Figures and Simulation Details}\label{sec:diab-supp}
\begin{figure}[ht]
    \centering
    \includegraphics[height = 5cm, width =6cm]{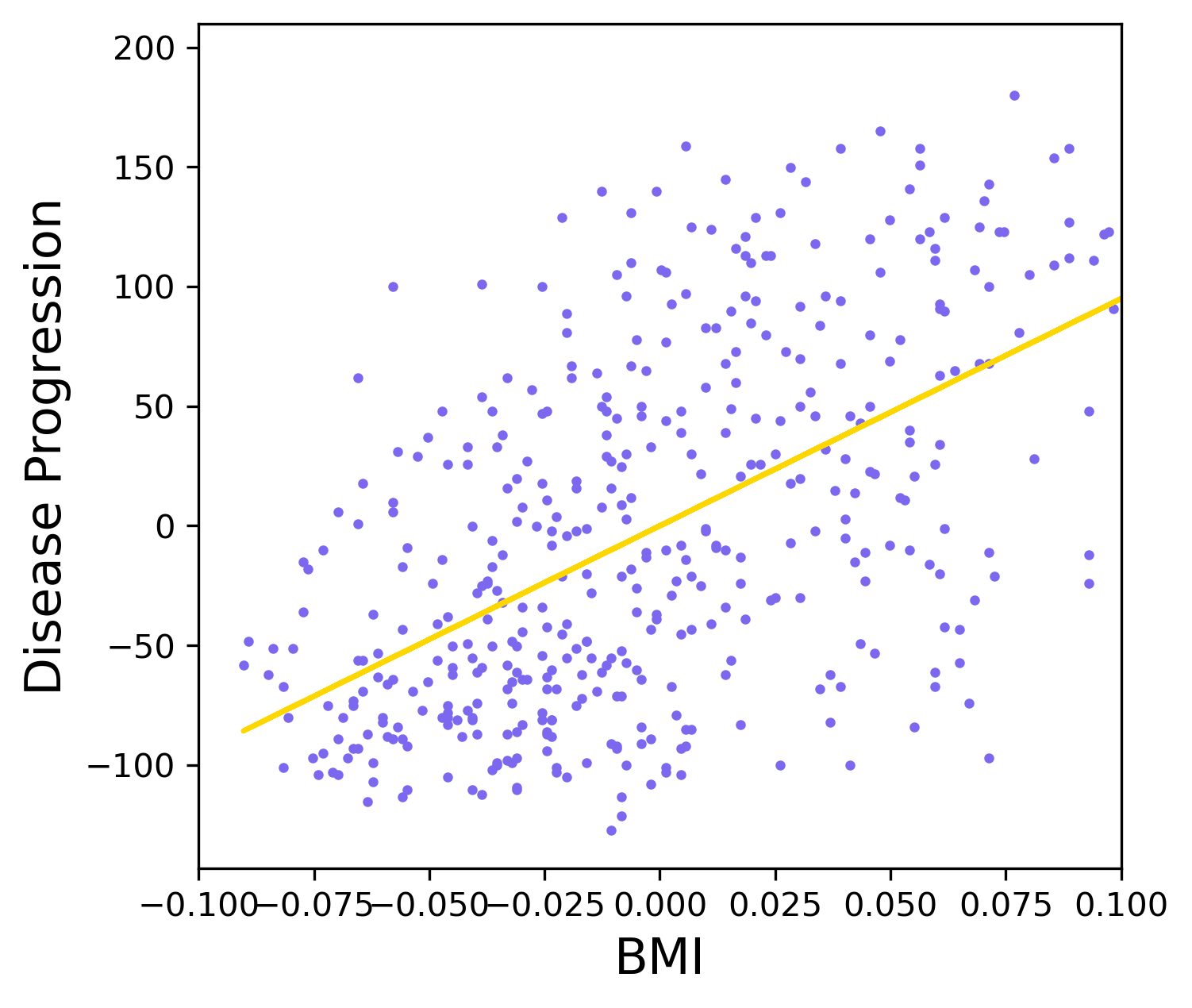}
  \caption{
  The scatterplot of the outcome and BMI for all 442 datapoints in the Diabetes dataset, with least squares line (Section \ref{diab}).}
  \label{fig:diabete-JCR-supp}
\end{figure}

We show the scatterplot of the outcome and BMI for all 442 datapoints (Section \ref{sec:comparison-lr}) in the Diabetes dataset in Figure \ref{fig:diabete-JCR-supp}.

\begin{figure}
    \centering
    \includegraphics[height = 7cm, width = 6.836cm]{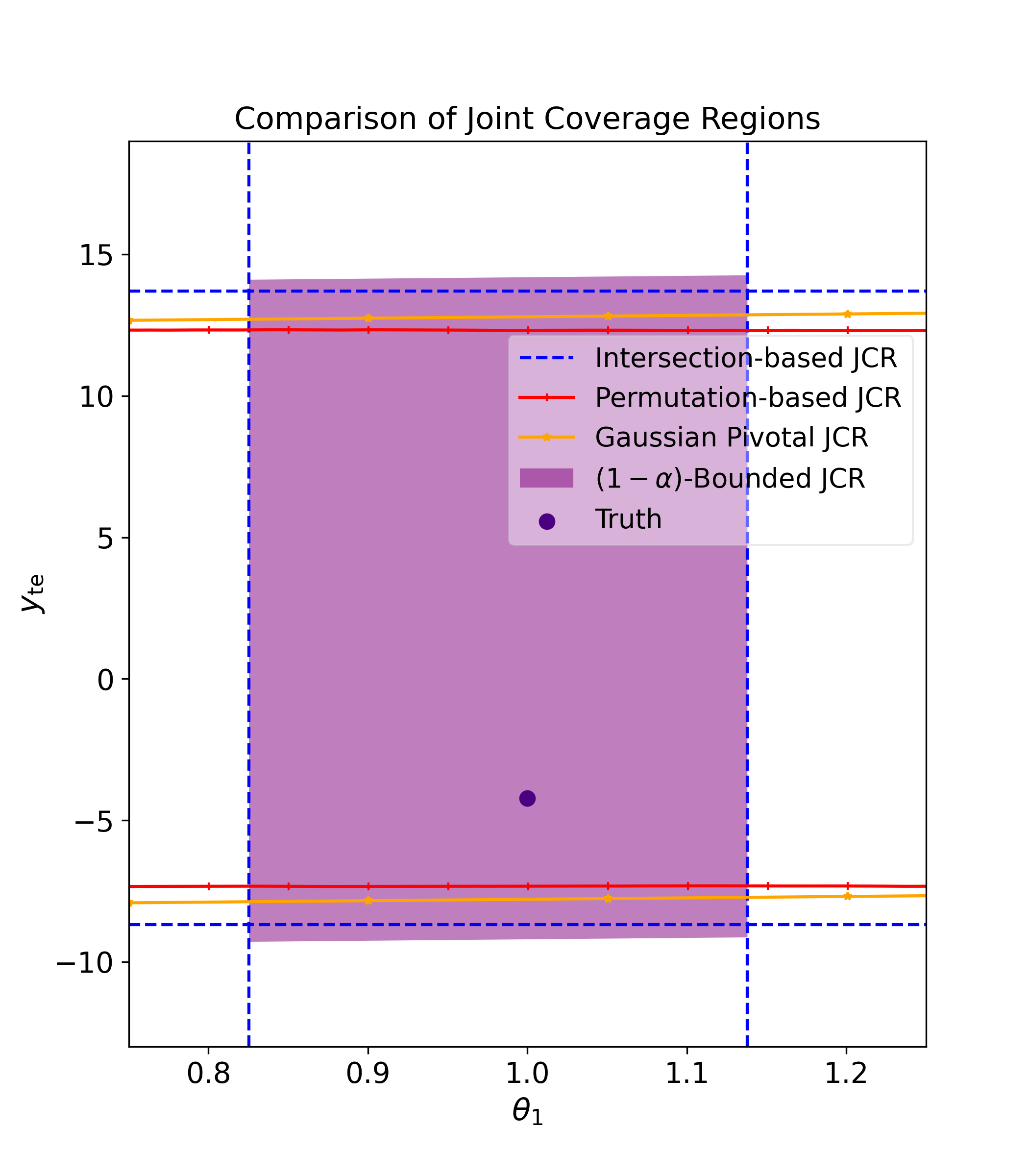}
\caption{A comparison of joint coverage regions with $1 - \alpha$ coverage, as presented in
Section \ref{sec:comparison-lr} (multi-dimensional case), while choosing  $\xte = (0.5, 0.5, 0.5, 0.5, 0.5)$ instead.}
\label{fig:comparison_highd_id}
\end{figure}

In addition,
following the multi-dimensional setting as shown in Section \ref{sec:comparison-lr}, we elaborate our construction of different joint coverage regions as below.

\begin{itemize}
    \item \textbf{Intersection-based joint coverage region.} We denote the OLS estimator $\hat{\theta} = (X^\top X)^{-1}X^\top y$, with covariance matrix
     $\sigma^2(X^\top X)^{-1}$. 
     For the component $\theta_1$ of $\theta$, we have the confidence interval
    $
     C_{\alpha} =    \hat{\theta}_1+
     \text{SE}(\hat{\theta}_1)\cdot [t_{n-p,\alpha/2},t_{n-p,1-\alpha/2} ]
    $
    where $\text{SE}(\hat{\theta}_1) = \sqrt{\hat{\sigma}^2[(X^\top X)^{-1}]_{11}}$ and $\hat{\sigma}^2 = \|y - X \hat{\theta}\|_2^2/(n - p)$.
    We also have the prediction interval
    $$
        T_{\alpha} = \xte^\top \hat{\theta} 
        + \sqrt{\hat{\sigma}^2\left(1 + \xte^\top (X^\top X)^{-1} \xte \right)} \cdot [t_{n-p,\alpha/2},t_{n-p,1-\alpha/2} ].
    $$
    We denote by $C_{\alpha/2} \times T_{\alpha/2}$ the intersection-based joint coverage region.

    \item \textbf{Gaussian pivotal joint coverage region.} Since $(\yte - \xte^\top \theta)/\hat{\sigma} \sim t_{n-p}$ is a pivot, 
    we have
    $$
     J_\alpha =    \left\{(\theta_1,\yte): \exists \ \theta_2,\ldots,\theta_n \text{ s.t. }(\yte - \xte^\top \theta)/\hat{\sigma} \in [t_{n-p,\alpha/2}, t_{n-p, 1-\alpha/2}] \right\}.
    $$

    \item \textbf{Permutation-based joint coverage region.} 
    Given that $y_i - x_i^\top \theta, i \in [n] \cup \{\text{te}\}$ are exchangeable, we may consider a permutation-based construction as in Section \ref{sec:comparison-lr}, while replacing the joint coverage region by its projection. Specifically, if we consider the 
    group of cyclic-shift permutations, 
    the joint coverage region based on the residuals is
     $$  J_\alpha =  \left\{(\theta_1,\yte): \exists \ \theta_2,\ldots,\theta_n \text{ s.t. }|\yte - \xte^\top \theta| \leq q_{1-\alpha}\left(|y_i - x_i^\top \theta|,\ i \in [n] \right) \right\}.
    $$

    \item \textbf{$(1-\alpha)$-Bounded joint coverage region}: We may use the intersection $J_{\alpha/2} \cap C_{\alpha/2}$ to obtain a bounded joint coverage region.
    
\end{itemize}

We also show the case with a different test point $\xte$. We take $\xte = (0.5, 0.5, 0.5, 0.5, 0.5)$ here, while keeping all the other conditions the same as in Section \ref{sec:comparison-lr}. It is shown that the group-invariance-based construction can be more conservative (after projection to a single coordinate) compared to the naive intersection; while its theoretical validity always holds.


{\small
\setlength{\bibsep}{0.2pt plus 0.3ex}
\bibliographystyle{plainnat-abbrev}

}



\end{document}